\documentclass[runningheads]{CSML}
\pdfoutput=1

\usepackage{lastpage}
\newcommand{\dOi}{13(2:9)2017}
\lmcsheading{}{1--\pageref{LastPage}}{}{}%
{Oct.~21, 2016}{Jun.~20, 2017}{}
\usepackage{amsthm}
\usepackage{xcolor}
\usepackage{amsmath,amssymb}
\usepackage{mathpartir}
\usepackage{stmaryrd}
\usepackage{etoolbox}
\usepackage{empheq,soul}
\usepackage{upgreek}
\usepackage{hyperref}
\hypersetup{hidelinks}
\usepackage{multicol}
\usepackage{hyperref}
\hypersetup{hidelinks}
\usepackage[capitalise]{cleveref}

\usepackage{mathtools}
\numberwithin{equation}{section}

\setstcolor{red}

\usepackage{bussproofs}

\usepackage[nomargin,inline,multiuser,author=,draft]{fixme} %% remove the 'draft' option to hide all comments/notes/fixmes
\fxsetup{theme=color,mode=multiuser}

\usepackage[T1]{fontenc}
\usepackage[utf8]{inputenc}

\usepackage{hyperref}
\hypersetup{hidelinks}

\FXRegisterAuthor{eM}{em}{{\color{orange} \underline{eM says} }}
\newcommand{\sid}{\textsf{SID}}

\def \bchpaula {\begin{color}{red}} 
\def \echpaula {\end{color}}

\def \bchm {\begin{color}{blue}} 
\def \echm {\end{color}}

\def \bcom {\begin{color}{blue}Mariangiola comments: } 
\def \ecom {\end{color}}

\def \bcompaula {\begin{color}{red} Paula comments: } 
\def \ecompaula {\end{color}}
\newif\ifproofs
\proofstrue

\newif\ifcomments
\commentstrue
\commentsfalse

\newif\iflong
\longtrue
\definecolor{mymagenta}{rgb}{0.5,0,0.5}
\definecolor{myred}{rgb}{0.5,0,0}
\definecolor{mygreen}{rgb}{0,0.4,0}
\definecolor{myblue}{rgb}{0,0,0.6}

\newcommand{\cf}{\emph{cf.}~}

\newcommand{\set}[1]{\{#1\}}
\newcommand{\natset}{\mathbb{N}}

\newcommand{\mkkeyword}[1]{\mathtt{\color{myblue}#1}}
\newcommand{\mkconstant}[1]{\mathtt{\color{mymagenta}#1}}

\newenvironment{lines}[1][t]{
  \begin{array}[#1]{@{}l@{}}
}{
  \end{array}
}

\newcommand{\ttunderscore}{\texttt{\char`_}}

\newcommand{\ttcomma}{\texttt{\upshape,}}

\newcommand{\ttpipe}{\texttt{\upshape|}}
\newcommand{\ttqmark}{\texttt{\upshape?}}
\newcommand{\ttemark}{\texttt{\upshape!}}

\newcommand{\defrule}[1]{%
  \hypertarget{rule:{#1}}{%
    \text{\scriptsize[\textsc{#1}]}%
  }%
}

\newcommand{\refrule}[1]{%
  \hyperlink{rule:{#1}}{%
    \text{\upshape\scriptsize[\textsc{#1}]}%
  }%
}

\newcommand{\introarrow}{$\arrow$I}
\newcommand{\introlarrow}{$\larrow$I}
\newcommand{\introbullet}{$\bullet$I}
\newcommand{\introprod}{$\times$I}
\newcommand{\fixrule}{fix}
\newcommand{\bindrule}{bind}
\newcommand{\futurerule}{future}

\newcommand{\elimarrow}{$\arrow$E}
\newcommand{\elimlarrow}{$\larrow$E}
\newcommand{\elimprod}{$\times$E}
\newcommand{\axiom}{axiom}
\newcommand{\const}{const}

\newcommand{\newrule}{new}

\newcommand{\Constant}{\mkconstant{k}}

\newcommand{\var}{\varX}
\newcommand{\varX}{{x}}
\newcommand{\varY}{{y}}
\newcommand{\varZ}{{z}}

\newcommand{\varW}{{w}}

\newcommand{\Polarity}{\PolarityP}
\newcommand{\PolarityP}{p}
\newcommand{\PolarityQ}{q}

\newcommand{\Channel}{\ChannelA}
\newcommand{\ChannelA}{{a}}
\newcommand{\ChannelB}{{b}}
\newcommand{\ChannelC}{{c}}

\newcommand{\Name}{\NameU}
\newcommand{\NameU}{u}

\newcommand{\Thing}{X}
\newcommand{\ThingY}{Y}
\newcommand{\BindableName}{\Thing}
\newcommand{\bindname}{\BindableName}

\newcommand{\BindableNameY}{Y}

\newcommand{\Expression}{\ExpressionE}
\newcommand{\ExpressionE}{e}
\newcommand{\ExpressionF}{f}
\newcommand{\ExpressionG}{g}

\newcommand{\Process}{\ProcessP}
\newcommand{\ProcessP}{P}
\newcommand{\ProcessQ}{Q}
\newcommand{\ProcessR}{R}

\newcommand{\Type}{\TypeT}
\newcommand{\TypeT}{t}
\newcommand{\TypeS}{s}

\newcommand{\SessionType}{\SessionTypeT}
\newcommand{\SessionTypeT}{T}
\newcommand{\SessionTypeS}{S}

\newcommand{\Context}{\mathcal{E}}
\newcommand{\PContext}{\mathcal{C}}
\newcommand{\Hole}{[~]}

\newcommand{\unit}{\mkconstant{unit}}

\newcommand{\pair}{\mkconstant{pair}}

\newcommand{\send}{\mkconstant{send}}
\newcommand{\receive}{\mkconstant{recv}}
\newcommand{\fix}{\mkconstant{fix}}

\newcommand{\open}{\mkconstant{open}}

\newcommand{\return}{\mkconstant{return}}
\newcommand{\bind}{\mkconstant{bind}}
\newcommand{\future}{\mkconstant{future}}

\newcommand{\emptyprocess}{\mkconstant{0}}

\newcommand{\Unit}{\mkbasictype{Unit}}

\newcommand{\splitkw}{\mkkeyword{split}}

\newcommand{\inkw}{\mkkeyword{in}}
\newcommand{\askw}{\mkkeyword{as}}

\newcommand{\Fun}[1]{\lambda#1.}

\newcommand{\Split}[4]{\splitkw~#1~\askw~#2,#3~\inkw~#4}

\newcommand{\Pair}[2]{\langle #1\ttcomma#2\rangle}
\newcommand{\Bind}[2]{#1 ~\texttt{>}\!\texttt{>=}~ #2}
\newcommand{\Return}[1]{\return~#1}
\newcommand{\Open}[1]{\open~#1}
\newcommand{\Fix}[2]{\fix~\Fun#1#2}
\newcommand{\Receive}[1]{\receive~#1}
\newcommand{\Send}[2]{\send~#1~#2}
\newcommand{\Future}[1]{\future~#1}

\newcommand{\streamf}{\mkconstant{from}}

\newcommand{\upkf}{\mkconstant{stream}}
\newcommand{\bupkf}{\mkconstant{stream_{\mkconstant{0}}}}
\newcommand{\dspf}{\mkconstant{display}}
\newcommand{\display}{\mkconstant{display}}
\newcommand{\bdisplay}{\display_{\mkconstant{0}}}
\newcommand{\mapf}{\mkconstant{incStream}}
\newcommand{\incf}{\mkconstant{inc}}

\newcommand{\skipeven}{\mkconstant{skip}}

\newcommand{\thread}[2]{#1 \Leftarrow #2}
\newcommand{\server}[2]{\mkkeyword{server}~#1~#2}

\newcommand{\parop}{\mathbin\ttpipe}
\newcommand{\new}[1]{(\nu#1)}
\newcommand{\mkbasictype}[1]{\mkkeyword{#1}}
\newcommand{\tbasic}{\mkbasictype{B}}

\newcommand{\arrow}{\to}
\newcommand{\larrow}{\multimap}

\newcommand{\tio}[1]{\mkbasictype{IO}~#1}

\newcommand{\tbullet}[1][]{%
  \bullet%
  \ifblank{#1}{}{^{#1}}%
}

\newcommand{\tshared}[1]{\langle#1\rangle}

\newcommand{\End}{\mkkeyword{end}}
\newcommand{\tin}[2]{{\ttqmark}#1.#2}
\newcommand{\tout}[2]{{\ttemark}#1.#2}

\newcommand{\replace}[2]{ [\,#2\, / \, #1\,]}

\newcommand{\repReceive}{\mathsf{r}}
\newcommand{\repSend}{\mathsf{s}}

\newcommand{\repsend}{[\mathsf{s}~\Channel^{\Polarity}~\Expression~\var]}
\newcommand{\repreceive} {[\mathsf{r}~\Channel^{\Polarity}~\Expression~\var]}
\newcommand{\repreceivec} {[\mathsf{r}~\Channel^{\co\Polarity}~\Expression~\varY]}

\newcommand{\subst}[2]{\{#1/#2\}}
\newcommand{\sust}[2]{\subst{#2}{#1}}

\newcommand{\dom}{\mathsf{dom}}

\newcommand{\fn}{\mathsf{fn}}
\newcommand{\fv}{\mathsf{fv}}

\newcommand{\typeof}{\mathsf{types}}

\newcommand{\co}[1]{\overline{#1}}

\newcommand{\eqdef}{= }% {\stackrel{\text{\tiny def}}{=}}

\newcommand{\wte}[3]{#1 \vdash #2 : #3}

\newcommand{\wtp}[3]{#1 \vdash #2 ~\triangleright~ #3}

\newcommand{\red}{\longrightarrow}
\newenvironment{framedmath}{%
  \begin{math}%
    \displaystyle%
}{%
  \end{math}%
}

\newcommand{\linpred}{\mathsf{lin}}
\newcommand{\linq}[1]{\linpred(#1)}
\newcommand{\unq}[1]{\mathsf{un}(#1)}

\newcommand{\sharedq}[1]{\mathsf{shared}(#1)}

\newcommand{\TypeContext}{\Upgamma}

\newcommand{\ProvideContext}{\Updelta}
\newcommand{\ptc}[2]{#1+#2}

\newcommand{\subprocess}{\subset}
\newcommand{\subprocesseq}{\subseteq}
\newcommand{\setSubProc}{{\mathcal S}}

\newcommand{\afterred}{\ \pmb{\succ} \ }
\newcommand{\prered}{\ \pmb{\prec} \ }

\newcommand\ti[1]{{\llbracket{#1}\rrbracket}}
\newcommand\indti[2]{{\llbracket{#1}\rrbracket}_{#2}}
\newcommand{\typerank}{rank}

\newcommand{\ind}{i}
\newcommand{\indj}{j}
\newcommand{\funsubst}{\delta}
\newcommand{\EE}{\mathfrak{E}} 

\newcommand{\setWN}{ \mathfrak{ N} }
\newcommand{\setNVAR}{\HH_{v} }%{\HH_{x} }
\newcommand{\setNIO}{\HH_{IO}}
\newcommand{\HH}{\setWN}
\newcommand{\WNVAR}{\setNVAR }
\newcommand{\HHPC}{\setNIO}

\newcommand{\modelsi}{\models_{\ind}}
\newcommand{\modelsj}{\models_{\indj}}

\newcommand{\omegafuture}{\omegaterm_{\future}}

\newcommand{\infinite}{\mkbasictype{{\bullet^\infty}}}

\newcommand{\ListN}{\mkbasictype{S}_{\nat}}
\newcommand{\ListNtwo}{\mkbasictype{S2}_{\nat}}

\newcommand{\badListN}{\mkbasictype{S}'_{\nat}}
\newcommand{\OutN}{\mkbasictype{Out}_{\nat}}
\newcommand{\InN}{\mkbasictype{In}_{\nat}}
\newcommand{\badOutN}{\mkbasictype{Out}'_{\nat}}
\newcommand{\badInN}{\mkbasictype{In}'_{\nat}}
\newcommand{\nat}{\mkbasictype{Nat}}

\newcommand{\omegaterm}{\mkconstant{\Upomega}}

\newcommand{\loopone}{\mkconstant{loop1}}
\newcommand{\looptwo}{\mkconstant{loop2}}
\newcommand{\sessiontypeloop}[1]{{\sf RS}_#1}
\newcommand{\sessiontypelooptwo}[1]{{\sf SR}_#1}

 \newtheorem{definition}{Definition}[section]
 \newtheorem{example}[definition]{Example}

\theoremstyle{plain}
 
 \newtheorem{lemma}[definition]{Lemma}
 \newtheorem{theorem}[definition]{Theorem}
 \newtheorem{corollary}[definition]{Corollary}

\newcommand{\CAA}{\mathcal{A}}
\newcommand{\CBB}{\mathcal{B}}
\newcommand{\CA}{}%{\mathcal{A}}
\newcommand{\CB}{}%{\mathcal{B}}
\newcommand{\CC}{}%{\mathcal{D}}

\newcommand{\Pol}{\mathcal{N}}

\newcommand{\indep}{ \  \#  \ }
\newcommand{\ri}[2]{{#2} of {#1}}

\newcommand{\redminus}{\red^{-}}
\newcommand{\numbred}{{\sf ns}}
\newcommand{\numCOM}{{\sf ct}}
\newcommand{\numCOMi}[2]{{\sf ct} ({#1}, #2) }

\newcommand{\delay}{{\sf delay}}
\newcommand{\threads}{{\sf threads}}
\newcommand{\servers}{{\sf servers}}
\newcommand{\boundnames}{{\sf bounds}}
\newcommand{\weight}{{\sf wt}}
\newcommand{\comma}{ }%{,  }

\newcommand\mycase[1]{ \vspace{0.2cm} \noindent {\em #1}. \vspace{0.2cm}}
 \newcommand\mycaseA[1]{ \vspace{0.3cm} \noindent {\em #1}. \vspace{0.005cm}}
\begin{document}

\title[On Sessions and Infinite Data]{On Sessions and Infinite Data}
\thanks{ All authors have been supported by the ICT COST European
  project \emph{Behavioural Types for Reliable Large-Scale Software
    Systems} (BETTY, COST Action IC1201).  }
  
\author[Severi et al.]{Paula Severi\rsuper a}
\address{\lsuper{a,c} Department of Computer Science, University of Leicester, UK}
\thanks{
  Paula Severi was supported by a Daphne Jackson fellowship
  sponsored by EPSRC and her department.}

\author[]{Luca Padovani\rsuper b}
\address{\lsuper{b,d} Dipartimento di Informatica,
  Universit\`a di Torino, Italy}

\author[]{Emilio Tuosto\rsuper c}
 
\author[]{Mariangiola Dezani-Ciancaglini\rsuper d}
\thanks{
  Mariangiola Dezani was partly supported by EU projects %
  H2020-644235 \emph{Rephrase} and H2020-644298 \emph{HyVar}, %
  ICT COST Actions IC1402 ARVI, IC1405 Reversible Computation, CA1523 EUTYPES and Ateneo/CSP project \emph{RunVar}. 
}

\keywords{Session types, the $\pi$-calculus, Infinite data, Type safety}

\subjclass{{F.1.2 }{[{\bf Computation by Abstract Devices}]: }{Modes of Computation}{---\em Parallelism and concurrency},
{F.3.3 }{[{\bf Logics and Meanings of Programs}]: }{Studies of Program Constructs}{---\em Type structure},
{H.3.5 }{[{\bf Information Storage and Retrieval}]: }{Online Information Services}{---\em Web-based services},
{H.5.3 }{[{\bf Information Interfaces and Presentation}]: }{Group and Organization Interfaces}{---\em Theory and models, Web-based interaction.}}

\begin{abstract}
%    We investigate some subtle issues that arise when programming
%    distributed computations over infinite data structures.
%
%    To do this, we formalise a calculus that combines a call-by-name
%    functional core with session-based communication primitives and
%    that allows session operations to be performed ``on demand''.
%
  We define a novel calculus that combines a call-by-name functional
  core with session-based communication primitives.
    We develop a typing discipline that guarantees both normalisation
    of expressions and progress of processes and that uncovers an
    unexpected interplay between evaluation and
    communication.
\end{abstract}
\maketitle

\section{Introduction}
\label{section:Introduction}

Infinite computations have long lost their negative connotation.
Two paradigmatic contexts in which they appear naturally are reactive
systems~\cite{mp12,Aceto:2007:RSM:1324845} and lazy functional
programming.
The former contemplates the use of infinite computations in order to
capture \emph{non-transformational} computations, that is computations
that cannot be expressed in terms of transformations from inputs to
outputs; rather, computations of reactive systems are naturally
modelled in terms of ongoing interactions with the environment.
Lazy functional programming is acknowledged as a paradigm that fosters
software modularity~\cite{Hughes89} and enables programmers to specify
computations over possibly infinite data structures in elegant and
concise ways.
% Besides featuring very elegant and concise programming styles, lazy
% functional languages offer the possibility of using infinite data
% structures for representing and analysing \lq\lq big data\rq\rq\ in a
% more accurate and faithful way.
Nowadays, the synergy between these two contexts has a wide range of
potential applications, including stream-processing networks,
real-time sensor monitoring, and internet-based media services.

Nonetheless, not all diverging programs -- those engaged in an
infinite sequence of possibly intertwined computations and
communications -- are necessarily useful. There exist degenerate forms
of divergence where programs do not produce results, in terms of
observable data or performed communications.
We investigate this issue by proposing a calculus for expressing
computations over possibly infinite data types and involving message
passing. The calculus -- called \sid, after Sessions with Infinite
Data -- combines a call-by-name functional core (inspired by Haskell)
with multi-threading and session-based communication primitives.

In the remainder of this section we provide an informal introduction
to \sid{} and its key features by means of a few examples.  The formal
definition of the calculus, of the type system, and its properties are
given in the rest of the paper.
A simple instance of computation producing an infinite data structure
is given by
\begin{equation*}%\label{streamify:eq}
  \streamf\ x \ = \Pair x {\streamf\ (x + 1)}
\end{equation*}
%\mmic{that is a function $\streamf$ of type $\TypeS \arrow \stream \TypeS$
%that given an expression $\Expression$ (of type $\TypeS$) produces an
%infinite stream of $\Expression$'s, namely a structure of type
%$\stream \TypeS$ defined as the maximal fix-point of the type equation
%$\stream \TypeS = \TypeS \times \stream \TypeS$.}
%{$\stream \TypeS$ is a pre-type! A $\bullet$ is needed, but I think here it would only complicate the matter}
where the function $\streamf$ applied to a number $n$ produces the
stream (infinite list) \[\Pair{n}{\Pair{n+1}{\Pair{n+2}\cdots}}\] of
integers starting from $n$. We can think of this list as abstracting
the frames of a video stream or the samples taken from a sensor.

The key issue we want to address is how infinite data can be exchanged
between communicating threads. The most straightforward way of doing
this in \sid{} is to take advantage of lazy evaluation.
For instance, the \sid{} process
\begin{equation*}%\label{thread:eq}
  \thread \varX {\Bind {\big( \send\ c^{+}\ (\streamf\ 0)\big)} \ExpressionF}
  \quad \parop \quad
  \thread \varY {\Bind{ \receive\ c^{-}} \ExpressionG}
\end{equation*}
represents two threads $\varX$ and $\varY$ running in parallel and
connected by a session $\ChannelC$, of which thread $\varX$ owns one
endpoint $\ChannelC^+$ and thread $\varY$ the corresponding peer
$\ChannelC^-$.
Thread $\varX$ sends a stream of natural numbers on $\ChannelC^+$ and
continues as $\ExpressionF~\ChannelC^+$, where $\ExpressionF$ is left
unspecified.  Thread $\varY$ receives the stream from $\ChannelC^-$
and continues as $(\ExpressionG~\Pair{\streamf\ 0}{\ChannelC^-})$.
The \emph{bind} operator $\Bind\ttunderscore\ttunderscore$ models
sequential composition and has the 
same semantics as in
Haskell, i.e.
it passes the result of performing the left %first 
action to
the (parametrised) right %second 
action.  
% In particular, 
%it applies the rhs to the result of the action on its lhs. 
The result of sending a message on the endpoint
$\Channel^+$ is the endpoint itself, while the result of receiving a
message from the endpoint $\Channel^-$ is a pair consisting of the
message and the endpoint.
In this example, the \emph{whole stream} is sent \emph{at once} in a
single interaction between $\varX$ and $\varY$.  This behaviour is
made possible by the fact that \sid\ evaluates expressions
\emph{lazily}: the message $(\streamf\ 0)$ is not evaluated until it
is used by the receiver.

% in~\cref{thread:eq} sequentially composes impure
% expressions~\cite{JonesW93}, as in functional languages such as
% Haskell. As shown in \cref{section:typing}, this requires
% expressions $\ExpressionF$ and $\ExpressionG$ in \cref{thread:eq} to
% be functions mapping their argument into impure computations.

In principle, exchanging \lq\lq infinite\rq\rq\ messages such as
$(\streamf\ 0)$ between different threads is no big deal. In the real
world, though, this interaction poses non-trivial challenges: the
message consists in fact of a mixture of data (the parts of the
messages that have already been evaluated, like the constant $0$) and
code (which lazily computes the remaining parts when necessary, like
$\streamf$).
This observation suggests an alternative, more viable modelling of this
interaction whereby the sender unpacks the stream element-wise, sends
each element of the stream as a separate message, and the receiver
gradually reconstructs the stream as each element arrives at
destination.
This modelling is intuitively simpler to realise (especially in a
distributed setting) because the messages exchanged at each
communication are basic values rather than a mixture of data and
code.
In \sid{} we can model this as a process
%\Reviewer2{unpack is a bad choice as it suggests existential types. stream maybe?}
% For a computation that visits a stream element-wise consider the
% thread $\var$ below.  It \lq\lq unpacks\rq\rq\ a stream of zeros
% sending its elements on a channel $\ChannelC$ in parallel with thread
% $\varY$ that just \lq\lq disposes of\rq\rq\ the elements received on
% $\ChannelC$:
\[
  \thread {\mathit{prod}} \bupkf\ \ChannelC^+\ (\streamf\ 0)
  \quad \parop \quad
  \thread {\mathit{cons}} {\bdisplay\ \ChannelC^-}
\]
where the functions $\bupkf$ and $\bdisplay$ are defined as:
%\mmic{\begin{eqnarray*}
%  \upkf & = & \fix\ \Fun f {
%    \Fun \ChannelC {
%      \Fun{{\varX s}}{
%        \Split{\varX s} \varY {\varY s}{
%          \Bind {\Send \ChannelC \varY}
%          {\Fun \ChannelC
%            \Future{(f \ \ChannelC \ \varY s )}
%          }
%        }
%      }
%    }
%  }
%  \\
%  \display & = & \Fun \ChannelC \fix\ \Fun \ExpressionF {
%    \Bind{\Receive \ChannelC}{
%      \Fun \varZ {\ExpressionF\ (\Return \varZ)}
%      }
%    }
%\end{eqnarray*}
%%
%We now comment on $\upkf$ and $\display$ and the new constructs therein.
%%
%The construct $\splitkw$ accesses the elements of pairs (if $\Expression$
%evaluates to a pair $\Pair{\Expression_1}{\Expression_2}$ then
%$\Split \Expression {\var_1} {\var_2} \ExpressionG$ evaluates to
%expression $\ExpressionG$ where $\Expression_1$ and $\Expression_2$
%are respectively replaced for $\var_1$ and $\var_2$). So, $\upkf$ repeatedly splits the stream $\varX s$ in its head $\varY$
%and tail $\varY s$, sends $\varY$ on $\ChannelC$ and recurs on the
%$\varY s$.
%}{}
\begin{equation}
\label{equation:upkf}
\begin{array}{lll}
 \bupkf~\varY \ \Pair\varX{\varX s} & = &
    \Bind{\Send\varY\varX}{
      \Fun{\varY'}{\bupkf~\varY'~\varX s}
    }
%  \upkf~\varY & = &
%  \Fun{\Pair\varX{\varX s}}{
%    \Bind{\Send\varY\varX}{
%      \Fun{\varY'}{\upkf~\varY'~\varX s}
%    }
%  }
  \\
    \bdisplay~\varY & = &
    \Bind{\Receive\varY}{
      \Fun{\Pair\varZ{\varY'}}{
        \Bind{\bdisplay~\varY'}{
          \Fun{\varZ s}{\ExpressionG~\Pair{\varZ}{\varZ s}}
        }
      }
    }
\end{array}
\end{equation}

The syntax $\Fun{\Pair\ttunderscore\ttunderscore}\Expression$ is just
syntactic sugar for a function that performs pattern matching on the
argument, which must be a pair, in order to access its components. In
$\bupkf$, pattern matching is used for accessing and sending each
element of the stream separately. In $\bdisplay$, the pair
$\Pair\varZ{\varY'}$ contains the received head $\varZ$ of the stream
along with the continuation $\varY'$ of the session endpoint from
which the element has been received. The recursive call $\bdisplay~\varY'$
retrieves the tail of the stream $\varZ s$, which is then combined
with the head $\varZ$ and passed as an argument to $\ExpressionG$.

The code of $\bdisplay$ looks reasonable at first, but conceals a
subtle and catastrophic pitfall: the recursive call $\bdisplay~\varY'$
is in charge of receiving the \emph{whole} tail $\varZ s$, which is an
infinite stream itself, and therefore it involves an infinite number
of synchronisations with the producing thread! This means that
$\bdisplay$ will hopelessly diverge striving to receive the whole
stream before releasing control to $\ExpressionG$.
This is a known problem which has led to the development of
primitives (such as \texttt{unsafeInterleaveIO} in
Haskell or \texttt{delayIO} in~\cite{JonesW93}) that allow the
execution of I/O actions to interleave with their continuation. In
this paper, we call such primitive $\future$, since its semantics is
also akin to that of \emph{future
  variables}~\cite{SabelS11}. Intuitively, an expression 
$\Bind{\future~\ExpressionE}{\Fun\var{\ExpressionF~\var}}$ 
allows 
to evaluate 
$\ExpressionF~\var$  
even if $\ExpressionE$,
 which typically
involves I/O, has not been completely performed. The variable $\var$
acts as a placeholder for the result of $\ExpressionE$; if
$\ExpressionF$ needs to inspect the structure of $\var$, its
evaluation is suspended until $\ExpressionE$ produces enough data.
%\bdelpaula For similar reasons $\bupkf$ cannot be typed. \edelpaula
%
Using $\future$ we can amend the definitions of $\bupkf$ and $\bdisplay$  thus
\begin{equation}
  \label{eq:dspgood}
  \begin{array}{lll}
   \upkf~\varY \ \Pair\varX{\varX s} & = &
    \Bind{\Send\varY\varX}{
      \Fun{\varY'}{\future~(\upkf~\varY'~\varX s)}
    }\\
  \display~\varY &=&
  \Bind{\Receive\varY}{
    \Fun{\Pair{\varZ}{\varY'}}{
      \Bind{\future~(\display~\varY')}{
        \Fun{\varZ s}{\ExpressionG~\Pair{\varZ}{\varZ s}}
      }
    }
  }
  \end{array}
\end{equation}
where $\display$
allows $\ExpressionG$ to start processing the
stream as its elements come through the connection with the producer
thread.
The type system that we develop in this paper allows us to reason on
sessions involving the exchange of infinite data and when such
exchanges can be done ``productively''. In particular, our type system
flags $\bupkf$ and $\bdisplay$ in~\eqref{equation:upkf} as ill-typed, while it
accepts $\upkf$ and $\display$ in~\eqref{eq:dspgood} as well-typed. To do so, the
type system uses a modal operator $\bullet$ 
  which guarantees that the number of communications
is finite if the number of generated threads is finite. 
% related to the normalisability of expressions. 
As hinted by the
examples~\eqref{equation:upkf} and~\eqref{eq:dspgood}, this operator
plays a major role in the type of $\future$.

We remark that \sid \ does not force exchanged messages to be basic, nor does it prevent exchanging infinite streams in one shot. The purpose of SID is to enable the modelling of systems where communications and infinite data structures are intertwined and to study a typing discipline that guarantees the preservation of productivity in this setting.

\subsection*{Contributions and Outline.}
The \sid\ calculus, defined in \cref{section:language}, combines in an
original way standard constructs from the $\lambda$-calculus and
process algebras  with 
session types  
in the spirit of~\cite{DBLP:conf/esop/HondaVK98,GayV10}.
%
%Our type system is given in~\cref{section:typing}.
The type system, given in~\cref{section:typing}, has the
  novelty of using the modal operator $\tbullet$ to control the
  recursion of programs that perform communications.  To the best of
  our knowledge, the interplay between $\tbullet$ and the type of
  $\future$ is investigated here for the first time.  
The properties of our framework, presented in~\cref{sec:propexp} and~\cref{sec:propproc},
include subject reduction~(\cref{theorem:subjectreduction} and~\cref{theorem:subjectreductionprocesses}),
normalisation of expressions~(\cref{theorem:weakheadnormalization}),
progress and confluence of
processes~(Theorems~\ref{theorem:reducestoreturnI},~\ref{thm:confluence}). Sections~\ref{section:relatedwork} and~\ref{section:conclusions} discuss related and future work, respectively. 
 Appendixes contain the proofs of three  
 theorems. 

\subsection*{Publication History.}
This paper is a thoroughly revised and extended version
of~\cite{SPTD16} and its companion technical
report~\cite{SeveriPadovaniTuostoDezani16TR}.  There are three
substantial improvements compared to previous versions of the paper.
First, we give a much simplified definition of well-polarisation
(\cref{definition:wpp}) resulting in simpler and cleaner proofs.
Second, we have strengthened the progress theorem
(\cref{theorem:reducestoreturnI}) and as a consequence part of its
proof is new.
Finally, the strong normalisation of the reduction without rules
\refrule{r-open} and \refrule{r-future} (\cref{theorem:snrm}) appears
here for the first time.

%%% Local Variables:
%%% mode: latex
%%% TeX-master: "LazySessions"
%%% End:

\section{The  \sid\ Calculus}
\label{section:language}

We use an infinite set of \emph{channels} $\ChannelA$, $\ChannelB$,
$\ChannelC$ and a disjoint, infinite set of \emph{variables} $\varX$,
$\varY$.
We distinguish between two kinds of channels: \emph{shared channels}
are public service identifiers that can only be used to initiate
sessions; \emph{session channels} represent private sessions on which
the actual communications take place.
We distinguish the two \emph{endpoints} of a session channel
$\ChannelC$ by means of a \emph{polarity}
$\Polarity \in \set{ {+}, {-} }$ and write them as $\ChannelC^+$ and
$\ChannelC^-$.  We write $\co\Polarity$ for the dual polarity of
$\Polarity$, where $\co+ = {-}$ and $\co- = {+}$, and we say that
$\ChannelC^\Polarity$ is the \emph{peer endpoint} of
$\ChannelC^{\co\Polarity}$.
A \emph{bindable name} $\bindname$ is either a channel or a variable
and a \emph{name} $\Name$ is either a bindable name or an
endpoint.

\begin{table}[b]
\caption{Syntax of expressions and processes.}
\label{table:syntax}
\framebox[\textwidth]{
\begin{framedmath}
\begin{array}{l}
\begin{array}{@{}c@{\quad}c@{}}
  \begin{array}[t]{r@{~~}c@{~~}l@{\quad}l}
    \Expression & ::= & & \textbf{Expression} \\
    &   & \Constant & \text{(constant)} \\
    & | & \Name        & \text{(name)} \\
    & | & \Fun\var\Expression & \text{(abstraction)} \\
    & | & \Expression\Expression & \text{(application)} \\
    & | & \Split\Expression\varX\varY\Expression & \text{(pair splitting)} \\
    \\
    \Constant & ::= & \rlap{$\unit \mid \pair \mid \open \mid \send \mid \receive \mid \future \mid \return \mid \bind$}
    \\ \\
    \Thing & ::= & 
     \Channel  
     \mid  \var  &   %\textbf{Bindable Name} 
  \end{array}
  &
  \begin{array}[t]{r@{~~}c@{~~}l@{\quad}l}
    \Process & ::= & & \textbf{Process} \\
    &   & \emptyprocess & \text{(idle process)} \\
    & | & \thread\var\Expression & \text{(thread)} \\
    & | & \server\Channel\Expression & \text{(server)} \\
    & | & \Process\parop\Process & \text{(parallel)} \\
    & | & \new\Thing\Process & \text{(restriction)} \\
  \end{array}
\end{array} 
\end{array}
\end{framedmath}
}
\end{table}

The syntax of \emph{expressions} and \emph{processes} is given in
\cref{table:syntax}.
In addition to the usual constructs of the $\lambda$-calculus,
expressions include constants, ranged over by $\Constant$, and pair
splitting. Constants are the unitary value $\unit$, the pair
constructor $\pair$, the primitives for session initiation and
communication $\open$, $\send$, and
$\receive$~\cite{DBLP:conf/esop/HondaVK98,GayV10}, the monadic
operations $\return$ and $\bind$~\cite{JonesW93}, and a primitive
$\future$ to defer computations~\cite{JonesGF96,JonesTutorial2001}.
%
%
%Curry's fixed point combinator need not be primitive since
%  it is typeable.
We do not need a primitive constant for the fixed point operator
because it can be expressed and typed inside the language.
For simplicity, we do not include primitives for branching
and selection typically found in session calculi.
They are straightforward to add and do not invalidate any of the
results.
Expressions are subject to the usual conventions of the
$\lambda$-calculus. In particular, we assume that the bodies of
abstractions extend as much as possible to the right, that
applications associate to the left, and we use parentheses to
disambiguate the notation when necessary.
Following established notation, we write
$\Pair\ExpressionE\ExpressionF$ in place of
$\pair~\ExpressionE~\ExpressionF$, and
$\Fun {\Pair{\varX_1}{\varX_2}} \Expression$ in place of
$\Fun {\varX}\Split{\varX}{\varX_1}{\varX_2}{\Expression}$, and
$\Bind\ExpressionE\ExpressionF$ in place of
$\bind~\ExpressionE~\ExpressionF$.
As usual, we assume that the infix operator $\Bind$ is
right-associative.

A process can be either the idle process $\emptyprocess$ that performs
no actions, a thread $\thread\var\Expression$ with name $\var$ and
body $\Expression$ that evaluates the body and binds the result to
variable $\var$, a $\server\Channel\Expression$ that waits for session
initiations on the shared channel $\Channel$ and spawns a new thread
computing $\Expression$ at each connection, the parallel composition
of processes, and the restriction of a bindable name.
%
% \textcolor{orange}{
% %
%   In \sid, processes do not directly cater for computations and
%   communications, which are in fact specified through expressions.
%   %
%   In other words, processes just yield the machinery to parallelise
%   the evaluation of expressions.
%   %
%   Note \sid\ does not feature recursion explicitly since a fixpoint
%   operator as in the $\lambda$-calculus can be derived (\cf
%   \cref{sec:typing}).
% }
%
In processes, restrictions bind tighter than parallel composition and
we may abbreviate
$\new{\Thing_1}\cdots\new{\Thing_n}\Process$ with
$\new{\Thing_1\cdots\Thing_n}\Process$.

%We have that $\Split\ExpressionF\varX\varY\Expression$ binds both
%$\varX$ and $\varY$ in $\Expression$ and $\new\Channel\Process$ binds
%$\Channel^+$ and $\Channel^-$ within $\Process$ in addition to
%$\Channel$.  
We have that $\Split\ExpressionE\varX\varY\ExpressionF$ binds both
$\varX$ and $\varY$ in $\ExpressionF$ and $\new\Channel\Process$ binds
any occurrence of the endpoints 
$\Channel^+$ and $\Channel^-$ or of the shared channel $\Channel$
 within $\Process$.  
The definitions of \emph{free} and \emph{bound} names
follow as expected. We identify expressions and processes up to
renaming of bound names.

\begin{table}[t]
\caption{Reduction semantics of expressions and processes.}
\label{table:semantics}
\framebox[\textwidth]{
\begin{framedmath}
  \begin{array}{@{}c@{}}
    \multicolumn{1}{@{}l@{}}{\textbf{Reduction of expressions}}
    \\\\
    \inferrule[\defrule{r-beta}]{}{
      (\Fun\var\ExpressionE)~\ExpressionF
      \red
      \ExpressionE \subst\ExpressionF\var
    }
    % \qquad
    % \inferrule[\defrule{r-fix}]{}{
    %   \fix~\Expression
    %   \red
    %   \Expression~(\fix~\Expression)
    % }
    \qquad
    \inferrule[\defrule{r-bind}]{}{
      \Bind{\Return\ExpressionE}\ExpressionF
      \red
      \ExpressionF\ExpressionE
    } 
    \\\\
    \inferrule[\defrule{r-split}]{}{
      \Split{\Pair{\Expression_1}{\Expression_2}}\varX\varY\ExpressionF
      \red
      \ExpressionF\subst{\Expression_1,\Expression_2}{\varX,\varY}
    }
    \qquad
    \inferrule[\defrule{r-ctxt}]{
      \ExpressionE \red \ExpressionF
    }{
      \Context[\ExpressionE]
      \red
      \Context[\ExpressionF]
    }
    \\\\
    \multicolumn{1}{@{}l@{}}{\textbf{Reduction of processes}}
    \\\\
    \inferrule[\defrule{r-open}]{}{
      \server\Channel\Expression \parop \thread\varX{\PContext[\open~\Channel]}
      \red
      \server\Channel\Expression
      \parop \new{\ChannelC\varY}(
      \thread\varX{\PContext[\Return{\ChannelC^+}]}
      \parop
      \thread\varY{\Expression~\ChannelC^-}
      )
    }
    \\\\
    \inferrule[\defrule{r-comm}]{}{
      \thread\varX{\PContext[\send~\Channel^\Polarity~\Expression]}
      \parop
      \thread\varY{\PContext'[\receive~\Channel^{\co\Polarity}]}
      \red
      \thread\varX{\PContext[\Return{\Channel^\Polarity}]}
      \parop
      \thread\varY{\PContext'[\Return{\Pair\Expression{\Channel^{\co\Polarity}}}]}
    }
    \\\\
    \inferrule[\defrule{r-future}]{}{
      \thread\varX{\PContext[\future~\Expression]}
      \red
      \new\varY(\thread\varX{\PContext[\Return\varY]} \parop \thread\varY\Expression)
    }
    \\\\
    \inferrule[\defrule{r-return}]{}{
      \new\var(\thread\var{\Return\Expression} \parop \Process)
      \red
      \Process\subst\Expression\var
    }
    \\\\
    \inferrule[\defrule{r-thread}]{
      \ExpressionE \red \ExpressionF
    }{
      \thread\var\ExpressionE
      \red
      \thread\var\ExpressionF
    }
    \qquad
    \inferrule[\defrule{r-new}]{
      \ProcessP \red \ProcessQ
    }{
      \new\Thing\ProcessP
      \red
      \new\Thing\ProcessQ
    }
    \qquad
    \inferrule[\defrule{r-par}]{
      \ProcessP \red \ProcessQ
    }{
      \ProcessP \parop \ProcessR
      \red
      \ProcessQ \parop \ProcessR
    }
    \qquad
    \inferrule[\defrule{r-cong}]{
      \ProcessP \equiv \ProcessP' \red \ProcessQ' \equiv \ProcessQ
    }{
      \ProcessP
      \red
      \ProcessQ
    }
  \end{array}
\end{framedmath}
}
\end{table}

The operational semantics of expressions is defined in the upper half
of \cref{table:semantics}.
Expressions reduce according to a standard \emph{call-by-name}
semantics, for which we define the \emph{evaluation contexts for
  expressions} below:
\[
\Context ::=
  \Hole
  \mid \Context\Expression
  \mid \Split\Context\varX\varY\Expression
  \mid \open~\Context
  \mid \send~\Context
  \mid \receive~\Context
  \mid \bind~\Context
\]
Note that evaluation contexts do not allow to reduce 
 pair components
or an expression $\Expression$ in 
 $\Fun{\var}{\Expression}$, 
$\bind\ \ExpressionF\ \Expression$,
$\return\ \Expression$, $\future\ \Expression$  and  $\send\ a^p\
\Expression$. 
We say that $\ExpressionE$ is in \emph{normal form} if there is no
$\ExpressionF$ such that $\ExpressionE \red \ExpressionF$.
%
% \textcolor{orange}{
% %
% Observe that the semantics for expressions %in \cref{table:semantics}
% yields the following reductions $\Context[\Bind \Expression
% \ExpressionF] \red^* \Context[\Bind {\return\ \Expression'}
% \ExpressionF] \red \Context[\ExpressionF\ \Expression']$, provided
% that $\Expression \red^* \return\ \Expression'$.
% %
% This permits to \lq\lq suspend\rq\rq\ at $\Expression'$ the evaluation
% of $\Expression$ and start evaluating the continuation of the bind
% $\ExpressionF$.
% }

The operational semantics of processes is given by a structural
congruence relation $\equiv$ (which we leave undetailed since it is
essentially the same as that of the $\pi$-calculus \cite{DBLP:books/daglib/0004377}) and a reduction
relation, defined in the bottom half of \cref{table:semantics}.
The \emph{evaluation contexts for processes} are defined
as\label{ec}
\[
  \PContext ::= \Hole \mid \Bind\PContext\Expression
\]
and force the left-to-right execution of monadic actions, as usual.

Rules~\refrule{r-open} and~\refrule{r-comm} model session
initiation and communication,  respectively. According to \refrule{r-open}, a client
thread opens a connection with a server $\ChannelA$. In the reduct, a
fresh session channel $\ChannelC$ is created, the $\open$ in the client 
%reduces to the  return of  $\ChannelC^+$ 
is replaced by the endpoint $\ChannelC^+$ 
wrapped in the constructor   $\return$.  
 Moreover,  a copy
of the server is spawned into a new thread that has a fresh name
$\varY$ and 
% a body obtained from   that of the server applied to $\ChannelC^-$.
a  body which is the application of 
the expression $\Expression$ (provided by the server) to $\ChannelC^-$. 
 This follows a continuation-passing style since $\Expression$
is a function expecting the end-point of a channel.  
So client and server can communicate using the private channel $\ChannelC$. 
According to \refrule{r-comm},
two threads communicate if one is ready to send some message
$\Expression$ on a session endpoint $\Channel^\Polarity$ and the other
is waiting for a message from the peer endpoint
$\Channel^{\co\Polarity}$. As in~\cite{GayV10}, 
%the
%communication primitives return the session endpoint being used, 
%with
%the difference that in our case the results 
% are monadic actions. 
%In
%particular, the result for the sender is the same session endpoint and
%the result for the receiver is a pair consisting of the received
%message and the session endpoint.
the result for the sender is the same session endpoint and
the result for the receiver is a pair consisting of the received
message and the session endpoint.
The difference is that in our case the results 
have to be wrapped in the constructor $\return$ for monadic actions.

Rules~\refrule{r-future} and~\refrule{r-return} deal with futures.
The former spawns an I/O action $\Expression$ in a separate thread
$\varY$, so that the spawner is able to reduce (using
\refrule{r-bind}) even if $\Expression$ has not been executed yet.
The name $\varY$ of the spawned thread 
% can be 
 is  used as a placeholder
for the value yielded by $\Expression$.
Rule~\refrule{r-return} deals with a future variable $\var$ that has
been evaluated to $\return~\Expression$. In this case, $\var$ 
can be 
replaced by $\Expression$ everywhere within its scope.
% \bchm The
%implementation of this rule is delicate being $\Process$ a parallel of
%running threads.  Since the replaced value is immutable, the
%implementation of the substitution can be done with a non atomic
%mechanism, like a broadcast or a multicast,
%\echm
%\bchpaula Emilio says: shall we remove the rest? Paula says:
%or do we remove the bit on
%broadcast and multicast? \echpaula
%\bchm and lazily only when the
%value of $\var$ is needed. \echm
Note that the rule replaces in a single step the variable $\varX$ in an arbitrary parallel composition of threads running on possibly different hosts. In this respect, the practical realisation of this rule may appear critical, if at all possible. In fact, since the replaced value is immutable, the reduction rule can be implemented without synchronising all the threads that are affected by the replacement, for example by means of a broadcast or multicast communication.

Rule~\refrule{r-thread} lifts reduction of expressions to reduction of
threads. The remaining rules close reduction under restrictions,
parallel compositions, and structural congruence, as expected.  Hereafter, we write $\red^*$ for the reflexive, transitive closure
of $\red$. 

As an example, let 
\[
  \ProcessQ ~ = ~ \new {\mathit{prod}\,\mathit{cons}\,\Channel \,\ChannelC}(\Process\parop\server\Channel~\dspf)
\]
where
\[
  \Process ~ = ~\thread {\mathit{prod}} \upkf\ \ChannelC^+\ (\streamf\ 0)
  \parop 
  \thread {\mathit{cons}} {\display\ \ChannelC^-}
\]
is the process discussed in the introduction. 
It is easy to verify
that
\[
  \Process_0 ~= ~
  \new {\,\mathit{prod}\,  \Channel}
  (
  \thread {\mathit{prod}} {\Bind{\open~\Channel}
    { \Fun \varY {\upkf\ \varY\ (\streamf\ 0)}}}  
  \parop     
  \server\Channel~\dspf  )
  \]
  reduces to process $\ProcessQ$. 

%
%We say that
%$\ExpressionE$ is \emph{irreducible} (or in \emph{normal form}) if
%there is no $\ExpressionF$ such that $\ExpressionE \red
%\ExpressionF$. Similarly, we say that $\ProcessP$ is
%\emph{irreducible} if there is no $\ProcessQ$ such that
%$\ProcessP \red \ProcessQ$.

%%% Local Variables:
%%% mode: latex
%%% TeX-master: "LazySessions"
%%% End:

\section{Typing \sid}
\label{section:typing}

We now develop a typing discipline for \sid.
The challenge comes from the fact that the calculus allows a mixture
of pure computations (handling data) and impure computations (doing
I/O).
In particular, \sid\ programs can manipulate potentially infinite data
while performing I/O operations that produce/consume pieces of such
data as shown by the examples of \cref{section:Introduction}.
Some ingredients of the type system are easily identified from the
syntax of the calculus. We have a core type language with unit,
products, and arrows. As in~\cite{GayV10}, we distinguish between
\emph{unlimited} and \emph{linear} arrows for there sometimes is the
need to specify that certain functions must be applied exactly
once. As in Haskell~\cite{JonesW93,JonesTutorial2001}, we use the
$\mkbasictype{IO}$ type constructor to denote monadic I/O actions. For
shared and session channels we respectively introduce channel types
and session types~\cite{DBLP:conf/esop/HondaVK98}. Finally,
following~\cite{Nakano00:lics}, we introduce the \emph{delay} type
constructor $\tbullet$, so that an expression of type $\tbullet\Type$
denotes a value of type $\Type$ that is available ``at the next moment
in time''.
This constructor is key to control recursion and attain normalisation
of expressions. Moreover, the type constructors $\tbullet$ and
$\mkbasictype{IO}$ interact in non-trivial ways as shown later by the
type of $\future$.

\subsection{Types}\label{subsec:types}

\begin{table}[h]
\caption{Syntax of Pre-types and Pre-session types.}
\label{table:typessyntax}
\framebox[\textwidth]{
\begin{framedmath}
\begin{array}{@{}c@{\qquad}c@{}}
  \begin{array}[t]{r@{~~}c@{~~}l@{\quad}l}
    \Type & ::=^{coind} & & \textbf{Pre-type} \\
    &   & \tbasic & \text{(basic type)} \\
    & | & \SessionType & \text{(session type)} \\
    & | & \tshared\SessionType & \text{(shared channel type)} \\
    & | & \Type \times \Type & \text{(product)} \\
    & | & \Type \arrow \Type & \text{(arrow)} \\
    & | & \Type \larrow \Type & \text{(linear arrow)} \\
    & | & \tio\Type & \text{(input/output)} \\
    & | & \tbullet\Type & \text{(delay)} \\
  \end{array}
  &
  \begin{array}[t]{r@{~~}c@{~~}l@{\quad}l}
    \SessionType & ::=^{coind} & & \textbf{Pre-session type} \\
    &   & \End & \text{(end)} \\
    & | & \tin\Type\SessionType & \text{(input)} \\
    & | & \tout\Type\SessionType & \text{(output)} \\
    & | & \tbullet\SessionType & \text{(delay)}
  \end{array}
\end{array}
\end{framedmath}
}
\end{table}

The syntax of \emph{pre-types} and \emph{pre-session types} is
given by the grammar in \cref{table:typessyntax}, whose productions
are meant to be interpreted coinductively. A pre-(session) type is
a possibly infinite tree, where each internal node is labelled by a
type constructor and has as many children as the arity of the
constructor.
The leaves of the tree (if any) are labelled by either basic types or
$\End$.
We use a coinductive syntax to describe the type of infinite data structures 
(such as streams) and arbitrarily long protocols, e.g. the one between 
$\mathit{prod}$ and $\mathit{cons}$ in Section~\ref{section:Introduction}.

We distinguish between unlimited pre-types (those denoting expressions
that can be used any number of times) from linear pre-types (those
denoting expressions that must be used exactly once).
Let $\linpred$ be the 
smallest predicate defined by
\[
\linq{\tin\Type\SessionType}
\qquad
\linq{\tout\Type\SessionType}
\qquad
\linq{\TypeT \larrow \TypeS}
\qquad
\linq{\tio\Type}
\qquad
\inferrule{
  \linq\TypeT
}{
  \linq{\TypeT \times \TypeS}
}
\qquad
\inferrule{
  \linq\TypeS
}{
  \linq{\TypeT \times \TypeS}
}
\qquad
\inferrule{
  \linq\Type
}{
  \linq{\tbullet\Type}
}
\]
We say that $\Type$ is \emph{linear} if $\linq\Type$ holds and that
$\Type$ is \emph{unlimited}, written $\unq\Type$, otherwise. Note that
all I/O actions are linear, since they may involve communications on
session channels which are linear resources.

\begin{definition}[Types]\label{def:types}
  A pre-(session) type $\Type$ is a \emph{(session) type} if:
  \begin{enumerate}
  \item\label{def:types4} For each sub-term $\Type_1 \arrow \Type_2$
    of $\Type$ such that $\unq{\Type_2}$ we have $\unq{\Type_1}$.
  \item\label{def:types5} For each sub-term $\Type_1 \larrow \Type_2$
    of $\Type$ we have $\linq{\Type_2}$.
  \item\label{def:types2} The tree representation of $\Type$ is
    regular, namely it has finitely many distinct sub-trees.
  \item\label{def:types3} Every infinite path in the tree
    representation of $\Type$ has infinitely many $\bullet$'s.
%  \item \label{def:types1} $\Type$ is not $\infinite$.
  \end{enumerate}
\end{definition}

All conditions except possibly~\ref{def:types3} %and~\ref{def:types1}
are natural.
Condition~\ref{def:types4} essentially says that unlimited functions
are \emph{pure}, namely they do not contain and they cannot
erase communications.
%have side effects. 
Indeed, an
unlimited function (one that does not contain linear names) %in its closure 
that accepts a linear argument should return 
a linear
%such argument in the
result.
%{\color{blue} should contain such argument in the result.}
%
Condition~\ref{def:types5} states that a linear function (one that may contain linear names) %in its closure 
always yields a linear
result. This is necessary to keep track of the presence of linear
names in the function, %'s closure 
even when the function is applied and
its linear arrow type eliminated.
For example, consider $\varZ$ of type $\nat \larrow \nat$ and both $\varY$ and $\varW$ of type $\nat$, then without Condition~\ref{def:types5}
 we could type $(\lambda \var.\varY)(\varZ\ \varW)$ with $\nat$.
 This would be incorrect,
  because it discharges the expression 
 %which would be wrong because it discharges the computation 
 $(\varZ\ \varW)$ involving
the linear name $\varZ$.
Condition~\ref{def:types2} implies that we only consider types
admitting a finite representation, for example using the well-known
``$\mu$ notation'' for expressing recursive types (for the relation
between regular trees and recursive types we refer to~\cite[Chapter
20]{Pierce02}).
%the relation
%between regular trees and recursive types can be found
%in~\cite[Ch. 20]{Pierce02}). 
We define infinite types as trees
satisfying a given recursive equation, for which the existence and
uniqueness of a solution follow from known
results~\cite{Courcelle83}. For example, there 
are unique
pre-types $\badListN$, $\ListN$, and $\infinite$ that respectively
satisfy the equations $\badListN=\nat\times\badListN$,
$\ListN=\nat\times \bullet \ListN$, and $\infinite=\bullet\infinite$.
%
% \Cref{figure:infinitetypes} depicts (part of) the tree representation
%of these pre-types.
%\begin{figure}
%  \include{FigurePreTypes}
%  \caption{Tree representation of some pre-types}
%  \label{figure:infinitetypes}
%\end{figure}
% $\badListN=\nat\times\badListN$,
%$\ListN=\nat\times \bullet \ListN$, and $\infinite=\bullet\infinite$.
%$\ListN$ is a type because  the only infinite path in its tree
%representation (the right spine of $\ListN$ in
%\cref{figure:infinitetypes}) has infinitely many $\bullet$'s. 
% Instead $\badListN$ is not a type because  its tree
%representation has an infinite path without any $\bullet$'s. 
\emph{En passant}, note that linearity is decidable on types due to
Condition~\ref{def:types2}.
The fact that $\linpred$ has been defined above as the
  \emph{smallest} predicate that satisfies certain axioms and rules is
  crucial.
  In particular, $\infinite$ is not linear.

Condition~\ref{def:types3} intuitively means that not all parts of an
infinite data structure can be available at once: those whose type is
prefixed by a $\bullet$ are ``delayed'' in the sense that recursive
calls on them must be deeper.
For example, $\ListN$ is a type %-- the only infinite path in its tree
%representation (the right spine of $\ListN$ in
%\cref{figure:infinitetypes}) has infinitely many $\bullet$'s -- 
that denotes streams of natural numbers where each subsequent element
of the stream is delayed by one $\bullet$ compared to its
predecessor. Instead $\badListN$ is not a type: %-- its tree
%representation has an infinite path without any $\bullet$'s -- 
it %which
would denote an infinite stream of natural numbers, whose elements are
all available right away.
Similarly, $\OutN$ and $\InN$ defined by
$\OutN = \tout{\nat} \bullet \OutN$ and
$\InN = \tin{\nat} \bullet \InN$ are session types, while $\badOutN$
and $\badInN$ defined by $\badOutN = \tout{\nat} \badOutN$ and
$\badInN = \tin{\nat} \badInN$ are not.
The type $\infinite$ is somehow degenerate in that it contains no
actual data constructors. Unsurprisingly, we will see that
non-normalising terms such as
$\omegaterm = (\lambda \var. \var\ \var)(\lambda \var. \var\ \var)$ can
only be typed with $\infinite$.  Without Condition~\ref{def:types3},
$\omegaterm$ could be given any type.

%We adopt the usual conventions of parentheses and  arrow types associate to the right.
We adopt the usual conventions of parentheses.
Arrow types associate to the right.
We assume the following precedence among type constructors: $\bullet$, $\mkbasictype{IO}$, 
$\times$, followed by $\arrow$ and $\larrow$ with the same (and lowest) precedence.  
We also need a notion of duality to relate the session types
associated with peer endpoints. Our definition extends the one
of~\cite{DBLP:conf/esop/HondaVK98} in the obvious way to delayed
types. More precisely, the \emph{dual} of a session type
$\SessionType$ is the session type $\co\SessionType$ coinductively
defined by the equations:
\[
  \co\End = \End
  \qquad
  \co{\tin\Type\SessionType} = \tout\Type\co\SessionType
  \qquad
  \co{\tout\Type\SessionType} = \tin\Type\co\SessionType
  \qquad
  \co{\tbullet\SessionType} = \tbullet\co\SessionType 
\]
  Sometimes we will write $\tbullet[n]\Type$ in place of
  $\smash{\underbrace{\tbullet\cdots\tbullet}_{n\text{-times}}\Type}$.
%, where there are $n$ $\tbullet$'s

\subsection{Typing Rules for Expressions}
\label{sec:typingexpressions}
%We show the typing of expressions and processes.
%
First we assign types to constants:
%, where $\TypeT$, $\TypeS$, and $\SessionType$ stand for any type and session type.
%
{\small\[
\begin{array}[t]{l@{\!\,\,\,}l@{\!\,\,\,}l}
  \unit & : & \Unit \\
  \return & : & \Type\arrow\tio\Type \\
  \open & : & \tshared\SessionType\arrow\tio\SessionType \\
\end{array}
\quad
\begin{array}[t]{l@{\!\,\,\,}l@{\!\,\,\,}l}
  \send & : & \tout\Type\SessionType\arrow\Type\larrow\tio\SessionType \\
  \receive & : & \tin\Type\SessionType \arrow \tio(\Type\times\SessionType) \\
  \future & : & \tbullet^n(\tio\Type)\arrow\tio{\tbullet^n\Type}
\end{array}
\quad
\begin{array}[t]{l@{\!\,\,\,}l@{\!\,\,\,}l}
  \pair & : & \Type\arrow\TypeS\larrow\Type\times\TypeS \hfill \text{if $\linq\Type$} \\
  \pair & : & \Type\arrow\TypeS\arrow\Type\times\TypeS \hfill \text{if $\unq\Type$} \\
  \bind & : & \tio\Type\arrow(\Type\larrow\tio\TypeS)\larrow\tio\TypeS
\end{array}
\]}

\noindent
Each constant $\Constant \neq \unit$ is polymorphic and we use
$\typeof(\Constant)$ to denote the set of types assigned to
$\Constant$, e.g. $\typeof(\return) = \bigcup_t\{ \Type \arrow \tio
\Type \}$.

The types of $\unit$ and $\return$ are as expected.
The type schema of $\bind$ is similar to the type it has in Haskell,
except for the two linear arrows.
The leftmost linear arrow allows linear functions as the second
argument of $\bind$. The rightmost linear arrow is needed to satisfy
Condition~\ref{def:types4} of \cref{def:types}, being $\tio\Type$
linear.
The type of $\pair$ is also familiar, except that the second arrow is
linear or unlimited depending on the first element of the pair. If the
first element of the pair is a linear expression, then it can (and
actually must) be used for creating exactly one pair.
The types of $\send$ and $\receive$ are almost the same as
in~\cite{GayV10}, except that these primitives return I/O actions
instead of performing them as side effects.
The type of $\open$ is standard and obviously justified by its
operational semantics.
The most interesting type is that of $\future$, which commutes delays
and the $\mkbasictype{IO}$ type constructor. Intuitively, $\future$
applied to a delayed I/O action returns an immediate I/O that yields a
delayed expression.
This fits with the semantics of $\future$, since its argument is
evaluated in a separate thread and the one invoking $\future$ can
proceed immediately with a placeholder for the delayed expression.
If the body of the new thread reduces to $\return~\Expression$, then
$\Expression$ substitutes the placeholder.
% \footnote{Funnily enough,
%   $\future$ plays the same role as $\mkconstant{delayIO}$
%   in~\cite{JonesW93}, where it is presented as a primitive that delays
%   an I/O action, whence the name.}
%
% The application of $\send$ and $\receive$ to polarised channels must
% be typed by linear types, since polarised channels cannot be erased or
% duplicated. The shapes of these types as well as the shape of the
% types of $\open$ are justified by looking at their reduction rules.
% The types of $\return$ and $\bind$ are as expected. The more
% interesting types are those of $\future$: they commute delays with
% input/output. If $n=0$, we get for $\future$ the type assumed
% in~\cite{SabelS11}. We can informally say that in our calculus
% $\future$ applied to a delayed input/output expression returns an
% input/output delayed expression. This fits with the semantics of
% $\future$, since the argument of $\future$ is evaluated in a fresh
% thread. If the body of this new thread reduces to
% $\return~\Expression$, then $\Expression$ becomes argument of
% $\return$ in the original thread.

The typing judgements for expressions have the shape
$\wte\TypeContext\Expression\Type$, where \emph{typing environments}
(for used resources) $\TypeContext$ are mappings from variables to
types, from shared channels to shared channel types, and from
endpoints to session types:
\[
  \TypeContext \quad ::= \quad
  \emptyset \quad\mid \quad
  \TypeContext,x:\Type \quad\mid\quad
  \TypeContext,\Channel:\tshared\SessionType \quad\mid\quad
  \TypeContext, \Channel^{\Polarity}:\SessionType
\]
The domain of $\TypeContext$, written $\dom(\TypeContext)$, is defined
as expected.  A typing environment $\TypeContext$ is {\em linear},
notation $\linq \TypeContext$, if there is
$\Name : \Type \in \TypeContext$ such that $\linq \Type$; otherwise
$\TypeContext$ is {\em unlimited}, notation $\unq \TypeContext$.  As
in~\cite{GayV10}, we use a (partial) combination operator $+$ for
environments that prevents names with linear types from being used
more than once.  Formally the environment
$\TypeContext + \TypeContext'$ is defined inductively on
$\TypeContext'$ by
\[
  \begin{array}{@{}rcl@{}}
    \TypeContext + \emptyset & = & \TypeContext
    \\
    \ptc{\TypeContext}{(\TypeContext', \Name:\Type)} & = &
    \ptc{ (\ptc{\TypeContext}{\TypeContext'})}{\Name:\Type}
  \end{array}
\quad\text{where}\quad
\ptc{\TypeContext}{\Name:\Type}=\begin{cases}
\TypeContext,\Name:\Type      & \text{if }\Name\not\in\dom(\TypeContext), \\
\TypeContext      & \text{if }\Name:\Type\in\TypeContext \text{ and }\unq\Type, \\
 \text{undefined}     & \text{otherwise}.
\end{cases}
\]

\begin{table}[t]
  \caption{\label{table:eTypingrules}Typing rules for expressions.}
  \framebox[\textwidth]{
    \begin{framedmath}
      \begin{array}{@{}c@{}}
        \inferrule[\defrule\introbullet]{
        \wte{\TypeContext}{\Expression}\Type
        }{
        \wte{\TypeContext}{\Expression}{\tbullet\Type}
        }
        % ~~\unq \Type
        \qquad
        \inferrule[\defrule\const]{
        \strut
        }{
        \wte\TypeContext\Constant\Type
        }
        ~~
        \begin{lines}[c]
          \unq{\TypeContext},
          \Type \in \typeof(\Constant)  
        \end{lines}
        \qquad
        \inferrule[\defrule\axiom]{
        \strut
        }{
        \wte{\TypeContext, \Name: \Type}{\Name}{\Type}
        }
        ~~
        \unq\TypeContext
        \\\\
        \inferrule[\defrule\introarrow]{
        \wte{\TypeContext,\varX: \tbullet[n]\Type }{\Expression}{\tbullet[n]\TypeS}
        }{
        \wte{\TypeContext}{\Fun \varX \Expression}{\tbullet[n](\Type \arrow \TypeS)}
        }
        ~~\unq{\TypeContext}
        \quad
        \inferrule[\defrule\elimarrow]{
          \wte{\TypeContext_1}{\Expression_1}{\tbullet[n](\Type \arrow \TypeS)}
          \\
          \wte{\TypeContext_2}{\Expression_2}{\tbullet[n]\Type}
        }{
        \wte{\TypeContext_1 + \TypeContext_2}{\Expression_1\Expression_2}{\tbullet[n]\TypeS}
        }
        \quad
        \inferrule[\defrule\introlarrow]{
        \wte{\TypeContext,\varX: \tbullet[n]\Type }{\Expression}{\tbullet[n]\TypeS}
        }{
        \wte{\TypeContext}{\Fun \varX \Expression}{\tbullet[n](\Type \larrow \TypeS)}
        }
        \\\\
        \inferrule[\defrule\elimlarrow]{
          \wte{\TypeContext_1}{\Expression_1}{\tbullet[n](\TypeT \larrow \TypeS)}
          \\
          \wte{\TypeContext_2}{\Expression_2}{\tbullet[n]\Type}
        }{
          \wte{\TypeContext_1 + \TypeContext_2}{\Expression_1\Expression_2}{\tbullet[n]\TypeS}
        }
        \quad
        \inferrule[\defrule\elimprod]{
          \wte{\TypeContext_1}{\ExpressionE}{\tbullet[n](\Type_1 \times \Type_2)}
          \\
          \wte{\TypeContext_2, \varX : \tbullet[n]\Type_1, \varY : \tbullet[n]\Type_2}{\ExpressionF}{\tbullet[n]\TypeS}
        }{
          \wte{\TypeContext_1+\TypeContext_2}{\Split\ExpressionE\varX\varY\ExpressionF}{\tbullet[n]\TypeS}
        }
      \end{array}
    \end{framedmath}
  }
\end{table}

The typing axioms and rules for expressions are given
in~\cref{table:eTypingrules}. 
The side condition $\unq\TypeContext$ in $\refrule\const$,
$\refrule\axiom$, and $\refrule\introarrow$ is standard~\cite{GayV10}.
The typing rules  differ from the ones in~\cite{GayV10} on
 two crucial details. 
First of all, each rule allows for an arbitrary delay in front of the
types of the entities involved. Intuitively, the number of
$\tbullet$'s represents the delay at which a value becomes
available. So for example, rule~\refrule\introarrow{} says that a
function which accepts an argument $\var$ of type $\TypeT$ delayed by
$n$ and produces a result of type $\TypeS$ delayed by the same $n$ has
type $\tbullet[n](\TypeT \arrow \TypeS)$, that is a function delayed
by $n$ that maps elements of $\TypeT$ into elements of $\TypeS$.
The second difference with respect to the type system in~\cite{GayV10}
is the presence of rule~\refrule\introbullet, which allows to further
delay a value of type $\Type$.
Crucially, it is not possible to \emph{anticipate} a delayed value: if
it is known that a value will only be available with delay $n$, then
it will also be available with any delay $m \geq n$, but not earlier.

\subsection{Examples of Type Derivations for  Expressions} 
Using rule $\refrule\introbullet$ and the recursive type 
$s= \tbullet s \rightarrow t$, we can derive that the fixed point
combinator\label{fix} % operator
\[
\fix = \lambda \varY. (\lambda \var. \varY\ (\var\ \var) ) (\lambda
\var. \varY\ (\var\ \var) )
\]
has type $(\tbullet \Type \rightarrow \Type)\rightarrow \Type$
by assigning 
the type $s \rightarrow t$
to the first occurrence of $\lambda \var. \varY\ (\var\ \var) $
and the type  $\tbullet s \rightarrow t$
to the second one~\cite{Nakano00:lics}.

%The difference is that $\fix$ is typeable in our system only  if $\unq \Type$
%due to the condition of the rule \refrule\introbullet. 
%  This condition is necessary, since otherwise we
%  could  derive that 
%  $(\fix~\return)$ 
%  has type $\Type$ where $\Type = \bullet (\tio \Type)$ 
%  and that  $\fix~\future$
%  has type $(\tio \infinite)$
%  where $\infinite = \tbullet \infinite$.
%
% There is no rule of $\bullet$ elimination, since we cannot anticipate
% something that we have not got yet.
%
% Notably in all rules but $\refrule\introbullet$ all involved types
% have the same number of $\bullet$'s. This is not a constraint in case
% of unlimited types, thanks to rule $\refrule\introbullet$, but it
% avoids postponing the use of expressions typed by linear types.

%In the examples, we will use the rules \refrule{\introprod},
%\refrule{\fixrule}, \refrule{\bindrule}, \refrule{\futurerule},
%\refrule{$\times \arrow$ I} and \refrule\returnrule
%

% \medskip

It is possible to derive the following types for the functions in
\Cref{section:Introduction}:
\[
  \begin{array}{l@{\quad}l@{\quad}l}
    \streamf :   \nat \arrow \ListN 
    &
      \upkf  :   \OutN \arrow \ListN \arrow \tio \infinite 
    &
      \dspf  :  \InN  \arrow \tio \ListN
  \end{array}
\]
where, in the derivation for $\dspf$, we assume type
$\ListN \arrow \tio \ListN$ for $\ExpressionG$. We show the most
interesting parts of this derivation. %For this 
We use the following
rules, which are easily  obtained  %derived 
from those
in~\cref{table:eTypingrules} and the types of the constants.
\[
\begin{array}{lll}
%  \inferrule[\defrule\introprod]{
%  \wte{\TypeContext_1}{\Expression_1}{\tbullet[n]\TypeT}
%  \\
%  \wte{\TypeContext_2}{\Expression_2}{\tbullet[n]\TypeS}
%  }{
%  \wte{\TypeContext_1 + \TypeContext_2}{\Pair{\Expression_1}{\Expression_2}}{\tbullet[n] (\Type \times \TypeS)}
%  }
%&\qquad\qquad&
  \inferrule[\defrule\fixrule]{
  \wte{\TypeContext, x: \bullet \Type }{\Expression}{\TypeT}
  }{
  \wte{\TypeContext}{\fix~\Fun{x} \Expression}{\Type}
  }~~\unq \TypeContext
  &\qquad\qquad&
  \inferrule[\defrule\bindrule]{
  \wte{\TypeContext_1}{\Expression_1}{\tbullet[n](\tio\TypeT)}
  \\
  \wte{\TypeContext_2}{\Expression_2}{\tbullet[n](\Type \larrow \tio \TypeS)}
  }{
  \wte{\TypeContext_1 + \TypeContext_2}{\Bind{\Expression_1}{\Expression_2}}{\tbullet[n] \tio \TypeS}
  }
\\\\
\inferrule[\defrule\futurerule]{
  \wte{\TypeContext }{\Expression}{\tbullet[n+m] \tio \TypeT}
  }{
  \wte{\TypeContext}{\future~\Expression}{\tbullet[n] \tio \tbullet[m] \Type}
  }
 &&
  \inferrule[\defrule{$\times \arrow$ I}]{
  \wte{\TypeContext, \varX_1: \tbullet[n]\Type_1, \varX_2 :\tbullet[n]\Type_2 }{\Expression}{\tbullet[n]\TypeS}
  }{
  \wte{\TypeContext}
      {\Fun{\Pair{\varX_1}{\varX_2}} \Expression}
      {\tbullet[n](\Type_1 \times \Type_2 \arrow \TypeS)} 
  } \quad \unq \TypeContext   
\end{array}
\]
In order to  type 
%of 
$\display$ we desugar its recursive definition
  %in Section~\ref{section:Introduction} 
as $\display = \fix~(\Fun \var \Fun \varY \Expression)$, where
\[
  \begin{array}{lcl}
    \Expression & = & \Bind{\Expression_1}{\Expression_2}
  \end{array}
  \qquad
  \begin{array}{lcl}
    \Expression_1 & = & \Receive \varY
    \\
    \Expression_2  & = & \Fun {\Pair{\varZ}{\varY'}}
                         \Bind{\Expression_3}{\Expression_4}
  \end{array}
  \qquad
  \begin{array}{lcl}
    \Expression_3 & = & \future{\big(\var \ \varY' \big)}
    \\
    \Expression_4 & = & \Fun {\varZ s}
                        {\ExpressionG {\Pair{\varZ}{\varZ s}}}
  \end{array}
\]

%Now 
We  derive 
\begin{prooftree}
   %\alwaysNoLine
   \AxiomC{$\vdots$}
   \UnaryInfC{$\TypeContext_2
    \vdash  \Expression_1 :\tio (\nat \times \bullet \InN) $}
       \AxiomC{$\nabla$}
     \UnaryInfC{$\TypeContext, \TypeContext_1,\TypeContext_3, \TypeContext_4
    \vdash  \Bind{\Expression_3}{\Expression_4}: \tio \ListN$}
   \LeftLabel{\refrule{$\times \arrow$ I}}
   \alwaysSingleLine
    \UnaryInfC{$\TypeContext,\TypeContext_1
    \vdash  \Expression_2 : (\nat \times \bullet \InN) \arrow \tio \ListN$}
   \LeftLabel{\refrule{\bindrule}} 
   \BinaryInfC{$\TypeContext, \TypeContext_1,\varY: \InN
    \vdash   \Expression : \tio \ListN$}
   \LeftLabel{\refrule{\introarrow}} 
   \UnaryInfC{$\TypeContext,\TypeContext_1
                     \vdash \Fun{\varY} \Expression :  \InN \arrow \tio  \ListN$}
\LeftLabel{\refrule{\fixrule}} 
\UnaryInfC{$\TypeContext\vdash \display : \InN \arrow \tio \ListN$}
\end{prooftree}
\vspace{4pt}
where 
$\TypeContext = \ExpressionG: \ListN \arrow \tio \ListN$,
$\TypeContext_1 = \var: \bullet (\InN \arrow \tio \ListN)$,
$\TypeContext_2 = \varY: \InN$,
$\TypeContext_3= \varY' : \bullet \InN$ and
$\TypeContext_4 = \varZ :\nat$.   
The derivation $\nabla$ is as follows.
\begin{prooftree}
 \AxiomC{$  \wte{\TypeContext_1}{\var}{\tbullet (\InN \arrow \tio \ListN)}
       $}
 \AxiomC{$\wte{\TypeContext_3}{\varY'}{\tbullet \InN }$}
   \LeftLabel{\refrule{\elimarrow}} 
   \BinaryInfC{$\wte{\TypeContext_1, \TypeContext_3}
                    {\var~\varY'}{\bullet \tio \ListN}$}
\LeftLabel{\refrule{\futurerule}} 
       \UnaryInfC{$\TypeContext_1, \TypeContext_3
    \vdash  \Expression_3: \tio \bullet \ListN$}
    \AxiomC{$\vdots$}
     \UnaryInfC{$ \TypeContext,\TypeContext_4
    \vdash  \Expression_4: \bullet \ListN \arrow \tio \ListN$}   
     \LeftLabel{\refrule{\bindrule}}
     \BinaryInfC{$\TypeContext, \TypeContext_1,\TypeContext_3, \TypeContext_4
    \vdash  \Bind{\Expression_3}{\Expression_4}: \tio \ListN$}
\end{prooftree} 
\vspace{4pt}
Note that the types of the premises of
\refrule{\elimarrow} in the above derivation have a $\bullet$
constructor in front.
Moreover, $\future$ has a type that pushes the $\bullet$ inside the
$\mkbasictype{IO}$; this is crucial for typing $\Expression_4$ with
$ (\bullet \ListN \arrow \tio \ListN)$.
We can assign the type $ \bullet \ListN \arrow \tio \ListN $ to
$\Expression_4$ by guarding the argument $\varZ$ of type
$\bullet \ListN$ under the constructor $\pair{}{}$.
Without $\future$,  
 the expression $\Bind{\var~\varY'}{\Expression_4}$ has
type $\bullet \tio \ListN$ and for this reason
$\bdisplay$ cannot be typed.

%the expression $\Bind{\Expression_3}{\Expression_4}$ would have
%type $\bullet (\tio \ListN)$ and 
%$\display$ would be
%untypeable.

Controlling guardedness of recursion is subtle as it could require
types with several bullets.
  For example, let $\Expression= \Split{ys}{y}{zs}{(s~zs)}$ and
  consider the function
  \[
  \skipeven = 
  \fix~\Fun s \Fun { \Pair{x} {ys}}
  \Pair{x}{\Expression}
  \]
  that deletes the elements at even positions of a stream.
  Function
  $\skipeven$ has type $\ListN \arrow \ListNtwo$, where $\ListNtwo =
  \nat \times \bullet \bullet \ListNtwo$. 
   We  derive
\begin{prooftree}
%\AxiomC{$\vdash \fix : (\bullet \Type \arrow \Type) \arrow \Type$}
%   \alwaysNoLine
%  \AxiomC{$xs: \ListN \vdash xs: \nat \times \bullet \ListN$} 
      \AxiomC{$\TypeContext \vdash x: \nat$}   
%         \alwaysNoLine
          \AxiomC{$\nabla$}
          \UnaryInfC{$\TypeContext \vdash  \Expression :\bullet \bullet \ListNtwo $}
          \alwaysSingleLine
      \LeftLabel{\refrule{\introprod}} 
   \BinaryInfC{$\TypeContext
   \vdash  \Pair{x} {\Expression}:\ListNtwo$}
   \alwaysSingleLine
   \LeftLabel{\refrule{\introarrow}} 
   \UnaryInfC{$s:\bullet ( \ListN \arrow \ListNtwo) \vdash 
    \Fun { \Pair{x} {ys}} \Pair{x}{\Expression} :\ListN \arrow \ListNtwo$}
  \LeftLabel{\refrule{\fixrule}} 
\UnaryInfC{$\vdash \skipeven : \ListN \arrow \ListNtwo$}
\end{prooftree}
\vspace{4pt}
where $\TypeContext = s:\bullet ( \ListN \arrow \ListNtwo), x: \nat,
ys: \bullet \ListN $, rule $\refrule\introprod$ is
\[\inferrule[\defrule\introprod]{
  \wte{\TypeContext_1}{\Expression_1}{\tbullet[n]\TypeT}
  \\
  \wte{\TypeContext_2}{\Expression_2}{\tbullet[n]\TypeS}
  }{
  \wte{\TypeContext_1 + \TypeContext_2}{\Pair{\Expression_1}{\Expression_2}}{\tbullet[n] (\Type \times \TypeS)}
  }\]  and the type derivation $\nabla$ is
\begin{prooftree}
\AxiomC{$\TypeContext \vdash ys:  \bullet (\nat \times \bullet  \ListN)$}
   \AxiomC{$\TypeContext' \vdash s:  \bullet (\ListN \arrow \ListNtwo)$}
   \LeftLabel{\refrule{\introbullet}} 
   \UnaryInfC{$\TypeContext' \vdash s: \bullet \bullet(\ListN \arrow \ListNtwo) $}
   \AxiomC{$\TypeContext' \vdash zs: \bullet\bullet \ListN$}
   \LeftLabel{\refrule{\elimarrow}} 
   \BinaryInfC{$ \TypeContext' \vdash s~zs: \bullet\bullet\ListNtwo$}
\LeftLabel{\refrule{\elimprod}} 
\BinaryInfC{$\TypeContext \vdash  \Expression :\bullet \bullet \ListNtwo $}
\end{prooftree}
\vspace{4pt}
where $\TypeContext'=s:\bullet (\ListN \arrow \ListNtwo), y: \bullet \nat, zs: \bullet\bullet \ListN $.
 Note that in the above derivation, the first premise of   
 $\refrule{\elimarrow}$
has two $\bullet$'s in front of the arrow type. The same derivation can be done in the system of~\cite{Nakano00:lics}. Instead~\cite{AM13} uses clock variables and~\cite{CBGB15} uses one constant to type this example as a particular case of lifting guarded recursive data to coinductive data.

\subsection{Typing Rules for Processes}
\label{sec:typingprocesses}
The typing judgements for processes have the shape
$\wtp\TypeContext\Process\ProvideContext$, where $\TypeContext$ is a
typing environment as before, while $\ProvideContext$ is a
\emph{resource environment}, keeping track of the resources defined in
$\Process$.
In particular, $\ProvideContext$ maps the names of threads and servers
in $\Process$ to their types and it is defined by
\[
\ProvideContext \quad ::= \quad \emptyset \quad\mid\quad
\ProvideContext,x:\Type \quad\mid\quad
\ProvideContext,\Channel:\tshared\SessionType
\]
%
% Hereafter, when writing $\ProvideContext_1, \ProvideContext_2$
% we implicitly assume that
% $\dom(\ProvideContext_1) \cap \dom(\ProvideContext_2) = \emptyset$.
%
\Cref{table:pTypingrules} gives the typing rules for processes.
A thread is well-typed if so is its body, which must be an I/O action.
The type of a thread is that of the result of its body, where the
delay moves from the I/O action to the result.
%
% The type associated with a thread is that of the result produced by
% its body with the same delay.
%
The side condition makes sure that the thread is unable to use the
very value that it is supposed to produce.
%and the resulting definition
%environment associates the name of the thread with the type of the value
%produced by its body.
%
%
The resulting environment for defined resources associates  the
name of the thread with the type of the action of its body.
%it
%  We can type a thread (rule $\refrule{thread}$) if its name
% does not occur in the use environment and the body has an output type
% (possibly delayed). This agrees with the reduction rule of $\return$,
% where the result of the evaluation of a thread is just the $\return$
% of an expression, then typed by an input/output type (possibly
% delayed).
%
%
A server is well-typed if so is its body $\Expression$, which must be
a function from the dual of $\SessionType$ to an I/O action.
This agrees with the reduction rule of the server, where the
application of $\Expression$ to an endpoint becomes the body of a new
thread each time the server is invoked.
It is natural to forbid occurrences of free variables and linear %shared
channels in server bodies. This is assured by the condition
$\sharedq \TypeContext$, which requires $\TypeContext$ to contain only
shared channels. Clearly $\sharedq \TypeContext$ implies
$\unq{\TypeContext}$, and then we can type the body $\Expression$ with
a non linear arrow.
The type of the new thread (which will be $\Type$
if $\Expression$ has type $\co\SessionType\arrow\tio\Type$) must be
unlimited, since a server can be invoked an arbitrary number of times. The environment $\TypeContext + \Channel : \tshared\SessionType$ in
the conclusion of the rule makes sure that the type of the server as
seen by its clients is consistent with its definition.

\begin{table}
\caption{\label{table:pTypingrules}Typing rules for processes.}
\framebox[\textwidth]{
\begin{framedmath}
\begin{array}{@{}c@{}}
  \inferrule[\defrule{thread}]{
    \wte{\TypeContext}{\Expression}{\tbullet[n](\tio\Type)}
  }{
    \wtp{\TypeContext}{\thread\var\Expression}{\varX:\tbullet[n]\Type}
  }~\varX \not\in \dom(\TypeContext)
  \qquad
   \inferrule[\defrule{server}]{
    \wte{\TypeContext}{\Expression}{\co\SessionType\arrow\tio\Type}
  }{
    \wtp{\TypeContext+\Channel:\tshared\SessionType}{\server\Channel\Expression}{\Channel:\tshared\SessionType}
  }
  ~~\begin{lines}[c]
    %\unq\TypeContext \\
     \sharedq \TypeContext \\
    \unq\Type
  \end{lines}
  \\\\
  \inferrule[\defrule{par}]{
    \wtp{\TypeContext_1}{\Process_1}{\ProvideContext_1}
    \qquad
    \wtp{\TypeContext_2}{\Process_2}{\ProvideContext_2}
    %
%        \\
%    \ctc{\TypeContext_1}{\TypeContext_2}
  }{
    \wtp{\TypeContext_1 + \TypeContext_2}{\Process_1 \parop \Process_2}{\ProvideContext_1, \ProvideContext_2}
  }
  \qquad    
  \inferrule[\defrule{session}]{
    \wtp{\TypeContext, \Channel^\Polarity : \SessionType, \Channel^{\co\Polarity} : \co\SessionType}\Process\ProvideContext
  }{
    \wtp\TypeContext{\new\Channel\Process}\ProvideContext
  }
  \qquad 
  \inferrule[\defrule{\newrule}]{
    \wtp{\TypeContext, \Thing : \Type}\Process{ \ProvideContext, \Thing : \Type }
  }{
    \wtp\TypeContext{\new\Thing\Process}\ProvideContext
  }
\end{array}
\end{framedmath}
}
%\\\\
%\paic{}{We need to add the condition in \refrule{thread}. Otherwise we have a problem with 
%$\nu x. (\thread x (\return x) \parop \thread y (\return~x))$}{}
\end{table}

%Without the condition $\unq \Type$, we could type  process 
%that have deadlock such as 
%$\nu \varX. (\server~\Channel~\Expression \parop 
%\thread \varX(\Bind{\open~\Channel}{\receive}))$
%where 
% $\Expression = \lambda x. \return~(\send~\varX~4)$.
%The second side condition requires 
% that the type of 
%the new thread (which will be $\Type$ if $\Expression$ has
%type $\co\SessionType\arrow\tio\Type$) are unlimited, since a server
%can be invoked an arbitrary number of times. 
%The condition  $\unq\Type$ in the typing rule \refrule{server}.
%is necessary to make sure that services use all the linear values they own. 
%We require that both the use environment and
%the type of the new thread (which will be $\Type$ if $\Expression$ has
%type $\co\SessionType\arrow\tio\Type$) are unlimited, since a server
%can be invoked an arbitrary number of times.  
%

The remaining rules are conventional.
In a parallel composition we require that the sets of entities
(threads and servers) defined by $\Process_1$ and $\Process_2$ are
disjoint.
%, namely that there is at most one thread and at most one server with a given name.
%
This is enforced by the fact that the respective resource environments
$\ProvideContext_1$ and $\ProvideContext_2$ are combined using the
operator $\_,\_$ which (as usual) implicitly requires that
$\dom(\ProvideContext_1) \cap \dom(\ProvideContext_2) = \emptyset$.
The restriction of a session channel $\Channel$ introduces
associations for both its endpoints $\Channel^+$ and $\Channel^-$ in
the typing environment with dual session types, as usual.
Finally, the restriction of a bindable name $\Thing$ 
%\eMnote{other
%  than a session name} 
  introduces associations in both the typing and
the resource environment with the same type $\Type$. This makes sure
that in $\Process$ there is exactly one definition for $\Thing$, which
can be either a variable which names a thread or a shared channel
which names a server,
%$\varX$ from a thread $\thread{\varX}{\Expression}$ or a channel $\Channel$
%from $(\server~\Channel~\Expression)$, 
and that every usage of $\Thing$
is consistent with its definition.

\subsection{Example of Type Derivation for  Processes} 
Let $\incf : \nat \arrow \nat$ be the increment function on natural numbers,
and consider
\begin{equation}
\label{equation:incList}
\begin{array}{llll }
    \mapf\ \varX & = &  & {
                \Bind {\Receive  \varX  } {}}\\
                & & & {{
                \Fun {\Pair{\varY}{\varX'}}}
                     \Bind{\future{\big(\mapf\  \varX' \big)}}{}}\\
                & & & {
                     { \Fun \varZ 
                           {\Return{\Pair{{\incf\ \varY}}{\varZ }}}}}
  \end{array}
\end{equation}
which receives natural numbers in a channel $\varX$, increments them by one
and  returns  
them in   
a stream.
Note that the function $\mapf$ in (\ref{equation:incList}) is the
function $\display$ in (\ref{eq:dspgood}) once $\ExpressionG$ is
instantiated with
$\Fun {\Pair{\var_1}{\var_2}}{\Return{\Pair{{\incf\ {\var_1}}}{\var_2
    }}}$.
Then, the process 
\begin{equation}
\label{equation:lastexampleintro}
  \thread \var \upkf\ \ChannelC^+\ (\streamf\ 0)
  \ \ \parop \ \
  \thread \varY \Bind {(\mapf\ \ \ChannelC^-)}
   {\Send \ChannelB^+ }
  \ \ \parop \ \
  \thread \varZ \receive\ \ChannelB^-
\end{equation}
sends on channel $\ChannelB$ 
% a stream of  ones.
the whole sequence of integers starting from $1$.
We show  part of a type derivation for the thread named $\varY$
in (\ref{equation:lastexampleintro}).
 \[
\inferrule*[right = \refrule{thread}]{
\inferrule*[right= \refrule{bind}]
{
\wte{\ChannelC^-:\InN}
      {\mapf\ \ \ChannelC^-}{\tio \ListN}
      %\qquad
 \\
 \wte{\ChannelB^+: \tout{\ListN}{\End}}
     {\Send \ChannelB^+ }{\tio \ListN \arrow \tio \End}     
}
{ 
\wte{\ChannelC^-:\InN, \ChannelB^+: \tout{\ListN}{\End}}
      {\Bind {(\mapf\ \ \ChannelC^-)}
   {\Send \ChannelB^+ }}{\tio \End}
   }
 }
{ \wtp{\ChannelC^-:\InN, \ChannelB^+: \tout{\ListN}{\End}}
      {\thread \varY \Bind {(\mapf\ \ \ChannelC^-)}
   {\Send \ChannelB^+ }}{\varY : \End}
 } 
   \]

%%% Local Variables:
%%% mode: latex
%%% TeX-master: "LazySessions"
%%% End:

\section{Properties of Typeable Expressions}
\label{sec:propexp}

This section is devoted to the proof of the two most relevant
properties of typeable expressions, which are subject reduction
(reduction of expressions preserves their types) and normalisation.
As informally motivated in \cref{section:typing}, the type constructor
$\bullet$ controls recursion and guarantees normalisation of any
expression that has a type different from $\infinite$. 

\subsection{Subject Reduction for Expressions}\label{subsec:sre}

The proof of subject reduction for expressions 
(Theorem \ref{theorem:subjectreduction})
is standard except for 
the fact that we are using the modal operator $\bullet$. 
For this, we need \cref{lemma:delay} below,
which says  
that 
the type of an expression should be delayed
as much as   the types in the environment.
This property reflects the fact that we can only move forward in time.
For example, 
from $\varX:  \Type \vdash \Fun \varY \varX: \TypeS \rightarrow \Type$
we can deduce that 
$\varX:  \bullet \Type \vdash 
\Fun \varY \varX: \bullet (\TypeS \rightarrow \Type)$,
but
we cannot deduce
$\varX:  \bullet \Type \vdash 
\Fun \varY \varX: \TypeS \rightarrow \Type$. 
 Notably we can derive 
 $\varX:  \bullet \Type, \varY : \bullet \Type \rightarrow \TypeS \vdash \varY \varX: \TypeS$, 
  i.e. the environment can contain types more delayed than the type of the expression. 

\begin{lemma}[Delay]
 \label{lemma:delay}
 If $\wte{\TypeContext}{\Expression}{\Type}$, then  
  $\wte{\TypeContext_1,\tbullet\TypeContext_2}{\Expression}{\tbullet\Type}$
    for $\TypeContext_1,\TypeContext_2 = \TypeContext$.
\end{lemma}
\begin{proof} By induction on the derivation. \end{proof}

The following property  tells that,
if an expression  contains an endpoint  
or a variable with a linear type, then the type
of that expression should be linear.
For example,
it is not possible to assign 
the unlimited type 
% $\nat \arrow \tout \nat \End$
$\nat \arrow \tio \End$
to the function 
$\mkconstant{lc} 
= \Fun \varX \send~\Channel^{\Polarity}~\varX$
 which  contains the free endpoint $\Channel^{\Polarity}$ of type $\tout \nat\End$. 
Otherwise, $\mkconstant{lc} $ could be erased 
in $(\Fun \varX \unit )\ \mkconstant{lc} $ 
or duplicated in $(\Fun \varX \Pair{\varX}{\varX})\ \mkconstant{lc}$.

\begin{lemma}
\label{lemma:unlimitedTypeimpliesunlimitedContext}
  If $\wte{\TypeContext}{\Expression}{\Type}$
  and $\unq \Type$, 
   then $\unq {\TypeContext}$.
\end{lemma}
\begin{proof} 
The proof  is 
by induction on the derivation of 
$\wte{\TypeContext}{\Expression}{\Type}$. 
%The first two conditions imposed in the definition of types
%(\cref{def:types}) play an important role in 
%the proof of this lemma. 
The case
of \refrule{\elimarrow} uses Condition~\ref{def:types4} 
% of \cref{def:types}
and the case of \refrule{\elimlarrow} uses Condition \ref{def:types5} 
 of \cref{def:types}. 
\end{proof} 

The following three lemmas are standard in proofs of subject reduction.

\begin{lemma}[Inversion for Expressions]\label{lem:inv}\mbox{}
  \begin{enumerate}
  \item\label{lem:inv1} If $\wte{\TypeContext}{\Constant}{\Type}$, then
    $\Type = \tbullet^n \Type'$ and $\Type'\in\typeof(\Constant)$  with $\unq{\TypeContext}$.
  \item\label{lem:inv2} If $\wte{\TypeContext}{\Name}{\Type}$, then
    $\Type = \tbullet^n \Type'$ and
    $\TypeContext = \TypeContext', \Name :\Type'$ with
    $\unq{\TypeContext'}$.
    % 
   % If $\linq{\Type}$ then $n=0$.
  \item\label{lem:inv3} If
    $\wte{\TypeContext}{\Fun \varX \Expression}{\Type}$ and
    $\unq \TypeContext$, then either $\Type = \tbullet^n (\Type_1 \arrow \Type_2)$ or
     $\Type = \tbullet[n](\Type_1 \larrow \Type_2)$
     and
    $\wte{\TypeContext, \varX:
      \tbullet^n\Type_1}{\Expression}{\tbullet^n\Type_2}$.
  \item \label{lem:inv4} If
    $\wte{\TypeContext}{\Fun \varX \Expression}{\Type}$ and
    $\linq \TypeContext$, then
    $\Type = \tbullet[n](\Type_1 \larrow \Type_2)$ and
    $\wte{\TypeContext, \varX:
      \tbullet[n]\Type_1}{\Expression}{\tbullet[n]\Type_2}$.
  \item\label{lem:inv5} If
    $\wte{\TypeContext}{\Expression_1 \Expression_2}{\Type}$, then
    $\Type = \tbullet^n\Type_2$ and
    $\TypeContext = \TypeContext_1 + \TypeContext_2$ with
    $\wte{\TypeContext_2}{\Expression_2}{\tbullet^n\Type_1}$ and either
    $\wte{\TypeContext_1}{\Expression_1}{\tbullet^n(\Type_1 \arrow
      \Type_2)}$
    or
    $\wte{\TypeContext_1}{\Expression_1}{\tbullet^n(\Type_1 \larrow
      \Type_2)}$.
    
  \item\label{lem:inv6} If
    $\wte{\TypeContext}{\Split\ExpressionE\varX\varY\ExpressionF}{\Type}$,
    then $\TypeContext = \TypeContext_1 + \TypeContext_2$ and
    $\Type = \tbullet^n\Type'$ with
    $\wte{\TypeContext_1}{\ExpressionE}{\tbullet^n(\Type_1 \times
      \Type_2)}$
    and
    $\wte{\TypeContext_2, \varX: \tbullet^n\Type_1,
      \varY:\tbullet^n\Type_2}{\ExpressionF}{\tbullet^n\Type'}$.  

  \end{enumerate}

\end{lemma}
\begin{proof} By case analysis and induction on the derivation.
We only show  \cref{lem:inv3} which is interesting because
we need to shift the environment in time and apply 
\cref{lemma:delay}.
 A  
 derivation of
    $\wte{\TypeContext}{\Fun \varX \Expression}{\Type}$
     ends with an
    application of either \refrule\introarrow, 
     \refrule\introlarrow\  or \refrule\introbullet.
    For the first two cases,   
    the proof is immediate. 
    If the last applied rule is \refrule\introbullet, then
    $\Type = \tbullet\Type'$ and we have
    \[
    \inferrule{ \wte{\TypeContext}{ \Fun \varX \Expression}\Type' }{
      \wte{\TypeContext}{\Fun \varX \Expression}{\tbullet\Type'} }
    \]
    By induction\comma $\Type' = \tbullet[n](\Type_1 \arrow \Type_2)$ or $\Type' = \tbullet^n (\Type_1 \larrow \Type_2)$ and
    $\wte{\TypeContext, \varX:
      \tbullet^n\Type_1}{\Expression}{\tbullet^n\Type_2}$.
    Hence, 
    \[\Type = \tbullet \Type' =\tbullet[n+1](\Type_1 \arrow \Type_2)\text{ or }\Type = \tbullet \Type' = \tbullet[n+1] (\Type_1 \larrow \Type_2)\]
    By \cref{lemma:delay}\comma we have that
    $\wte{\TypeContext, \varX:
      \tbullet^{n+1}\Type_1}{\Expression}{\tbullet^{n+1}\Type_2}$.
\end{proof}

% of \cref{def:types}.

\begin{lemma}[Substitution]
  \label{lemma:substitution}
  If
    $\wte{\TypeContext_1, \varX: \TypeS}{\Expression}{\Type}$ and
    $\wte{\TypeContext_2}{\ExpressionF}{\TypeS}$ and
    $\TypeContext_1 +\TypeContext_2$ is defined, then
   \[\wte{\TypeContext_1+\TypeContext_2}{\Expression\subst\ExpressionF\varX}{\Type}.\]
 \end{lemma}

\begin{proof}
  By induction on the structure of expressions.
  We only consider the case $\Expression = \Constant$, to show the 
   application 
  of  \cref{lemma:unlimitedTypeimpliesunlimitedContext}.   
    It follows from \ri{\cref{lem:inv}}{ \cref{lem:inv1}} that 
    $\unq{\TypeContext_1, \varX : \TypeS}$ and
    $\Type = \tbullet[n]\Type'$ with $\Type' = \typeof(\Constant) $. 
    From $\unq \TypeS$, $\wte{\TypeContext_2}{\ExpressionF}{\TypeS}$
    and  \cref{lemma:unlimitedTypeimpliesunlimitedContext},
     we derive
    $\unq{\TypeContext_2}$, and therefore
    $\wte{\TypeContext_1+\TypeContext_2}{\Constant}{\Type}$ by
    \refrule\const\ and $\Constant\subst\ExpressionF\varX = \Constant$.
  \end{proof}

\begin{lemma}[Evaluation Contexts for Expressions]
  \label{lem:keye}
  If 
  $\wte{\TypeContext}{\Context[\Expression]}{\Type}$,
then
 $\TypeContext = \TypeContext_1 + \TypeContext_2$ and
$\wte{\TypeContext_1,\varX:\TypeS}{\Context[\varX]}{\Type}$ and 
   $\wte{\TypeContext_2}{\Expression}{\TypeS}$ for some $\TypeS$.
 \end{lemma}
\begin{proof}
  By induction on the structure of $\Context$.
  \end{proof}

  \begin{theorem}[Subject Reduction for Expressions]
  \label{theorem:subjectreduction}
   If $\wte{\TypeContext}{\Expression}{\Type}$
     and $\Expression \red \Expression'$, then
     \mbox{$\wte{\TypeContext}{\Expression'}{\Type}$.}
\end{theorem}
 
 \begin{proof}
By induction on the definition of $\red$. 
 \cref{lem:keye} is useful for rule \defrule{r-ctxt}.
We only consider the case 
$(\Fun\varX\ExpressionE)~\ExpressionF
      \red    
      \ExpressionE \subst\ExpressionF\varX$.
    Suppose 
    $
    \wte{\TypeContext}{(\Fun\varX\ExpressionE)~\ExpressionF}{\Type}
    $.
    By \ri{\cref{lem:inv}}{\cref{lem:inv5}}\comma
    $\Type=\tbullet^n\Type_2$ and
    $\TypeContext=\TypeContext_1 +\TypeContext_2 $ and
%    \begin{equation}
%      \label{eq:argument}
\[
      \wte{\TypeContext_2}{\ExpressionF}{\tbullet^n\Type_1}
      \quad \text{and either}\quad
      \wte{\TypeContext_1}{\Fun\varX\ExpressionE}{\tbullet^n(\Type_1 \arrow \Type_2)}
      \quad\text{or}\quad
      \wte{\TypeContext_1}{\Fun\varX\ExpressionE}{\tbullet^n(\Type_1 \larrow \Type_2)}
   \]
   % \end{equation}
    In both cases, 
    it follows from \ri{\cref{lem:inv}}{\cref{lem:inv3}} that
    \begin{equation}
      \label{eq:function}
      \wte{\TypeContext_1, \varX: \tbullet^n\Type_1}{\Expression}{\tbullet^n\Type_2}
    \end{equation}
    By applying \cref{lemma:substitution} to \eqref{eq:function}\comma
     we get
    $\wte{\TypeContext}{\Expression\subst\ExpressionF\varX}{\tbullet^n\Type_2}$.   
    \end{proof}

\subsection{Normalisation of Expressions}\label{subsec:norm}

In this section we prove that any typeable  expression 
whose  type is  different from  $\infinite$ 
reduces to a normal form  (Theorem \ref{theorem:weakheadnormalization}).
For this, we define  a type interpretation indexed on the
set of natural numbers for dealing with the temporal operator
$\bullet$. 
The time is discrete and   represented using the set
of natural numbers.
The semantics reflects the fact that one 
$\bullet$ corresponds to one unit of time
 by shifting
the interpretation from $i$ to $i+1$. 
A similar interpretation of the modal operator with
indexed sets  is given in \cite{Nakano00:lics}.
%A more general setting based on  category theory
%is presented in 
%\cite{LMCSBirkedaletal}.
%%\cite{BirkedalMSS11:lics}. 
%
For simplicity we consider only $\Unit$ as basic type, the addition of other basic types is easy. 

%Normalisation is needed also for open expressions 
%to prove progress of processes.
Before introducing %defining 
the type interpretation, we give a few definitions.  
Let  $\EE$ be the set of expressions.
%
%Hereafter, we write $\red^*$ for the reflexive, transitive closure
%of $\red$.
% 
We define the following subsets of $\EE$:
\[
\begin{array}{r@{~}l}
 \setWN & =
\{ \Expression \mid \Expression \red^{*} \ExpressionF\ \&  \ 
\ExpressionF \mbox{ is a normal form} \}\\
\setNVAR & = 
\{ \Expression \mid \Expression \red^{*} \Context[\varX] \ \&  \ \varX \mbox{ is a variable} \}\\
\setNIO & = 
\{ \Expression \mid \Expression \red^{*} \PContext[\Expression_0] \ \&  \ 
\Expression_0 \in \{
\Send{\Channel^{\Polarity}}{\Expression_1},
          \Receive{\Channel^{\Polarity}},
          \Open{\Channel},
\Future \Expression_1 \}
\}
\end{array}
\]
The sets $\setNVAR$ and $\setNIO$ are sets of expressions which
reduce to normal forms of particular shapes. They are disjoint and
both subsets of $\setWN$.
We will do induction on the rank of types.  For $\Unit$, session
types, and shared channel types the rank is always 0.  For the other
types, the rank measures the depth of all what we can observe at time
$0$. We could also compute it by taking the maximal $0$-length of all
the paths in the tree representation of the type, where the $0$-length
of a path is the number of type constructors different from $\bullet$
from the root to a leaf or to a $\bullet$.

\begin{definition}[Rank of a Type]
\label{definition:typerank}
The rank  
of a type $\Type$ (notation $\typerank (\Type)$) is defined as
follows.
\[
\begin{array}{r@{~}ll}
\typerank (\Unit) & = 
\typerank(\SessionType) =
\typerank(\tshared\SessionType ) = 
\typerank (\bullet \Type)  = 0 \\
\typerank (\tio  \Type) & = \typerank ( \Type) +1 \\
\typerank (\Type \times \TypeS) & = max(\typerank ( \Type), \typerank(\TypeS)) +1 \\
\typerank (\Type \rightarrow \TypeS) & = max(\typerank ( \Type), \typerank(\TypeS)) +1 \\
\typerank (\Type \larrow \TypeS) & = max(\typerank ( \Type), \typerank(\TypeS)) +1 \\
\end{array}
\]
\end{definition}

The rank is well defined (and finite) because 
 the tree representation of
a type cannot have an infinite branch with 
no $\bullet$'s at all
(Condition~\ref{def:types3}
 in Definition \ref{def:types})
 and $\typerank (\bullet \Type)$ is set to $ 0$.
 
We now define the type interpretation
$\ti \Type  \in \natset \rightarrow \mathcal{P} (\EE)$, which is
an indexed set, where $\natset$ is the set of natural numbers and $ \mathcal{P}$ is the powerset constructor.
%A similar interpretation of the modal operator with
%indexed sets  is given in \cite{Nakano00:lics}.
%A more general setting based on  category theory
%is presented in \cite{BirkedalMSS11:lics}. 

%\paic{}{PROBLEM!!!}{
%In this section, we assume  that  $\infinite$ is a type, 
%unless stated explicitly.
%This assumption is needed  to interpret 
% types such as $\tio \infinite$ or $\infinite \arrow \infinite$
% and it is also needed 
% in the proofs.
%}
%\begin{table}
%\caption{Type interpretation.}\label{tab:ti}
%\framebox[\textwidth]{
%\begin{framedmath}
\begin{definition}[Type Interpretation]
\label{definition:firsttypeinterpretation}
We define 
$ \indti \Type \ind \subseteq \EE$ by induction on
 $(\ind, \typerank(\Type))$.
\[
\begin{array}{rcl}
\indti\Unit \ind & = & \WNVAR \cup  
\{ \Expression \mid \Expression \red^* \unit \} \\[3pt]
%%%%%%%%%%%%%%%%%%%%%%%%%%
\indti\SessionType  \ind & = &  \WNVAR \cup
 \{ \Expression  \mid \Expression \red^* \Channel^{\Polarity} \} \\[3pt]
 %%%%%%%%%%%%%%%%%%%%
\indti{\tshared\SessionType} \ind
 & = & \WNVAR \cup \{ \Expression  \mid \Expression \red^* \Channel \} \\[3pt]
 %%%%%%%%%%%%%%%%%%%%%%%
 \indti{ \Type \times \TypeS} \ind & = &  
 \WNVAR \cup    \set{\Expression \mid 
       \Expression \red^* \Pair{\Expression_1}{\Expression_2}   \mbox{ and }
      \Expression_1\in\indti\Type \ind \mbox{ and } \ ~\Expression_2\in\indti\TypeS \ind}\\[3pt]
 %%%%%%%%%%%%%%%%%%%%%%%%%%%%
\indti{\Type \arrow \TypeS}{\ind} & =  &  
 \indti{\Type \larrow \TypeS}{\ind}  \\
 & = & \WNVAR \cup 
 \set{\Expression \mid 
 \Expression \red^{*} \lambda x. 
 \ExpressionF \mbox{ and }
 \Expression\Expression' \in \indti{\TypeS}{\indj}  \ \ 
 \forall \Expression' \in \indti{\Type}{\indj}, \indj \leq \ind}\\
 &  & {} \cup \set{\Expression \mid 
 \Expression \red^{*} \Context[\Constant] \mbox{ and }
 \Expression\Expression' \in \indti{\TypeS}{\indj}  \ \ 
 \forall \Expression' \in \indti{\Type}{\indj}, \indj \leq \ind}
  \\[3pt]
 \indti{\tio\Type}{\ind} & = & \WNVAR \cup    \setNIO \cup
                        \{ \Expression \mid 
                            \Expression  \red^* \return~\Expression'\mbox{ and } 
                            \Expression' \in \indti{\Type}{\ind}    \}
                             \\[3pt]   
   %%%%%%%%%%%%%% %%%%%%%%%%%%%%              
     \indti{\tbullet\Type}{0} & = &     \EE  \\[3pt]
\indti{\tbullet\Type}{\ind+1}  & = & 
           \indti{\Type}{\ind}                
\end{array}
\]
%\end{framedmath}}
%\end{table}
\end{definition} 
\noindent
Note that $\indti{\infinite}{\ind} = \EE $ for all $\ind \in \natset$.
In the interpretation of the arrow type,  
 the requirement  ``for all  $j \leq i$'' (and not just 
 ``for all $i$'') is crucial
for dealing with the contra-variance of the arrow type in the proof of \ri{\cref{lem:A}}{\cref{lemma:monotonicityofinterpretation}}.
%solving the problem of  the contra-variance.

The next properties of the type interpretation are expected.
 
\begin{lemma} \mbox{} %The following properties hold:
\label{lemma:interpretationofmanybullets}
\begin{enumerate}
\item\label{lemma:interpretationofmanybullets1}  $\indti{\tbullet[n] \Type}{\ind} = \EE$ if $\ind < n$.
\item\label{lemma:interpretationofmanybullets2}  $\indti{\tbullet[n] \Type}{\ind}= \indti{\Type}{\ind -n}$ if $\ind \geq  n$.
\end{enumerate}
\end{lemma}
\begin{proof}
Both items are proved by induction on $n$.
\end{proof}

\begin{lemma}\label{lem:A}
\mbox{} %The following properties hold:
\begin{enumerate}
\item\label{lemma:wnvar}
 For all types $\Type$ and $\ind \in \natset$, we have 
$\WNVAR \subseteq \indti{\Type}{\ind}$. 
\item\label{lem:A2}
If 
$\Type\not=\bullet\TypeS$, then 
$
\indti{\bullet^{n+1} \Type}{n+1} \subseteq  \HH  
$.

\item 
\label{lemma:monotonicityofinterpretation}
For all $\ind \in \mathbb{N}$, 
$\indti{\Type}{\ind+1} \subseteq \indti{\Type}{\ind}$.

\item \label{lemma:interpretationoftypes}
 If $\Type \not = \infinite$, then
 $\bigcap_{\ind \in \natset} \indti{\Type}{\ind} \subseteq \setWN
$.

\end{enumerate}
\end{lemma}
\begin{proof} (\cref{lemma:wnvar}).
By 
 induction on $\ind$ and doing case analysis on 
 the shape of the type. All cases are trivial except
 when the type is $\bullet \Type$.
 
(\cref{lem:A2}). Using \ri{\cref{lemma:interpretationofmanybullets}}{\cref{lemma:interpretationofmanybullets2}}. 
%By induction on $n$.

(\cref{lemma:monotonicityofinterpretation}).  By induction on
 $(\ind, \typerank(\Type))$. 
% The case of the arrow type is interesting because
% it shows how the   problem of the contra-variance
% is solved.
 Suppose $\Expression \in  \indti{\Type \arrow  \TypeS}{i+1}$.
 Then\comma
$\Expression \Expression' \in \indti{\TypeS}{j}$ 
 for $j \leq i+1$.
 This is equivalent to saying that
 $\Expression \Expression' \in \indti{\TypeS}{j'+1}$
 for $j' \leq i$.
 By induction hypothesis\comma
$\indti{\TypeS}{j'+1} \subseteq \indti{\TypeS}{j'}$.
Hence, 
 $\Expression \in  \indti{\Type \arrow  \TypeS}{i}$.
The remaining cases are easy.

(\cref{lemma:interpretationoftypes}). 
 All the cases are trivial except for a type starting by $\bullet$.
Since $\Type \not = \infinite$, 
we have that $\Type = \bullet^{n+1} \TypeS$
and $\TypeS\not= \bullet \TypeS'$.
It follows from \cref{lem:A2}
that 
$\indti{\bullet^{n+1} \TypeS}{n+1} \subseteq \HH $ and  hence 
%from \cref{lemma:monotonicityofinterpretation}
%that 
%$\bigcap_{\ind \in \natset} \indti{\Type}{\ind}  \subseteq 
% \indti{\Type}{n+1} $. Hence, 
\begin{equation}
 \bigcap_{\ind \in \natset} \indti{\Type}{\ind}  \subseteq 
 \indti{\Type}{n+1} = \indti{\bullet^{n+1} \TypeS}{n+1} \subseteq   \HH.
 \tag*{\qEd}
\end{equation}
\def\popQED{}
 \end{proof}

\bigskip

In order to deal with open expressions we resort
to  substitution functions, as usual.
A substitution function is
a mapping from (a finite set of) variables to $\EE$. We use $\funsubst$ to range over substitution functions.
Substitution functions allows us to  
 extend the semantics to typing %type 
judgements 
% define indexed semantic typing 
(notation $\TypeContext \modelsi \Expression: \Type$). 

\begin{definition}[Typing  %Type 
Judgement Interpretation]
\label{definition:modelsi}
Let $\funsubst$ be a substitution function. 

\begin{enumerate} 

\item $\funsubst \modelsi \TypeContext$
if $\funsubst(\varX) \in 
  \indti{\Type}{\ind}$
 for all $\varX : \Type \in \TypeContext$.
 
\item $\TypeContext \modelsi \Expression: \Type$
if $\funsubst (\Expression) \in  \indti{\Type}{\ind}$
for all  $\funsubst \modelsi \TypeContext$.
\end{enumerate}
\end{definition}

As expected we can show the soundness of our type system 
with respect to 
the indexed semantics. 

\begin{theorem}[Soundness]
\label{theorem:soundness} 
If $\wte{\TypeContext}\Expression\Type$,
then $\TypeContext \modelsi \Expression:\Type$ for all $\ind \in \mathbb{N}$.
\end{theorem}

%We prove that $\TypeContext \modelsi \Expression:\Type$ for all $\ind \in \mathbb{N}$
%by induction on  $\wte{\TypeContext}\Expression\Type$.
%The proof can be found in~\cref{appendix:normalisation}.

The proof of this theorem by induction on  $\wte{\TypeContext}\Expression\Type$ can be found in~\cref{appendix:normalisation}.

\begin{theorem}[Normalisation of Typeable Expressions] 
\label{theorem:weakheadnormalization}
If $\wte\TypeContext\Expression\Type$ and $\Type\not=\infinite$, then
$\Expression$ reduces (in zero or more steps) to a normal form.
% Every typeable expression reduces to a normal form. 
\end{theorem}
\begin{proof}
It follows from \cref{theorem:soundness} that
\begin{equation}
\label{equation:modelsi}
\TypeContext \modelsi \Expression : \Type
\end{equation}
for all $\ind \in \mathbb{N}$. 
Let $id$ be the identity substitution and
 suppose $\varX: \TypeS \in \TypeContext$. Then
\[
\begin{array}{lll}
id(\varX) = \varX & \in \WNVAR \\
      &  \subseteq \indti{\TypeS}{\ind} & \mbox{by \ri{\cref{lem:A}}{\cref{lemma:wnvar}}.}
      \end{array}
      \]
This means that $id \modelsi \TypeContext$
for all $\ind \in \mathbb{N}$.
From (\ref{equation:modelsi}) we have that
$id(\Expression) = \Expression \in \indti{\Type}{\ind} $ for all $\ind$.
Hence,
\[
\Expression \in \bigcap_{\ind \in \natset} \indti{\Type}{\ind}\]
It follows from \ri{\cref{lem:A}}{\cref{lemma:interpretationoftypes}} that
 $\Expression  \in\HH$.
\end{proof}
 Notice that there are normalising expressions that cannot be typed, for example $\lambda x. \omegaterm {\bf I}$, where 
  $\omegaterm$ is defined at the end of \cref{subsec:types}  %$\omegaterm = (\lambda \varY. \varY\ \varY)(\lambda \varY. \varY\ \varY)$ 
 and ${\bf I}=\lambda z.z$. In fact $\omegaterm$ has type $\infinite$ and by previous theorem it cannot have other types, and this implies that the application $\omegaterm {\bf I}$ has no type.

\section{Properties of Reachable Processes}
\label{sec:propproc}

In general, processes lack subject reduction. For example, the process
\begin{equation}
\label{eq:bad1}
\new{\varX\varY}(\thread\var\return~\varY\parop\thread\varY\return~\var)
\end{equation}
is well-typed by assigning both $\varX$ and $\varY$ any unlimited
type, but its reduct
\[
\new \var(\thread\var\return~\var)
\]
 is ill-typed
because the thread name $\var$ occurs free in its body (\cf the side
condition of rule \refrule{thread}).
Another paradigmatic example is
%
%\[\wtp{\Channel^+:\tout\Type\End,\Channel^-:\tin\Type\End}{\thread\var\send~\Channel^+\varY\parop\thread\varY\receive~\Channel^-}{\var:\End,\varY:\Type}\]
%
\begin{equation}
\label{eq:bad2}
\thread\var\send~\Channel^+\varY\parop\thread\varY\receive~\Channel^-
\end{equation}
which is well-typed in the environment
$\Channel^+:\tout\Type\End,\Channel^-:\tin\Type\End$ where
$\Type=\bullet(\Type\times\End)$ and which reduces to
$\thread\var\return~\Channel^+\parop\thread\varY\return~\Pair\varY{\Channel^-}$.
Again, the reduct is ill-typed because the thread name $\varY$ occurs
free in its body.
In general, these examples show that the reduction 
rules \refrule{r-return}
and \refrule{r-comm} can violate the side condition of 
 the typing rule 
\refrule{thread}, which requires 
that a future variable is never defined in
terms of itself.

Another source of problems is the fact that, as in many session
calculi~\cite{BCDLDY08,CDYP2016}, there exist well-typed processes
that are (or reduce to) configurations where mutual dependencies
between sessions and/or thread names prevent progress.
For instance, both
\begin{eqnarray}
  \label{eq:bad3}
  & \new{\varX\varY\ChannelA\ChannelB}
  (\thread\varX\Bind{\send~\Channel^+~4}{\lambda\var.\receive~\ChannelB^-}
  \parop
  \thread\varY\Bind{\send~\ChannelB^+~2}{\lambda\var.\receive~\Channel^-})
  \qquad
  \\
  \label{eq:bad4}
  % & \new{\varX\varY\Channel}
  % (\thread\varX{\varY(\receive~\Channel^-)}
  % \parop  
  % \thread\varY\Bind{\send~\Channel^+~4}{\lambda\varZ.\return~(\lambda\varV.\varV)})
  & \new{\varX\Channel}
    (\thread\varX{\Bind{\receive~\Channel^-}{\lambda\Pair\varY\varZ.\send~\Channel^+~\varY}})
\end{eqnarray}
are well-typed but also deadlocked.

The point is that none of the troublesome processes (including those
shown above) is relevant to us, because they cannot be obtained by
reducing a so-called \emph{initial} process modelling the beginning of
a computation. A closed, well-typed process $\ProcessP$ is
\emph{initial} if
\[
\begin{array}{lll}
\ProcessP \equiv &
\new{\varX \Channel_1 \cdots \Channel_m}
(
\thread{\varX}{\Expression} \parop
\server{\Channel_1}{\Expression_1} \parop 
\cdots
\parop
\server{\Channel_m}{\Expression_m})
\end{array}
\]
namely if it refers to no undefined names and if it consists of one
thread $\var$ -- usually called ``main'' in most programming languages
-- and an arbitrary number of servers that are necessary for the
computation. In particular, typeability guarantees that all bodies reduce
to normal forms and all $\open$'s refer to existing
servers.
Clearly, an initial process is typeable in the empty environment.

We call \emph{reachable} all processes that can be obtained by
reducing an initial process. A reachable process may have several
threads running in parallel, resulting from either service invocation
or $\future$'s.  

 This section is organised as follows. \cref{subsec:wpp} defines the set of well-polarised processes, which includes the set of reachable processes. Subject reduction for reachable processes then follows from subject reduction for well-polarised processes (\cref{subsec:srp}).  Well-polarisation of reachable processes is also used in~\cref{subsec:prog}
to show progress and in~\cref{subsec:conf}  to show confluence. 

\subsection{Well-polarised Processes}
\label{subsec:wpp}

The most original and critical aspect of the following proofs is to
check that reachable processes do not have circular dependencies on
session channels and variables.  The absence of circularities can be
properly formalised by means of a judgement that characterises the
sharing of names among threads, inspired by the typing of the parallel
composition given in~\cite{LindleyMorris15}.  Intuitively, the notion
of {\em well-polarisation} captures the following properties of
reachable processes and makes them suitable for proving subject
reduction, progress and confluence:

\begin{enumerate}
%\item\label{c1} the name of a thread never occurs in its body;
\item\label{c3} two threads can share at most one session channel;
\item\label{c2} distinct endpoints of a session channel always occur
  in different threads;
\item\label{c4} if the name of one thread occurs in the body of
  another thread, then these threads cannot share session channels nor
  can the first thread mention the second.
\end{enumerate}
Note that \eqref{eq:bad1} and~\eqref{eq:bad2} violate
condition~(\ref{c4}), \eqref{eq:bad3} violates condition~(\ref{c3}),
and \eqref{eq:bad4} violates condition~(\ref{c2}).
In order to define well-polarised processes, we need a few auxiliary
notions. To begin with, we define 
 functions to extract bounds, threads and servers from processes.

\begin{definition}[Bounds, Threads, Servers]
  We define  
  \[
   \boundnames(\ProcessP) = 
   \{ \BindableName_1, \ldots, \BindableName_n \} \ \ \ \
   \threads(\ProcessP) = \ProcessQ \ \ \ \ 
 \servers(\ProcessP) = \ProcessR
 \]   
 if
 % $\Process\equiv\new{\BindableName_1\cdots\BindableName_n}(\ProcessQ\parop\ProcessR)$
 % and $\ProcessQ\parop\ProcessR$ does not contain restrictions,
 % $\ProcessQ$ is a parallel composition of threads), $\ProcessR$ is a
 % parallel composition of servers.
   $\Process\equiv\new{\BindableName_1 \cdots \BindableName_n}(\ProcessQ\parop\ProcessR)$,
   $\ProcessQ\parop\ProcessR$ does not contain restrictions,
   $\ProcessQ$ is \emph{thread-only} (namely, it is a parallel
 composition of threads), and $\ProcessR$ is \emph{server-only}
 (namely, it is a parallel composition of servers).
\end{definition}
\noindent
%Note that $\threads(\ProcessP)$ is uniquely defined up to structural
%congruence and the bindable names 
% $\BindableName_1,\ldots,\BindableName_n$'. 
%\bchpaula  Note that 
% if $\Process \equiv \ProcessQ$ and 
% $\boundnames(\ProcessP) = \{ \BindableName_1, \ldots, \BindableName_n \}$
% and 
% $\boundnames(\ProcessQ) = \{ \BindableNameY_1, \ldots, \BindableNameY_n \}$
% then
% $\new{\BindableName_1\cdots\BindableName_n} \threads(\ProcessP) \equiv 
% \new{\BindableNameY_1\cdots\BindableNameY_n} \threads(\ProcessQ)$.
% \echpaula
 
Next, we define a
mapping that computes the set of polarised names occurring free in an
expression or parallel composition of threads.

\begin{definition}[Polarised Names]
Let $\Pol$ be  defined on expressions and thread-only processes by:
\[
  \begin{array}{rcl}
    \Pol(\Expression)
    &=& \set{\Channel^\Polarity\mid \Channel^\Polarity\in\fn(\Expression)}\cup\set{\var^+\mid ~\var\in\fn(\Expression)}\\
    \Pol(\thread\var\Expression)
    &=& \set{\var^-}\cup\Pol (\Expression)\\
    \Pol(\ProcessP\parop\ProcessQ)
    &=& \Pol (\ProcessP)\cup \Pol (\ProcessQ)
  \end{array}
\] 
\end{definition}

Let $\CAA, \CBB$ be %We use $\CAA, \CBB$ to range over 
sets of polarised variables and
endpoints.  We say that $\CAA$ and $\CBB$ are \emph{independent},
notation $\CAA \indep \CBB$, if for every $\Thing^\PolarityP \in \CAA$ and
$\Thing^\PolarityQ \in \CBB$ we have $\PolarityP = \PolarityQ$. Then
$\CAA \indep \CAA$ implies that $\CAA$ cannot contain the same name with
opposite polarities.
  
\begin{definition}[Well-polarised Processes]
  \label{definition:wpp}
  Let $\CA \models \Process$ be the least 
  predicate on thread-only
    processes such that
\[
  \begin{array}{c}
    \inferrule[\defrule{wp-empty}]{
      \mathstrut
    }{
      \models \emptyprocess
    }
    \qquad
    \inferrule[\defrule{wp-thread}]{
      \mathstrut
    }{ 
      \models \thread\var\Expression
    }
    ~~
    \begin{lines}[c]
      \var^+\not\in\Pol(\Expression)
      \\
      \Pol(\Expression) \indep \Pol(\Expression)
    \end{lines}
    \qquad
    \inferrule[\defrule{wp-par}]{
      \CA \models \ProcessP
      \\
      \CB \models \ProcessQ
    }{
      \models\ProcessP\parop\ProcessQ
    }
    ~~
    \Pol (\ProcessP)\setminus\set{\Thing^\Polarity}\indep \Pol (\ProcessQ)\setminus\set{\Thing^{\co\Polarity}}
  \end{array}
\]
We say that $\ProcessP$ is \emph{well-polarised} if
$\CA \models \ProcessQ$ for some $\ProcessQ\equiv\threads(\Process)$.
\end{definition}

Note that the variable $\Thing$ in \defrule{wp-par} is existentially
quantified.
The empty process is trivially well-polarised and a thread
$\thread\var\Expression$ is well-polarised if $\Expression$ does not
contain references to both $\Channel^\Polarity$ and
$\Channel^{\co\Polarity}$,
  %$\BindableName^{+}$ and $\BindableName^{-}$, 
  nor to the thread name $\var$. %variable $\var$ being defined by the thread.
  A parallel composition $\ProcessP \parop \ProcessQ$ is
  well-polarised if there is
%
%As formalised in \ri{\cref{lem:Luca}}{\cref{lem:Luca2ii}}, the
%application of rule \defrule{wp-par} requires that
at most one variable or endpoint that occurs with opposite polarities
in $\Pol (\ProcessP)$ and $\Pol (\ProcessQ)$. 
%This means that either
%$\ProcessP$ contains $\Channel^\Polarity$ and $\ProcessQ$ contains
%$\Channel^{\co\Polarity}$ or $\ProcessP$ has a thread labelled $\var$
%and $\ProcessQ$ has a thread whose body contains $\var$. In particular
%we can have $\Pol (\ProcessP)\indep\Pol (\ProcessQ)$, i.e.
%$\Pol (\ProcessP)$ and $\Pol (\ProcessQ)$ do not contain the same name
%with opposite polarities.
 This means that  either: 
\begin{enumerate}
\item 
$\ProcessP$ contains $\Channel^\Polarity$ and $\ProcessQ$ contains
$\Channel^{\co\Polarity}$,
\item  $\ProcessP$ has a thread labelled $\var$
and $\ProcessQ$ has a thread whose body contains $\var$ (or vice versa),
\item  $\Pol (\ProcessP)\indep\Pol (\ProcessQ)$, i.e.
$\Pol (\ProcessP)$ and $\Pol (\ProcessQ)$ do not share 
 names with opposite polarities. 
\end{enumerate}

Note that $\CA \models \Process$ can hold even if $\Process$ cannot be
typed, for example $\CA \models \thread \var {\omegaterm {\bf I}}$, see the end of \cref{subsec:norm}. 
 Well-polarisation cannot be incorporated into the typing rules
because it is not closed under structural equivalence, i.e.
$\CA \models \Process$ and $\Process \equiv \Process'$ do not imply
$\CA \models \Process'$.  A counterexample is shown below:
\[
\begin{array}{r@{~}l}
\Process & = 
(\thread {\varX_1}{\return~1} \parop
 \thread {\varY_2}{\return~\varX_1}) \parop (
 \thread {\varX_2}{ \return~\varY_1} \parop 
 \thread{\varY_1}{\return~2} ) 
 \\
\Process' & = (\thread {\varX_1}{\return~1}
\parop  \thread {\varX_2}{ \return~\varY_1})
\parop  
(\thread {\varY_1}{\return~2} \parop 
\thread{\varY_2}{\return~\varX_1})
\end{array}
\]
% Another example that uses several occurrences of the same variable
% shows that not even associativity holds:
% \[
% \begin{array}{ll}
% \Process & = 
% (\thread {\varX}{\return~\Pair{\varY}{\varZ}} \parop
%  \thread {\varY}{\return~\varZ}) \parop
%  \thread {\varZ}{ \return~1}   
% \\
% \Process' & = 
% \thread {\varX}{\return~\Pair{\varY}{\varZ}} \parop
% ( \thread {\varY}{\return~\varZ} \parop
%  \thread {\varZ}{ \return~1}  )

% \end{array}
% \]

\begin{definition}
   We  write 
  $\Process \subprocesseq \ProcessQ$ if
  $\Process \in \setSubProc(Q)$, where $\setSubProc$ is defined by
\[\begin{array}{c}
\setSubProc(\emptyprocess) = \set\emptyprocess
\qquad
\setSubProc(\thread \varX \Expression) = 
\{ \thread \varX \Expression\}
\\[0.5ex]
\setSubProc (\Process_1 \parop \Process_2) = 
\setSubProc(\Process_1) \cup \setSubProc(\Process_2)
    \cup \{ \ProcessP'_1 \parop \Process'_2 \mid \Process'_1 \in  \setSubProc(\Process_1)  \text{ and } \Process'_2 \in  \setSubProc(\Process_2)\}
\end{array}
\]
We write $\Process \subprocess \ProcessQ$ if
$\Process \subprocesseq \ProcessQ$ and $\Process \not = \ProcessQ$.
\end{definition}
\noindent
Note that, if $\Process \subprocess \ProcessQ$, then all threads of
$\Process$ respect the syntactic structure of $\ProcessQ$.
This means that if we represent the processes as trees,
$\Process$ is a sub-tree of $\ProcessQ$.   This is
important because $\models$ is not necessarily preserved by structural
equivalence.
% and the parenthesis of $\Process$ should be in the same
% position as they are in $\ProcessQ$.

\begin{lemma}
\label{lem:Luca3} 
If $\CA \models \Process$ and $\ProcessQ \subprocesseq  \Process$,
then 
$\CB \models \ProcessQ$.
\end{lemma}
 
\begin{proof} By induction on the derivation of
  $\CA \models \Process$.
\end{proof}

The proof that well-polarisation of typeable processes is preserved by
reductions (Theorem \ref{theorem:invariance}) is a bit involved
because $\Process_0 \red\ProcessQ_0$ and $\CA \models \Process$ with
$\Process\equiv \threads(\Process_0)$ do not 
imply that $\CA \models \ProcessQ$ for an arbitrary 
 $\ProcessQ \equiv \threads(\ProcessQ_0)$.  We
will prove a variant of the above property: 
%, which is 
if $\Process_0 \red\ProcessQ_0$ and $\CA \models \Process$ with
$\Process\equiv \threads(\Process_0)$, then there exists $\ProcessQ'$
with $\ProcessQ' \equiv \threads(\ProcessQ_0)$ such that
$\CA \models \ProcessQ'$.  The problem lies on the reduction rules
\refrule{r-comm} and \refrule{r-return}.

\begin{example}
\label{example:r-comm}
This example shows that using the rule \refrule{r-comm}
we can obtain $\ProcessQ$ from $\Process$ such that 
$\CA  \models \Process$, but $\CA \not \models \ProcessQ$. 
\[
\begin{array}{r@{~}l}
\Process & =
(\thread \varX \send~\Channel^{+}~z \parop \thread \varZ \return~1 )\parop 
 \thread \varY \receive~\Channel^{-} 
 \\
 \ProcessQ & =
(\thread \varX \return~\Channel^{+} \parop \thread \varZ \return~1 )\parop 
 \thread \varY \return~\Pair{z}{\Channel^{-}}
 \end{array}
 \]
 By re-arranging the threads of $\ProcessQ$\comma we get 
 a process $\ProcessQ'$ 
 such that $\CA \models \ProcessQ'$:
\[
\ProcessQ' =
(\thread \varX \return~\Channel^{+} \parop\thread \varY \return~\Pair{z}{\Channel^{-}}) \parop
\thread \varZ \return~1
\]
\end{example}

The rule \refrule{r-return} has a similar problem 
 as illustrated
by the following example.

\begin{example} 
\label{example:r-return}
Take 
% for example,   
$\Process_0 = \new \varX \Process$  and 
\[
\begin{array}{lll}
 \Process  & =&
((\thread \varX \return~\Pair{\varZ_1}{\varZ_2}
 \parop \thread {\varZ_1}  \return~\varZ_2) \parop
 \thread {\varZ_2} \return~1) \parop \\
 & &
 (\thread \varY \send~\Channel^{+} \varX \parop
 \thread u \receive~\Channel^{-})
 \\
 \ProcessQ & = &
 (\thread {\varZ_1}  \return~\varZ_2 \parop
 \thread {\varZ_2} \return~1) \parop  \\
 & & 
 (\thread \varY \send~\Channel^{+}~\Pair{\varZ_1}{\varZ_2}  \parop
 \thread u \receive~\Channel^{-})
 \end{array}
 \]
Then\comma $\CA  \models \Process$ but $\CA \not \models \ProcessQ$. 
We have that $\CA  \models \ProcessQ'$
and $\ProcessQ \equiv \ProcessQ'$ where 
 \[
\ProcessQ' =
((\thread \varY \send~\Channel^{+}~\Pair{\varZ_1}{\varZ_2}
\parop
  \thread u \receive~\Channel^{-})
\parop \thread {\varZ_1}  \return~\varZ_2) \parop
 \thread {\varZ_2} \return~1
\]  
\end{example}

The details for finding a 
 $\ProcessQ' \equiv \threads(\ProcessQ_0)$ 
for any $\ProcessQ_0$ such that $\Process_0 \red\ProcessQ_0$ 
and  $\CA \models \Process$ with 
$\Process\equiv \threads(\Process_0)$  
  are given in~\cref{appendix:SRprocesses}.  Here
we only give the formalisation of the properties (\ref{c3}), (\ref{c2}) and (\ref{c4}) listed at the beginning of this  section. 

\begin{lemma}\label{lem:Luca}
\mbox{} %The following properties hold:
\begin{enumerate}
\item  \label{lem:Luca2} 
Let  $\CA\models\Process\parop\ProcessQ$ and either 
 $\Thing\not=\BindableNameY$ or $\Polarity \not = \PolarityQ$. 
If  $\Thing^\Polarity,\BindableNameY^\PolarityQ\in \Pol(\Process)$
and $\Thing^{\co\Polarity} \in \Pol(\ProcessQ)$, 
then $ \BindableNameY^{\co\PolarityQ} \not \in  \Pol(\ProcessQ)$.
Similarly,
if  \mbox{$\Thing^\Polarity \in\Pol(\Process)$}
and $\Thing^{\co\Polarity}, \BindableNameY^\PolarityQ \in \Pol(\ProcessQ)$,
then $ \BindableNameY^{\co\PolarityQ} \not \in \Pol(\Process)$.

 \item\label{lem:Luca4} 
If $\CA \models \Process$ and $\thread\var\Expression\subprocess\Process$ and $ \thread\varY\ExpressionF\subprocess\Process$ and $\var$ occurs in $\ExpressionF$, then $\varY$ cannot occur in $\Expression$.

 \item\label{lem:Luca5} 
If $\CA \models \Process$ and $\thread\var\Expression\subprocess\Process$ and $ \thread\varY\ExpressionF\subprocess\Process$ 
and $\Channel^\Polarity$ occurs in $\Expression$ and
$\Channel^{\co \Polarity}$ occurs in $\ExpressionF$, 
then $\varY$ cannot occur in $\Expression$.
 \end{enumerate}
\end{lemma}
\begin{proof}
\cref{lem:Luca2}
% A derivation of $\CA\models\Process\parop\ProcessQ$ ends 
% by the application of rule $\refrule{wp-par}$.
%So we have that
%  $ \models \Process$ and $ \models \ProcessQ$.
%  \cref{lem:Luca2ii} 
is easy to verify.

(\cref{lem:Luca4}). 
There exists
 a point in the derivation
 of $\CA \models \Process$ 
  where we split
 the two threads $\thread\var\Expression$ and $\thread\varY\ExpressionF$.
 This means that there is 
 $\Process_1 \parop \Process_2\subprocesseq \Process$
 such that
$\thread\var\Expression \subprocesseq \Process_1$
  and $\thread\varY\ExpressionF \subprocesseq \Process_2$
  (or vice versa).
  \cref{lem:Luca3} implies
  $\CB \models \Process_1\parop \Process_2$,
 % By  \cref{lem:Luca2},
  % \cref{lem:Luca2i}, 
  $ \models \Process_1$ and $\models \Process_2$.
  Hence, $\var^{-} \in \Pol(\Process_1)$ and $\varY^{-}, \varX^{+} \in \Pol(\Process_2)$,
  because we assume that $\var$ occurs in $\ExpressionF$.
  By \cref{lem:Luca2}\comma  
  % \cref{lem:Luca2ii}, 
  $\varY^{+} \not \in \Pol(\Process_1)$,
  which means that $\varY$ cannot occur in $\Expression$.

(\cref{lem:Luca5}). Similar to the previous item. 
\end{proof}

The interest in well-polarisation comes from its preservation by reduction of typeable processes, as stated in the following theorem whose proof can be found in  \cref{appendix:SRprocesses}.

\begin{theorem}
\label{theorem:invariance}
If $\Process\red\Process'$ and $\Process$ is typeable and well-polarised, then $\Process'$ is well-polarised too.
\end{theorem}

As an immediate consequence we have that reachable processes are well-polarised, since an initial process is trivially well-polarised.

  \begin{corollary}\label{corollary:reachableimpliesacyclic}
  Each reachable process is well-polarised. 
  \end{corollary}

\subsection{Subject Reduction for Reachable Processes}
\label{subsec:srp}

The following three lemmas  are for processes  
% like 
as  
 Lemmas~\ref{lem:inv},~\ref{lemma:substitution} and ~\ref{lem:keye}
are  for expressions. 

\begin{lemma}[Inversion for Processes]\label{lem:invp}\mbox{}
  \begin{enumerate}
 \item\label{lem:inv7} If
    $\wtp{\TypeContext}{\thread\var\Expression}{\ProvideContext}$,
    then $\ProvideContext = \varX:\tbullet[n]\Type$ with
    $\wte{\TypeContext}{\ExpressionE}{\tbullet^n(\tio\Type)}$ and $\var\not\in\dom(\TypeContext)$.

  \item\label{lem:inv8} If
    $\wtp{\TypeContext}{\server\Channel\Expression}{\ProvideContext}$,
    then $\TypeContext = \TypeContext', \Channel:\tshared\SessionType$
    and $\ProvideContext = \Channel:\tshared\SessionType$ with
    $\wte{\TypeContext}{\ExpressionE}{(\co\SessionType\arrow\tio\Type)}$
    and $\sharedq\TypeContext$ and $\unq\Type$.

  \item\label{lem:inv9} If
    $\wtp{\TypeContext}{\Process_1 \parop
      \Process_2}{\ProvideContext}$,
    then $\TypeContext = \TypeContext_1 +\TypeContext_2 $ and
    $\ProvideContext = \ProvideContext_1 ,\ProvideContext_2 $ with
    $\wtp{\TypeContext_1}{\Process_1}{\ProvideContext_1}$ and
    $ \wtp{\TypeContext_2}{\Process_2}{\ProvideContext_2}$.
    % and $\ctc{\TypeContext_1}{\TypeContext_2}$. 
    
  \item\label{lem:inv10}
    If $\wtp{\TypeContext}{\new\Channel\Process}{\ProvideContext}$, then  either $  \wtp{\TypeContext, \Channel^\Polarity : \SessionType, \Channel^{\co\Polarity} : \co\SessionType}\Process\ProvideContext$ or $ \wtp{\TypeContext, \Channel : \tshared\SessionType}\Process{ \ProvideContext, \Channel : \tshared\SessionType }$.

  \item\label{lem:inv11}
    If $\wtp{\TypeContext}{\new\varX\Process}{\ProvideContext}$, then  
    $ \wtp{\TypeContext, \varX : \Type}\Process{ \ProvideContext, \varX : \Type }$.
  \end{enumerate}

\end{lemma}
\begin{proof} By case analysis %and induction 
on the derivation. \end{proof}

\begin{lemma}[Substitution]
  \label{lemma:substitutionp}
   Let  
    $\wtp{\TypeContext_1,\varX:\Type}{\Process}{\ProvideContext}$ 
    with
   $\varX \not\in \dom(\ProvideContext)$
   and
    $\wte{\TypeContext_2}{\Expression}{\Type}$ and
    $\TypeContext_1 + \TypeContext_2$ be %is 
    defined and $\dom(\TypeContext_2) \cap \dom(\ProvideContext) = \emptyset$. 
    Then
    $\wtp{\TypeContext_1+\TypeContext_2}{\Process\subst\Expression\varX}{\ProvideContext}$.
\end{lemma}
\begin{proof}
By induction on the structure of processes. We only 
 discuss  %show 
 the case of rule \refrule{thread}.
 The interesting observation %in this case 
 is 
 that we need to use the hypothesis 
 $\dom(\TypeContext_2) \cap \dom(\ProvideContext) = \emptyset$
 to ensure that the name of the thread does not belong to its own body.
 We also use \cref{lemma:substitution} to type 
 the body of the thread itself.  \end{proof}
 
 \begin{lemma}[Evaluation Contexts for Processes]
  \label{lem:key}
If
   $\wte \TypeContext {\PContext[\Expression]}
   {\tbullet[n](\tio\TypeS)}$,
   then $\TypeContext = \TypeContext_1 + \TypeContext_2$ and
   \mbox{$\wte {\TypeContext_1,\varX:\tbullet[n](\tio\Type)
   }{\PContext[\varX]} {\tbullet[n](\tio\TypeS)}$}
   and
   $\wte{\TypeContext_2} \Expression {\tbullet[n](\tio\Type)}$.
 \end{lemma}
\begin{proof}
  By induction on the structure of evaluation contexts for processes.
  \end{proof}
  
   A useful consequence of the previous lemma is the following property of contexts filled by communication expressions. 
  \begin{lemma}
\label{lemma:delayreceive}
If $\wte{\TypeContext, \Channel^{\Polarity} : \tbullet[n] \SessionType}{\PContext[\send~\Channel^\Polarity~\Expression]}
{\tbullet[m](\tio {\Type})}$ or $\wte{\TypeContext, \Channel^{\Polarity} : \tbullet[n] \SessionType}{\PContext[\receive~\Channel^{\Polarity}]}
{\tbullet[m] (\tio {\Type})}$,
 then $n \leq m$.
\end{lemma}

\begin{proof}
We only consider the case $\wte{\TypeContext, \Channel^{\Polarity} : \tbullet[n] \SessionType}{\PContext[\send~\Channel^\Polarity~\Expression]}
{\tbullet[m] (\tio {\Type})}$. By \cref{lem:key} $\TypeContext=\TypeContext_1+\TypeContext_2$ and $\wte{\TypeContext_1, \var : \tbullet[m] (\tio {\TypeS})}{\PContext[\var]}
{\tbullet[m] (\tio {\Type})}$ and $\wte{\TypeContext_2, \Channel^{\Polarity} : \tbullet[n] \SessionType}{\send~\Channel^\Polarity~\Expression}
{\tbullet[m] (\tio {\TypeS})}$. Since rule $\defrule\elimarrow$ requires at least $n$ bullets in front of the type of $\send$ we get $n \leq m$.
\end{proof}

  %Communication consumes the session types of the endpoints.
%
%Therefore, reducing processes requires environments to be reduced too.
%%
%Following~\cite{DBLP:conf/esop/HondaVK98}, the  {\em reduction}
%$\redctx$ on environments 
%is the smallest reflexive relation closed under the
%following axiom:
%\[
%\TypeContext, \Channel^{\Polarity}: \tbullet[n](\tin\Type\SessionType),
%\Channel^{\co \Polarity} : \tbullet[n](\tout\Type\SessionTypeS) \
%\redctx \ \TypeContext, \Channel^{\Polarity}:
%\tbullet[n]\SessionType, \Channel^{\co \Polarity} :
%\tbullet[n]\SessionTypeS
%\]

We say that an environment $\TypeContext$ is \emph{balanced} if
$\Channel^{\Polarity} : \SessionType \in \TypeContext$ and
$\Channel^{\co \Polarity} : \SessionTypeS \in \TypeContext$ imply
$\SessionType = \co{\SessionTypeS}$.  
%A key property of environments'
%reduction is the preservation of balancing. 
%The proof is straightforward.  
  We can now state subject reduction of well-polarised processes. The proof of this theorem %can be found 
  is the content of~\cref{appendix:SRP}.
\begin{theorem}[Subject Reduction for Well-polarised Processes]
\label{theorem:acyclicimpliesSR}
Let $\TypeContext$ be balanced and $\Process$ be well-polarised. If
$\Process \red \Process'$ and
$\wtp{\TypeContext}{\Process}{\ProvideContext}$, then there is 
% an
balanced environment $\TypeContext'$ such that
% $\TypeContext \ctxred \TypeContext'$ and
$\wtp{\TypeContext'}{\Process'}{\ProvideContext}$.
\end{theorem}

\begin{theorem}[Subject Reduction for Reachable Processes]
\label{theorem:subjectreductionprocesses}
All reachable processes are typeable.
\end{theorem}
  \noindent
%{\bf Proof of \cref{theorem:subjectreductionprocesses}}.

\begin{proof}
  This follows from \cref{corollary:reachableimpliesacyclic} and
  \cref{theorem:acyclicimpliesSR}, observing that the empty session
  environment is balanced.
\end{proof}

\subsection{Progress of Reachable Processes}\label{subsec:prog}

We now turn our attention to the progress property
(\Cref{theorem:reducestoreturnI}).  A computation stops when there are
no threads left.
Recall that the reduction rule $\refrule{r-return}$ (\cf
\Cref{table:semantics}) erases threads. Since servers are permanent we
say that a process $\Process$ is {\em final} if
\[
\ProcessP \equiv
\new{\Channel_1 \cdots \Channel_m}
(
\server{\Channel_1}{\Expression_1} \parop 
\cdots
\parop
\server{\Channel_m}{\Expression_m})
\]
In particular, the idle process is final, since $m$ can be $0$.

The following lemma gives fundamental features of
linear types, which play an important role in the proof of progress.

\begin{lemma}[Linearity]
  \label{lemma:linearity}
  \mbox{} %The following properties hold:
  \begin{enumerate}
  \item \label{l1} If
  $\wte{\TypeContext, \Name  : \Type}{\Expression}{\TypeS}$
  and $\linq \Type$, then $\Name$ occurs exactly once
    in $\Expression$.
  \item \label{l2} If 
  $\wtp{\TypeContext, \Name : \Type}{\Process}{\ProvideContext}$
   and $\linq \Type$, then %$\Name$ occurs exactly once in $\Process$, i.e.  
    there exists exactly one thread
    $\thread {\var} \Expression$ of $\Process$ where $\Name$ occurs
    only once in $\Expression$ 
     and $\Name$ occurs as  name of another thread if $\Name : \Type\in\ProvideContext$ and nowhere else.

\end{enumerate}
\end{lemma}
\begin{proof}
Both items are proved by induction on derivations.\end{proof}  

The following properties of typeable processes are handy in the proof
of progress.
\begin{lemma}\label{lemma:s}\label{lemma:rt}
%  Let
%  $\Process\equiv{\new{\BindableNameY_1\cdots\BindableNameY_n}}\new\bindname\Process'$
%  be typeable,
%  where $\Process'$ does not contain restrictions. 
  Let $\ProcessP$ be typeable.  
  Then:
\begin{enumerate}
\item\label{item:rt1} 
%If $\new\bindname\Process'\equiv\new\varX(\Context[\varX]\parop\ProcessQ)$, then $\ProcessQ\equiv\thread\var\Expression\parop\ProcessQ'$.
If $\varX \in \boundnames(\ProcessP)$ and
$\thread{\varY}{\Context[\varX]} \subprocesseq 
\threads(\Process)$,
then $\thread{\varX}{\Expression} \subprocess \threads(\Process)$.

\item\label{item:rt2}
% If $\new\bindname\Process'\equiv\new\Channel(\PContext[\open~\Channel]\parop\ProcessQ)$, then $\ProcessQ\equiv\server~\Channel~\Expression\parop\ProcessQ'$.
If $\Channel \in \boundnames(\ProcessP)$ and
$\thread{\varX}{\PContext[\open~\Channel]} \subprocesseq \threads(\Process)$,
then $\server~\Channel~\Expression \subprocesseq \servers (\ProcessP)$.

\item\label{item:rt3} 
%If $\new\bindname\Process'\equiv\Channel(\thread\var{\PContext[\send~\Channel^\Polarity~\Expression]}\parop\ProcessQ)$, then $\ProcessQ\equiv\thread\varY\ExpressionF\parop\ProcessQ'$, where $\Channel^{\co\Polarity}$ only occurs in expression $\ExpressionF$ and the typing environment for $\thread\var{\PContext[\send~\Channel^\Polarity~\Expression]}\parop \thread\varY\ExpressionF$ contains $\Channel^{\Polarity} : \tout \Type \SessionType$
%and $\Channel^{\co \Polarity} : \tin \Type \co\SessionType$.
If $\Channel \in \boundnames(\ProcessP)$ and
$\thread{\varX}{\PContext[\send~\Channel^\Polarity~\Expression]} \subprocesseq \threads(\Process)$,
then 
$\thread\varY\ExpressionF \subprocess \threads(\Process)$, 
where $\Channel^{\co\Polarity}$ only occurs in expression $\ExpressionF$ and the typing environment for  $\threads(\Process)$ 
%$\thread\var{\PContext[\send~\Channel^\Polarity~\Expression]}\parop \thread\varY\ExpressionF$ 
contains both $\Channel^{\Polarity} : \tout \Type \SessionType$ and $\Channel^{\co \Polarity} : \tin \Type \co\SessionType$.

\item\label{item:rt4} 
%If $\new\bindname\Process'\equiv\new\Channel(\thread\var{\PContext[\receive~\Channel^\Polarity]}\parop\ProcessQ)$, then $\ProcessQ\equiv\thread\varY\ExpressionF\parop\ProcessQ'$, 
If $\Channel \in \boundnames(\ProcessP)$ and
$\thread{\varX}{\PContext[\receive~\Channel^\Polarity]} \subprocesseq \threads(\Process)$,
then 
$\thread\varY\ExpressionF \subprocess \threads(\Process)$, 
where $\Channel^{\co\Polarity}$ only occurs in expression $\ExpressionF$ and the typing environment for $\threads(\Process)$ 
%$\thread\var{\PContext[\receive~\Channel^\Polarity]}\parop$ $\thread\varY\ExpressionF$ 
contains both $\Channel^{\Polarity} : \tin \Type \SessionType$ and $\Channel^{\co \Polarity} : \tout \Type \co\SessionType$.
\end{enumerate}
\end{lemma}
\begin{proof}
%Typeability of $\Process$ implies typeability of $\new\bindname\Process'$. 
(\cref{item:rt1}) and (\cref{item:rt2}). 
%In both cases 
To type 
% $\new\bindname\Process'$ 
the restriction of $\varX$ (or $\Channel$),
we need to use rule $\refrule{\newrule}$, 
which requires $\varX$ (or $\Channel$)  
to occur in the resource environment. 
Rule $\refrule{thread}$ is the only rule that puts the name of a thread in the resource environment. Rule $\refrule{server}$ is the only rule that puts the name of a server in the resource context. 

(\cref{item:rt3}). To type 
% $\new\bindname\Process'$ 
the restriction of $\Channel$, 
we need to use rule $\refrule{session}$, which requires the environment to contain dual session types for $\Channel^\Polarity$ and $\Channel^{\co\Polarity}$.  
Since $\Channel^\Polarity$ is an argument of $\send$,
its type is of the form 
$\tout \Type \SessionType$ and hence, $\Channel^{\co\Polarity}$
should have type $\tin  \Type \co \SessionType$.
The fact that  $\Channel^{\co\Polarity}$
occurs in only one thread follows from \ri{\cref{lemma:linearity}}{\cref{l2}}. 

(\cref{item:rt4}). The proof is similar to \cref{item:rt3}. 
\end{proof}

The proof of 
%progress 
Theorem \ref{theorem:reducestoreturnI}
requires to define a standard precedence between threads
and show %prove 
that this relation is acyclic.  Informally, a thread precedes
another one if the first thread must be evaluated before the second
one.  The simpler case is when the body of one thread is an evaluation
context containing the name of another thread, i.e.
$\thread \varX \Expression$ precedes $\thread \varY \Context[\varX]$.
In the % other 
remaining cases 
the bodies of the threads are 
% normal forms 
%requiring other threads to reduce. 
% This happens when the bodies are
%expressions of the shapes
%$\PContext[\send~\Channel^{\Polarity}~\Expression]$ and
%$\PContext[\receive~\Channel^{\Polarity}]$. 
% This thread has to wait
% for $\Channel^{\co\Polarity}$ to be 
%  %if $\Channel^{\co\Polarity}$ is  
%inside an evaluation context.  
 the normal forms 
$\PContext[\send~\Channel^{\Polarity}~\Expression]$ 
or $\PContext[\receive~\Channel^{\Polarity}]$
which have to wait 
 for $\Channel^{\co\Polarity}$ to be  
inside an evaluation context.
This is
formalised in the following definition.

%Essentially a thread ready to execute a communication action on the endpoint $\Channel^{\Polarity}$ follows a thread containing $\Channel^{\co\Polarity}$ not inside an evaluation context. This is formalised in the following definition.

\begin{definition}[Precedence]\mbox{}
%We say that:
\begin{enumerate}
\item The endpoint $\Channel^{\Polarity}$ is {\em ready} in $\Expression$
if 
\[
\Expression\in\set{\PContext[\send~\Channel^{\Polarity}~\ExpressionF], 
\PContext[\receive~\Channel^{\Polarity}]}
\]

\item The endpoint $\Channel^{\Polarity}$ is {\em blocked} in
$
\Expression$
if
%at least 
one of the following conditions holds:
\begin{enumerate}
\item $\Expression=\PContext[\send~\ChannelB^{\PolarityQ}~\ExpressionF]$
 and 
$\Channel\not=\ChannelB$ 
 and
 $\Channel^{\Polarity}$ 
 occurs in $\PContext$ or in $\ExpressionF$;
\item $\Expression=\PContext[\receive~\ChannelB^{\PolarityQ}]$
 and $\Channel\not=\ChannelB$ 
  and $\Channel^{\Polarity}$ occurs in $\PContext$;
\item $\Expression=\Context[\varX]$ and $\Channel^{\Polarity}$ occurs in $\Context$.
\end{enumerate}

\item The expression $\Expression$ {\em precedes} the expression $\ExpressionF$
 (notation $\Expression \prered \ExpressionF$) if $\Channel^{\Polarity}$ is  ready 
  in $\ExpressionF$ while 
 $\Channel^{\co \Polarity}$ is blocked  in $\Expression$.

\item The thread $\thread \varX \Expression$ \emph{precedes} the
  thread $\thread \varY\ExpressionF $ (notation
  $\thread \varX \Expression \prered \thread \varY \ExpressionF$) if
  either $\Expression \prered \ExpressionF$ or
  $\ExpressionF=\Context[\varX]$.
\end{enumerate}
\end{definition}
\noindent
Note that a channel is either ready or blocked in a typeable expression.
%if $\Channel^\Polarity$ is ready in $\Expression$ then
%$\Channel^{\co \Polarity}$ cannot be blocked in
%the same expression $\Expression$.  

The following lemma follows easily from the definition of
$\pmb{\prec}$.   The proof of the third item uses
%The third part follows by applying 
\cref{lemma:linearity} and typeability of $\Process$. 

\begin{lemma}
\label{lemma:variablesofprecedence}
 Let $\Process$ be a reachable process and 
$\thread{\varX}{\Expression} \parop \thread{\varY}{\ExpressionF} \subprocesseq \Process$
and   $\thread{\varX}{\Expression} \prered \thread{\varY}{\ExpressionF}$.
Then there is $\Thing$ such that: 
\begin{enumerate}
\item \label{lemma:variablesprecedencePartOne}
 $\Thing^{\Polarity} \in \Pol(\thread{\varX}{\Expression})$ and
$\Thing^{\co \Polarity} \in \Pol(\thread{\varY}{\ExpressionF})$.

\item \label{lemma:variablesprecedencePartTwo}
If $\Thing$ is a variable, then $\Thing =\varX$.

 \item \label{lemma:variablesprecedencePartThree}
 If $\Thing = \Channel$, then $\Channel^{\Polarity}$ is blocked in 
 $\Expression$ while $\Channel^{ \co   \Polarity}$ is  ready 
  in $\ExpressionF$. Moreover, these are the only occurrences of 
  $\Channel^{\Polarity}$ and $\Channel^{ \co   \Polarity}$ in $\Process$.
\end{enumerate}
\end{lemma}

A process $\Process$ is {\em acyclic} if the precedence between the
threads in $\threads(\Process)$ has no cycles. As we will see in the
proof of ~\cref{theorem:reducestoreturnI} acyclicity is a crucial
property to assure progress.  We can %It is easy to 
show that each reachable 
process is acyclic.

\begin{lemma}\label{lemma:pt}
Each 
reachable   
process is acyclic. 
\end{lemma}

\begin{proof} 
  Suppose towards a contradiction that a reachable process
  $\Process$ contains a cycle and consider 
  $\ProcessQ \subprocesseq  \threads(\Process)$ such that $\ProcessQ$
  includes  all the threads involved in that cycle.  By
  \cref{lem:Luca3} $ \models Q$. 
   
  Suppose first that  $\ProcessQ=\thread\varX \Expression$.  
  Then  
  $\thread\varX \Expression \prered \thread\varX \Expression$.  It
  follows from \ri{\cref{lemma:variablesofprecedence}}{\cref{lemma:variablesprecedencePartOne}} that both
  $\Thing^{\Polarity} $ and $\Thing^{\co \Polarity} $ occur in
  $\thread{\varX}{\Expression}$.  This contradicts %the fact that
  $\models \ProcessQ$.

  Suppose now that 
  $\ProcessQ = \ProcessQ_1 \parop \ProcessQ_2$. 
  Since $\ProcessQ$ contain all the threads involved in the cycle,
  there are two threads $\thread {\varX_1}{ \Expression_1}$ and
  $ \thread {\varY_1}{ \ExpressionF_1}$ of $\ProcessQ_1$ and two
  threads $\thread {\varX_2}{ \Expression_2}$ and
  $ \thread {\varY_2}{ \ExpressionF_2}$ of $\ProcessQ_2$ such that
\[
\begin{array}{lll}
\thread {\varX_1}{ \Expression_1} \prered \thread {\varX_2} {\Expression_2}
\\
\thread {\varY_1}{ \ExpressionF_1} \afterred
 \thread {\varY_2} {\ExpressionF_2}
\end{array}
\]
\cref{lemma:variablesprecedencePartOne} of \cref{lemma:variablesofprecedence}  
gives  
$\Thing^{\Polarity}\in\Pol(\thread {\varX_1}{ \Expression_1})$,  $\Thing^{\co \Polarity}\in\Pol(\thread {\varX_2}{ \Expression_2})$
 and
$\ThingY^{\PolarityQ} \in \Pol(\thread {\varY_1}{ \ExpressionF_1})$,  $ \ThingY^{\co \PolarityQ} 
\in \Pol( \thread {\varY_2} {\ExpressionF_2})$.
 \cref{lem:Luca2} of \cref{lem:Luca} requires $\Thing = \ThingY$ and $\Polarity=\PolarityQ$. 
Suppose $\Thing= \ThingY$ is a variable. Then\comma it follows from
  \cref{lemma:variablesprecedencePartTwo} of \cref{lemma:variablesofprecedence} 
  that $\varX_1 = \Thing = \ThingY =\varY_2$.
  This contradicts the typeability of the process $\Process$, since 
the typing rule \defrule{par} guarantees that all threads
have different names.
Suppose now that $\Thing = \Channel$.
Then
$\Channel^\Polarity\in\Pol( \Expression_1)$, 
$\Channel^{\co \Polarity}\in\Pol( \Expression_2)$ and $\Channel^\Polarity\in\Pol( \ExpressionF_1)$, 
$\Channel^{\co \Polarity}\in\Pol( \ExpressionF_2)$.
It follows from 
 \cref{lemma:variablesprecedencePartThree} of \cref{lemma:variablesofprecedence} that  
$\Channel^\Polarity$ and $\Channel^{\co \Polarity}$
 occur only once in $\Process$. 
  This 
 is possible only if 
 $\Expression_1=\ExpressionF_1$ and 
  $\Expression_ 2=\ExpressionF_ 2$.
  \ri{\cref{lemma:variablesofprecedence}}{\cref{lemma:variablesprecedencePartThree}} implies that $\Channel^\Polarity$ is blocked in $\Expression_1$ while ready in $\ExpressionF_1$, and $\Channel^{\co \Polarity}$ is ready in $\Expression_2$ while blocked in $\ExpressionF_2$. 
  This is absurdum since $\Expression_1=\ExpressionF_1$ and 
  $\Expression_ 2=\ExpressionF_ 2$. 
% We also know from
% \cref{lemma:variablesprecedencePartThree} of \cref{lemma:variablesofprecedence} that  
%$\ThingY^\PolarityQ = \Channel^\Polarity$ is blocked in 
%$\ExpressionF_2$ 
%and $\Thing^{\co \Polarity} = \Channel^{\co \Polarity}$
% is ready in $\Expression_2$.
% This contradicts the definition of precedence  
%  since $\Expression_2 = \ExpressionF_2$.    
\end{proof}

For the proof of progress it is useful  
to consider the reduction 
$\redminus$ without
rules \refrule{r-open} and \refrule{r-future}
and  to show that it is strongly normalising for typeable processes.
The process $\bdisplay$ %defined in
of~\cref{section:Introduction} 
 which has an infinite $\redminus$-reduction sequence
  is rejected by our type system.

\begin{theorem}[Strong Normalisation of $\redminus$]\label{theorem:snrm}
The reduction $\redminus$ on  typeable processes
is strongly normalising.
\end{theorem}

\begin{proof}  The proof requires some definitions for getting a weight of typeable processes which decreases by reduction. 
For $\Type\not=\infinite$ we define 
the function $\delay(\Type)$ which counts the number of initial bullets in a type $\Type$ as follows.
\[\delay(\Type)=\begin{cases}
 1+\delay(\Type')      & \text{if }
 \Type=\bullet\Type' \\
    0 & \text{otherwise}
\end{cases}\]
We extend $\delay$ to resource environments by 
\[
\delay ( x_1: \Type_1, \ldots, x_n : \Type_n)
 = max \{ \delay(\Type_i) \mid 1 \leq i \leq n \}\]
Let $\SessionType\not=\infinite$ and $m \in \mathbb{N}$. We define 
the function $\numCOMi{m}{\SessionType}$
that counts the number of $?$ and $!$ in a session type 
$\SessionType$  only until time $m$ as follows.
\[
\begin{array}{lll}
\numCOMi{m}{\End }  = 0 &\quad&
\numCOMi{m}{\tin\Type\SessionType}  =
   \numCOMi{m}{\tout\Type\SessionType}=1+ \numCOMi{m}{\SessionType} \\   
\numCOMi{0}{\tbullet \SessionType}  = 0 &&
 \numCOMi{m+1}{\tbullet \SessionType} = 
  \numCOMi{m}{ \SessionType}
\end{array}
\]
   We extend the function $\numCOM$
   to type environments by
$
\numCOMi{m}{\TypeContext} = 
\numCOMi{ m }{\SessionType_1}
+ \ldots + \numCOMi{m}{\SessionType_n}
$,
where $\Channel_1^{+} : \SessionType_1, \ldots,
\Channel_n^{+} : \SessionType_n$
are  the type declarations for 
the positive session channels occurring in $\TypeContext$.

Let 
$\numbred (\Expression)$ be the number of reduction steps
to reach the normal form of $\Expression$. 
We define the {\em weight of the typeable process $\Process$ for the environments $\ProvideContext$ and $\TypeContext$} by
\begin{center}
$
\weight (\Process, \ProvideContext,\TypeContext) = 
(k, \numCOMi{\delay(\ProvideContext)}{\TypeContext}, \numbred (\Expression_1) + \ldots
+ \numbred (\Expression_k))
$ 
\end{center}
where $\threads(\Process) \equiv \ProcessQ
 =  \thread{x_1}{\Expression_1} \parop
\ldots \parop \thread{x_k}{\Expression_k}$  and 
$\wtp{\TypeContext}{\ProcessQ}{\ProvideContext}$.

We will prove that
if  $\Process \redminus \Process'$ 
and $\wtp{\TypeContext}{\threads(\Process)}{\ProvideContext}$ and $\wtp{\TypeContext'}{\threads(\Process')}{\ProvideContext}$, where the derivation $\wtp{\TypeContext'}{\threads(\Process')}{\ProvideContext}$ is obtained from $\wtp{\TypeContext}{\threads(\Process)}{\ProvideContext}$ as in the proof of 
the Subject Reduction Theorem, 
 then $\weight (\Process, \ProvideContext,\TypeContext)> \weight (\Process', \ProvideContext,\TypeContext')$. 

The only interesting case is  $\threads(\Process)\equiv\thread\varX{\PContext[\send~\Channel^+~\Expression]}
      \parop
      \thread\varY{\PContext'[\receive~\Channel^-]}
      \parop\ProcessQ$
   and \mbox{$\threads(\Process')\equiv
      \thread\varX{\PContext[\Return{\Channel^+}]}
      \parop
      \thread\varY{\PContext'[\Return{\Pair\Expression{\Channel^-}}]}
      \parop\ProcessQ$.}
 The first components of
  $\weight (\Process, \ProvideContext,\TypeContext)$ and 
  $\weight (\Process', \ProvideContext,\TypeContext')$  
  are equal since the number of threads does not change.
  We will prove that the second component decreases.
       By the Inversion Lemma $\TypeContext$ must contain suitable session types for $\Channel^+,\Channel^-$. We can then assume
\[\TypeContext= 
\Channel^+ : \tbullet[n](\tout\Type\SessionType), \Channel^-: \tbullet[n](\tin\Type\co\SessionType),\TypeContext_0\] We get 
$\TypeContext'= \Channel^+ :
\tbullet[n]\SessionType,\Channel^-:
\tbullet[n]\co\SessionType, \TypeContext_0$. Let $m = \delay (\ProvideContext)$. By the Inversion Lemma if the type of $\PContext[\send~\Channel^+~\Expression]$ is $\tbullet[m']\tio\TypeS$, then $\ProvideContext$ contains $\var:\tbullet[m']\TypeS$.
By Lemma \ref{lemma:delayreceive} $n \leq m'$, which implies $n \leq m$.  From $\numCOMi m {\tbullet[n](\tout\Type\SessionType)}= \numCOMi {m-n} {\tout\Type\SessionType}=1+\numCOMi {m-n} {\SessionType}$ and 
$\numCOMi m {\tbullet[n]\SessionType}= \numCOMi {m-n} {\SessionType}$ we conclude $\numCOMi m \TypeContext> \numCOMi m {\TypeContext'}$ as desired. 
 \end{proof}

 As a consequence of the above theorem, 
 every infinite reduction of a typeable  %reachable 
 process
%performs infinitely many communications and/or 
spawns infinitely many
threads.

\begin{theorem}[Progress of Reachable Processes]
\label{theorem:reducestoreturnI} 
A reachable process either reduces or it is final.
Moreover a non-terminating
 reachable process reduces in a finite number of
steps to a process to which one of the rules $\refrule{r-open}$
%$\refrule{r-comm}$ 
or   $\refrule{r-future}$ 
% or $\refrule{r-return}$ 
must be applied.
\end{theorem}

%\noindent
%{\bf Proof of~\cref{theorem:reducestoreturnI}}

\begin{proof}
If a process has no thread, then it is final. In discussing the other cases we omit to mention the application of rules $\refrule{r-new}$, $\refrule{r-par}$ and $\refrule{r-cong}$.

 If a process has a thread whose body is a reducible expression, then the process is reducible by rule $\refrule{r-thread}$. 
If a process has a thread whose body is $\PContext[\future~\Expression]$, then the process is reducible by rule $\refrule{r-future}$.
If a process has a thread whose body is $\return~\Expression$, then the process is reducible by rule $\refrule{r-return}$.
If a process has a thread whose body is $\PContext[\open~\Channel]$, then by~\ri{\cref{lemma:rt}}{\cref{item:rt2}} the process has a server named $\Channel$. Therefore the process  is reducible by rule $\refrule{r-open}$.

Otherwise all the bodies of the threads of the process are of the shapes $\PContext[\send~\Channel^{\Polarity}~\Expression]$, 
$\PContext[\receive~\Channel^{\Polarity}]$ and $\Context[\varX]$. 
Since reachable processes are  well-polarised,  \cref{lemma:pt} assures that there is at least one minimal thread in the precedence order, let it be $\thread\varX\Expression$. The expression $\Expression$ cannot be $\Context[\varY]$, since~\ri{\cref{lemma:rt}}{\cref{item:rt1}} implies that the process  should have one thread $\thread\varY\ExpressionF$; and by  definition of precedence $\thread\varY\ExpressionF\prered\thread\varX\Expression$, which contradicts the minimality of $\thread\varX\Expression$. Let $\Expression=\PContext[\send~\Channel^{\Polarity}~\Expression']$. \ri{\cref{lemma:rt}}{\cref{item:rt3}} implies that the process  should have  one thread $\thread\varY\ExpressionF$ and $\Channel^{\co\Polarity}$ occurs in $\ExpressionF$. The expression $\ExpressionF$ 
 % cannot be
  can be neither of the following:  
 \begin{itemize}\item 
$\PContext'[\send~\ChannelB^{\PolarityQ}~\ExpressionF']$ with $\ChannelB\not=\Channel$ and $\Channel^{\co\Polarity}$ occurring in $\PContext'$ or 
$\ExpressionF'$   
\item 
$\PContext'[\receive~\ChannelB^{\PolarityQ}]$ with $\ChannelB\not=\Channel$ and $\Channel^{\co\Polarity}$ occurring in $\PContext'$ 
\item $\Context[\varZ]$ with $\Channel^{\co\Polarity}$ occurring in $\Context$
\end{itemize} since we would get $\thread\varY\ExpressionF\prered\thread\varX\Expression$. Then $\ExpressionF$ can only be either $\PContext'[\send~\Channel^{\co\Polarity}~\ExpressionF']$ or 
$\PContext'[\receive~\Channel^{\co\Polarity}]$. 
Since reachable processes are  typeable,  \ri{\cref{lemma:rt}}{\cref{item:rt3}} gives type $\tin\Type\co\SessionType$ for $\Channel^{\co\Polarity}$, so we have $\ExpressionF=\PContext'[\receive~\Channel^{\co\Polarity}]$. The process can then be reduced using rule $\refrule{r-comm}$. The proof for the case 
$\Expression=\PContext[\receive~\Channel^{\Polarity}]$ is similar  and it uses~\ri{\cref{lemma:rt}}{\cref{item:rt4}}. 

\cref{theorem:snrm} assures that infinite applications of rules $\refrule{r-open}$
and   $\refrule{r-future}$ are needed to get infinite computations. 
\end{proof}

Let   $\Process_0$ and $\ProcessQ$ be defined as at the end of~\cref{section:language}. Note that $\Process_0$ is typeable, and indeed an initial process.
  Hence, by Theorems~\ref{theorem:subjectreductionprocesses}
  and~\ref{theorem:reducestoreturnI}, process 
   $\ProcessQ$ is typeable
  and has progress.
%   This implies that also $\Process$ is typeable. Moreover $\Process$ has progress too, since it is obtained from $\ProcessQ$ by erasing the restrictions and the server $\Channel$, and it does not contain $\open~\Channel$ .

We now show  two initial processes whose
progress is somewhat degenerate. The first one
realises an infinite sequence of \emph{delegations} (the act of
sending an endpoint as a message), thereby postponing the use of the
endpoint forever:
\[\begin{array}{lll} \mkconstant{badserver} &  \eqdef 
&
\new {\varX    \Channel \ChannelB}
(\thread \varX \Bind{\open~\Channel}{\loopone}  \parop\\
&&\phantom{\new {\varX    \Channel \ChannelB}(}
\server~\Channel~{\Fun \varY  \Bind{\open~\ChannelB}{\looptwo~\varY}  } \parop
\server~\ChannelB~\receive
 )
 \end{array}
\]
where
\[
\begin{array}{lll}
\loopone & \eqdef & 
\Fix{f}  
\Fun{x}\Bind {\Receive x}{}
 \Bind{\Fun{y} \Split{y}{y_1}{y_2}{\Send{y_2}{y_1}} } {}\\
 &&\hfill 
 \Fun {z} \Future{ (f z)}
\\\\
\looptwo & \eqdef & 
\Fix{g}  
\Fun{y x}\Bind {\Send{x}{y}}{}
\Fun{z}\Bind {\Receive z}{}\\
&&\hfill \Fun{u} \Split{u}{u_1}{u_2}{\Future{(g u_1 u_2) }} 
\end{array}
\]
We have that
$ \loopone : \sessiontypeloop{\Type} \arrow \tio \infinite $ and
$\looptwo : \Type \arrow \sessiontypelooptwo{\Type} \larrow \tio
\infinite$,
where
$ \sessiontypeloop{\Type} = \tin{\Type} \tout{\Type} \tbullet
\sessiontypeloop{\Type} $
and
$ \sessiontypelooptwo{\Type} = \tout{\Type} \tin{\Type} \tbullet
\sessiontypelooptwo{\Type} $.
Since no communication ever takes place on the session created with
server $\ChannelB$, $\mkconstant{badserver}$ violates lock freedom,
which is
%the 
progress 
%property as defined
in~\cite{DY11}. 
%\bcompaula Which property of that paper is exactly
%that we are referring to?
%\ecompaula

The second example is the initial process
$\new \varX (\thread \varX \omegafuture)$, where
$\omegafuture = \fix~\future$.  This process only creates new threads.

\subsection{Confluence of Reachable Processes}\label{subsec:conf}

In this section we  prove that the reduction relation 
 is confluent on reachable processes.
The proof  is trivial for expressions, since there
is only one redex at each reduction step.
However, for processes we may have several redexes to contract at a
time and the proof requires to 
analyse  these possibilities.
Once again well-polarisation plays a crucial role in the proof. 
The fact that we can mix pure evaluations and communications and still
preserve determinism is of practical interest.

Notice that typeability forbids processes where the same variable can be replaced by different expressions, like the process 
\[\new \var (\thread \var {\return~ 0} \parop \thread \var {\return~ 1}  \parop  \thread \varY {\upkf~ \ChannelC^+ ~(\streamf~ \var)})\]
which reduces to both $\thread \varY {\upkf~ \ChannelC^+ ~(\streamf~ 0)}$ and $\thread \varY {\upkf~ \ChannelC^+ ~(\streamf~ 1)}$. 

%
%The last property of \sid{} we discuss is the diamond
%  property~\cite[\S 30.3]{Pierce02}.
%
\begin{theorem}[Confluence of Reachable Processes]\label{thm:confluence}
Let $\Process$ be a reachable process.  
If $\Process \red \Process_1$ and
 $\Process \red \Process_2$, 
 then either  $\Process_1 \equiv \Process_2$ or 
 there is
 $\Process_3$ 
  such that 
  $\Process_1 \red \Process_3$ and 
 $\Process_2 \red \Process_3$.
% Diagrammatically, 
% \[
%\xymatrix{
%&\ar@{->}[dl] \Process   \ar@{->}[dr] &  \\ 
%\Process_1 \ar@{.>}[dr]  &  & \Process_2 \ar@{.>}[dl] \\
%  & \Process_3 &
%}
%\]
% \[
%\xymatrix{
%&\ar@{->}[dl]_>>{*} \Process   \ar@{->}[dr]^>>{*} &  \\ 
%\Process_1 \ar@{.>}[dr]_>>{*}  &  & \Process_2 \ar@{.>}[dl]^>>{*} \\
%  & \Process_3 &
%}
%\]
% 
\end{theorem}

\begin{proof}
The proof proceeds by case analysis.

\begin{enumerate}

\item Suppose rule
\refrule{r-return} is not applied.
%Since the redexes are non-overlapping, it is easy to see
%that there is a common reduct 
%$\ProcessQ'_1 \parop \ProcessQ'_2$
% since all cases follow the same pattern
% where $\Process \equiv \ProcessQ_1 \parop \ProcessQ_2$
%and $\ProcessQ_1 \red \ProcessQ'_1$ and
%$\ProcessQ_2 \red \ProcessQ'_2$.
Since the redexes are non-overlapping, it is easy to see that 
 $\Process \equiv \ProcessQ_1 \parop \ProcessQ_2$ and $\Process_1 \equiv \ProcessQ'_1 \parop \ProcessQ_2$ and $\Process_2 \equiv \ProcessQ_1 \parop \ProcessQ'_2$ 
 from 
 $\ProcessQ_1 \red \ProcessQ'_1$ and
$\ProcessQ_2 \red \ProcessQ'_2$. The 
common reduct is then 
$\ProcessQ'_1 \parop \ProcessQ'_2$. 
% \[
%\xymatrix{
%&\ar@{->}[dl] \ProcessQ_1 \parop \ProcessQ_2   \ar@{->}[dr] &  \\ 
%\ProcessQ'_1 \parop\ProcessQ_2   \ar@{.>}[dr]  &  &
% \ProcessQ_1 \parop \ProcessQ'_2 \ar@{.>}[dl] \\
%  & \ProcessQ'_1 \parop \ProcessQ'_2
%}
%\]
\item Let $\Process\equiv\new{ \varX \varY}
(\thread \varX \return~\Expression \parop 
\thread \varY  \return~\ExpressionF \parop
\ProcessR) $ and suppose we apply rule
\refrule{r-return}
in both directions. 
 Typing rule $\defrule{par}$  implies $\var\not=\varY$. 
 Since $\Process$  is reachable, and then well-polarised by~\cref{corollary:reachableimpliesacyclic},
we cannot have both $\varY  \in \fv(\Expression)$ and 
$\varX \in \fv(\ExpressionF)$ by
 \ri{\cref{lem:Luca}}{\cref{lem:Luca4}}.
  Suppose $\varY \not \in \fv(\Expression)$. Then\comma 
   \[
   \begin{array}{ll}
   \Process_1 & \equiv \new{ \varX}
(\thread \varX \return~\Expression
\parop 
\ProcessR\sust{\varY}{\ExpressionF}) \text{ and }
\\
\Process_2 & 
\equiv\new{\varY}
(
\thread \varY  \return~(\ExpressionF \sust{\varX}{\Expression})
 \parop
\ProcessR\sust{\varX}{\Expression})
\end{array}
\]
and the common reduct of  $\Process_1$ and $\Process_2$ is
$
 \ProcessR  
      \sust{\varX}{\Expression} 
      \sust{\varY}{ \ExpressionF \sust{\varX}{\Expression}} 
$.
%
%\[
%\xymatrix{
%&\ar@{->}[dl] \Process   \ar@{->}[dr] &  \\ 
% \Process_1  \ar@{.>}[dr]  &  &
% \Process_2 \ar@{.>}[dl] \\
%  & \ProcessR  
%      \sust{\varX}{\Expression} 
%      \sust{\varY}{ \ExpressionF \sust{\varX}{\Expression}}
%}
%\]

\item Let $\Process\equiv\new \varZ 
 (\thread\varX{\PContext[\send~\Channel^\Polarity~\Expression]}
    \parop
    \thread\varY{\PContext'[\receive~\Channel^{\co\Polarity}]}
    \parop
    \thread \varZ \return~\ExpressionF \parop \ProcessR)$ and suppose that in one direction we apply \refrule{r-return}
 and in the other direction we apply \refrule{r-comm}. Then\comma
%By applying \refrule{r-return}, we obtain  
 \[
   \begin{array}{ll}
   \Process_1 & \equiv   \thread\varX{\PContext \sust{\varZ}{\ExpressionF}
    [\send~\Channel^\Polarity~{\Expression \sust{\varZ}{\ExpressionF}}]}
    \parop
    \thread\varY{\PContext' \sust{\varZ}{\ExpressionF} [\receive~\Channel^{\co\Polarity}]}
    \parop \ProcessR\sust{\varZ}{\ExpressionF}\\
 \Process_2 & \equiv   \new \varZ( \thread\varX{\PContext [\return~\Channel^\Polarity]}
    \parop
    \thread\varY{\PContext'
    [\return~\Pair{\Expression} {\Channel^{\co\Polarity}}]} 
    \parop
    \thread \varZ \return~\ExpressionF 
    \parop    
    \ProcessR)
\end{array}\]
%\begin{equation}
%\label{equation:reducereturn}
%\thread\varX{\PContext \sust{\varZ}{\ExpressionF}
%    [\send~\Channel^\Polarity~{\Expression \sust{\varZ}{\ExpressionF}}]}
%    \parop
%    \thread\varY{\PContext' \sust{\varZ}{\ExpressionF} [\receive~\Channel^{\co\Polarity}]}
%    \parop \ProcessR\sust{\varZ}{\ExpressionF}
%    \end{equation}
%In the other direction, we apply 
%\refrule{r-comm} and obtain:
%\begin{equation}
%\label{equation:reducesendreceive}
%   \new \varZ( \thread\varX{\PContext [\return~\Channel^\Polarity]}
%    \parop
%    \thread\varY{\PContext'
%    [\return~\Pair{\Expression} {\Channel^{\co\Polarity}}]} 
%    \parop
%    \thread \varZ \return~\ExpressionF 
%    \parop    
%    \ProcessR)
%\end{equation}
%
It is easy to see that $\Process_1$ and $\Process_2$ %\eqref{equation:reducereturn} and \eqref{equation:reducesendreceive} 
have the  common reduct: 
\[
    \thread\varX{\PContext \sust{\varZ}{\ExpressionF} [\return~\Channel^\Polarity]}
    \parop
    \thread\varY{\PContext'\sust{\varZ}{\ExpressionF}
    [\return~\Pair{\Expression\sust{\varZ}{\ExpressionF}} {\Channel^{\co\Polarity}}]} \parop
    \ProcessR\sust{\varZ}{\ExpressionF}
\]

\item The remaining cases are similar to the last one.
\qedhere
\end{enumerate}
\end{proof}

\section{Related Work}
\label{section:relatedwork}

To the best of our knowledge, \sid\ is the first calculus that
combines {\em session-based} communication
primitives~\cite{DBLP:conf/esop/HondaVK98,DBLP:journals/iandc/Vasconcelos12}
with a {\em call-by-need} operational semantics
\cite{WadworthPhDthesis,AriolaFMOW95,MaraistOW98}.

There are many calculi with functional and concurrent features, one of
the more interesting ones being Boudol's blue calculus \cite{B98}.
 In the context of communication-centric calculi, infinite
data are explicitly considered in~\cite{DBLP:conf/icfp/LindleyM16,
  clmrv14} and \lq\lq implicitly\rq\rq\ handled
in~\cite{TCP13,ToninhoCP14}, where 
recursive/coinductive sessions are used to
encode 
%recursive session types are used to
%indirectly handle 
infinite communications. 

%\bcompaula delete this long paragraph on \cite{TCP13} ? \ecompaula
Toninho et al. \cite{TCP13} integrate \ the Curry-Howard interpretation
of linear sequent calculus as session-typed processes in a functional
language. The main construct is a contextual monad encapsulating open
concurrent computations, which can be communicated between processes
in the style of higher-order processes. This allows for example to
construct a stream transducer. 
%Infinite loops are avoided since
%communications are synchronous, so outputs are produced only when
%matching inputs require them. 
In the same framework \cite{ToninhoCP14}
handles infinite data by encoding them as coinductive sessions.

Lindley and Morris \cite{DBLP:conf/icfp/LindleyM16} combine
recursive and co-recursive data types with communication primitives.
They have fold and unfold over both recursive and corecursive session types instead of a general
fixed point operator.
The constructors ${\sf in}$ and ${\sf out}$ witness the
isomorphism of recursion and corecursion. 
The operational semantics 
is call-by-value, but sending code is allowed  because fold and unfold are values.

 SSCC~\cite{clmrv14}
offers an explicit primitive to deal with streams.
Our language enables the modelling of more intricate interactions
between infinite data structures and infinite communications.
Besides, the type system of SSCC considers only finite sessions types
and does not guarantee progress of processes.

Following~\cite{Nakano00:lics},
we use a modal operator $\bullet$ to restrict the application of the
fixed point operator and  exclude  degenerate forms of divergence.
This paper is an improvement over past typed lambda calculi with a
temporal modal operator in two respects.
Firstly, we do not need any subtyping relation as in
\cite{Nakano00:lics} and secondly \sid\ programs are not cluttered
with constructs for the introduction and elimination of individuals of
type $\bullet$ as in
\cite{KrishnaswamiB11,severidevriesICFP2012,KrishnaswamiBH12,AM13,DBLP:conf/popl/CaveFPP14,CBGB15}.
A weak criterion to ensure productivity of infinite data is the
\emph{guardedness condition} \cite{Coquand93}.
We do not need such condition because we can type more normalising
expressions (such as $\dspf$ in~\eqref{eq:dspgood}) using the modal
operator $\tbullet$.

Futures originated in functional programming as annotations for implicitly
parallelising programs~\cite{H85}. Different operational semantics for an idealised functional language with futures are discusses in~\cite{FlanaganF99}. 

The papers more related to  ours  are~\cite{NSS06} and~\cite{SabelS11}. The call-by-value calculus of~\cite{NSS06} models Alice~\cite{RBTBS06}, a concurrent extension of standard ML~\cite{MTHM97}, where synchronisation is  based on futures as placeholders for values. A linear type system assures safety. The call-by-need $\lambda$-calculus in~\cite{SabelS11} provides a semantic foundation for the concurrent Haskell extended with futures. It shows the correctness of several program transformations using contextual semantics. Our calculus shares threads with these calculi. A main difference is the way in which the threads interact: through thread names and cells in~\cite{NSS06} and through shared memory in form of Haskell's mutable variable and a global heap of shared expressions   in~\cite{SabelS11}.   Recursion is obtained by allowing the body of a thread to contain the thread name in ~\cite{NSS06} and by recursive heaps in~\cite{SabelS11}.
% A key novelty of our calculus is the delay type constructor.   

%Futures originated in functional programming and related paradigms for
%parallelising a program \cite{FlanaganF99}.
%%
%The call-by-need $\lambda$-calculus with futures in~\cite{SabelS11} is
%used for studying contextual equivalence and has no type system.
%The call-by-value $\lambda$-calculus in~\cite{NSS06}
%has several variants of futures 
%which are particular case of ours. Besides,  their type system
%does not control the use of dual endpoints. 

In the session calculi literature, the word \lq\lq progress\rq\rq\ has
two different meanings.
Sometimes it is synonym of deadlock freedom~\cite{BCDLDY08}, at other
times it means lock freedom, i.e. that each offered communication in
an open session eventually happens~\cite{DY11,Padovani14B,CDYP2016}.
Reachable  %Typed 
\sid\ processes cannot be stuck, and if they do not terminate
they %communicate and/or 
generate new threads infinitely often.  This
means that the property of progress satisfied by our calculus is
stronger than that of~\cite{BCDLDY08} and weaker than that
of~\cite{DY11,Padovani14B,CDYP2016}.

\section{Conclusions}
\label{section:conclusions}

%This work addresses the problem of studying the interaction between
%communications and infinite data structures by means of a calculus
%that combines sessions with lazy evaluation.
This paper %works 
studies the interaction between
communications and infinite data structures by means of a calculus
that combines sessions with lazy evaluation.
A distinguished feature of \sid\ is the possibility of modelling
computations in which infinite communications interleave with the
production and consumption of infinite data (\cf the examples in
\cref{section:Introduction}).
Our examples considered infinite streams for simplicity. However, more
general infinite data structures can be handled in \sid. An evaluation
of the expressiveness of \sid\ in dealing with (distributed)
algorithms based on such structures is scope for future
investigations.

The typing discipline we have developed for \sid\ guarantees
normalisation of expressions with a type other than $\infinite$ and
progress of (reachable) processes, besides the standard properties of
sessions (communication safety, protocol fidelity, determinism). The
type system crucially relies on a modal operator $\bullet$ which has
been used in a number of previous
works~\cite{Nakano00:lics,KrishnaswamiB11,severidevriesICFP2012,DBLP:conf/popl/CaveFPP14}
to ensure productivity of well-typed expressions.
In this paper, we have uncovered for the first time some intriguing
interactions between this operator and the typing of impure
expressions with the monadic $\mkbasictype{IO}$ type constructor.
Conventionally, the type of $\future$ primitive is simply
$\tio\Type \arrow \tio\Type$ and says nothing about the semantics of
the primitive itself. In our type system, the type of $\future$
reveals its effect as an operator that turns a delayed computation
into another that can be performed immediately, but which produces a
delayed result.

%
% Our work was motivated by some interesting patterns of computations
% over infinite data structures and our results have a twofold
% interpretation.
% %
% Firstly, \sid\ can be regarded as the core of a lazy functional
% language featuring (higher-order) communications of infinite data
% structures.
% %
% Secondly, \sid\ is the first language to contemplate infinite data in
% session types (as discussed in \Cref{section:Introduction} this is
% not done elsewhere to the best of our knowledge).

As observed at the end of \Cref{section:relatedwork} and formalised
in \cref{theorem:reducestoreturnI}, our notion of progress sits
somehow in between deadlock and lock freedom.
It would be desirable to strengthen the type system so as to guarantee
the (eventual) execution of all pending communications and exclude,
for instance, the degenerate examples discussed in
\cref{sec:propproc}. This is relatively easy to achieve in conventional
process calculi, where expressions only consist of names or ground
values~\cite{BCDLDY08,Padovani14B,CDYP2016}, but it is far more
challenging in the case of \sid{}, where expressions embed the
$\lambda$-calculus.
We conjecture that one critical condition to be imposed is to forbid
postponing linear computations, namely restricting the application of
{\refrule\introbullet} to non-linear types. Investigations in this
direction are left for future work.

Another obvious development, which is key to the practical
applicability of our theory, is the definition of a type inference
algorithm for our type system. 
% In this respect, the modal operator
%$\bullet$ is challenging to deal with because it is intrinsically
%non-structural, not corresponding to any expression form in the
%calculus.
First steps in this direction have already been taken
in \cite{severifossacs2017}
 by solving type inference  
 for the pure part of \sid \ (without $\tio$ and concurrency)
 combining unification of types with  integer linear programming.

\subsection*{Acknowledgments.} 
We are grateful to the anonymous reviewers of COORDINATION'16 and of
LMCS for their useful suggestions, which led to substantial
improvements.

\bibliographystyle{abbrv}
\bibliography{references}
% \mmic{}{We must agree on the way conferences and journal are cited, I propose short names for conferences (in some references we have only those) and full names for journal (some abbreviations are arbitrary). We must also decide if the titles must be all with initial uppercase characters or not, for sure Haskell needs it.}{}

\iflong
\newpage
\appendix

\section{Proof of \cref{theorem:soundness}}%{Detailed Proofs of Normalisation}
\label{appendix:normalisation}

\begin{lemma}\label{lem:B}
\begin{enumerate}
\item\label{lemma:closedunderreductionandexpansion}
Let $\Expression\red\Expression'$. 
Then\comma 
$\Expression\in\indti\Type\ind$ iff $\Expression'\in\indti\Type\ind$
for all $\ind \in \natset$ and type $\Type$.
\item \label{lemma:constants}
If $\Constant:\Type$ and $\Type \in \typeof(\Constant)$, then
 $\Constant\in \bigcap_{\ind \in \natset} \indti{\Type}{\ind}$.
\end{enumerate}
\end{lemma}
\begin{proof} (\cref{lemma:closedunderreductionandexpansion}). 
By induction on $(\ind, \typerank(\Type))$.

(\cref{lemma:constants}).  We only consider  the case  $\Constant = \bind$ and prove
that 
\[
 \bind \in \indti{\tio\Type\arrow(\Type\larrow\tio\TypeS)\larrow\tio\TypeS}{\ind}\]
  Suppose  
$\Expression_1 \in \indti{\tio\Type}{\indj}$
 and 
$\Expression_2 \in \indti{\Type\larrow\tio\TypeS}{\indj}$
for $\indj \leq \ind$. We show that $\bind ~\Expression_1~\Expression_2\in\indti{\tio\TypeS}{\ind}$. 
By definition of $\indti{\tio\Type}{\indj}$\comma we have three cases:
\begin{enumerate}

\item Case $\Expression_1 \in \WNVAR$. Hence
$  \bind~\Expression_1~\Expression_2 \red^{*}
    \bind~\Context[\varX]~\Expression_2
    $. 
Taking $\Context'[x] = \bind~\Context[\varX]~\Expression_2$\comma we have
that 
$
\bind~\Expression_1~\Expression_2 \in \WNVAR$ and 
$\WNVAR\subseteq \indti{\tio\TypeS}{\indj}
$
by \ri{\cref{lem:A}}{\cref{lemma:wnvar}}.
\item Case $\Expression_1 \in \HHPC$. 
 Hence $\Expression_1 \red^{*}
\PContext[\Expression_0] $ and 
$\Expression_0 \in \{
\Send{\Channel^{\Polarity}}{\Expression'_1},
          \Receive{\Channel^{\Polarity}},
          \Open{\Channel},
\Future \Expression'_1 \}$.  Then\comma
 \[  \bind~\Expression_1~\Expression_2 \red^{*} \bind~\PContext[\Expression_0]~\Expression_2  \in \HHPC
 \] which implies $\bind~\PContext[\Expression_0]~\Expression_2  \in\indti{\tio{\TypeS}}{\indj}$ by definition of $\indti{\tio{\TypeS}}{\indj}$.
By \cref{lemma:closedunderreductionandexpansion} we conclude that 
$ \bind~\Expression_1~\Expression_2  \in \indti{\tio\TypeS}{\indj}$.

\item  Case $ \Expression_1 \red^* \return~ \Expression'_1$ and 
 $\Expression'_1\in \indti{\Type}{\indj}$. 
 This gives $\Expression_2~ \Expression'_1\in\indti{\tio\TypeS}{\indj}$.  Since
 \[  \bind~\Expression_1~\Expression_2 \red^{*}
    \bind~(\return~\Expression'_1)~\Expression_2
     \red     \Expression_2~\Expression'_1
     \]
  we conclude that $ \bind~\Expression_1~\Expression_2  \in\indti{\tio\TypeS}{\indj}$ 
   by \cref{lemma:closedunderreductionandexpansion}.
   \qedhere
\end{enumerate}
\end{proof}

\begin{lemma}\label{lem:C}
\begin{enumerate}
\item\label{lemma:modelscontext}
If $\funsubst \modelsi \TypeContext_1 + \TypeContext_2$,
then
$\funsubst \modelsi \TypeContext_1$ and
$\funsubst \modelsi \TypeContext_2$.
\item \label{lemma:modelscontextlessindex}
If $\funsubst \modelsi \TypeContext$,  then
$\funsubst \modelsj \TypeContext$ for all $\indj \leq \ind$.
\end{enumerate}
\end{lemma}
\begin{proof}
\cref{lemma:modelscontext} is an easy consequence of~\cref{definition:modelsi}.

\cref{lemma:modelscontextlessindex}  follows from \ri{\cref{lem:A}}{\cref{lemma:monotonicityofinterpretation}}.
\end{proof}

\begin{proof}[Proof of ~\cref{theorem:soundness}]
We prove that $\TypeContext \modelsi \Expression:\Type$ for all $\ind \in \mathbb{N}$
by induction on  $\wte{\TypeContext}\Expression\Type$.
 We only show some interesting cases.

%\item 
\mycase{Rule  \refrule{\const}} 

\noindent
It  
follows from \ri{\cref{lem:B}}{\cref{lemma:constants}}.

\mycase{Rule \refrule\introbullet}
The derivation ends with the rule:
\[
\inferrule[\defrule\introbullet]{
        \wte{\TypeContext}{\Expression}\Type
        }{
        \wte{\TypeContext}{\Expression}{\tbullet\Type}
        }
\]

\noindent 
Suppose $i=0$. Then\comma 
\[
\funsubst(\Expression) \in \indti{\Type} {0}= \EE
\]
Suppose $i >0$ and  $\funsubst \modelsi \TypeContext$.
It follows from \ri{\cref{lem:C}}{\cref{lemma:modelscontextlessindex}} 
 that
$\funsubst \models_{i-1} \TypeContext$.
By induction hypothesis\comma 
  $\TypeContext \modelsj \Expression:\Type$ for all $j \in \mathbb{N}$.   
In particular\comma 
$\TypeContext \models_{i-1} \Expression:\Type$.
Hence\comma 
$\funsubst(\Expression) \in \indti{\Type} {\ind-1}$
and 
\[
\funsubst(\Expression) \in \indti{\Type} {\ind-1}= \indti{\tbullet{\Type}}{\ind}
\]

 \mycase{Rule \refrule\elimarrow}
 
 \noindent The  derivation   ends with the rule:   
  \[
  \inferrule{
    \wte{\TypeContext_1}{\Expression_1}{\tbullet[n](\TypeS \arrow \Type)}
    \\
    \wte{\TypeContext_2}{\Expression_2}{\tbullet[n]\TypeS}
  }{
    \wte{\TypeContext_1 + \TypeContext_2}{\Expression_1\Expression_2}{\tbullet[n]\Type}
  }
  \]
  with $\TypeContext = \TypeContext_1 + \TypeContext_2$ and
  $\Expression = \Expression_1\Expression_2$.
By induction hypothesis\comma for all $\ind \in \natset$
\begin{equation}
\label{equation:inductiohypothesisoperator}
\TypeContext_1 \modelsi \Expression_1: \tbullet[n](\TypeS \arrow \Type)
\end{equation}
\begin{equation}
\label{equation:inductiohypothesisargument}
\TypeContext_2 \modelsi\Expression_2: \tbullet[n]\TypeS
\end{equation}
We have two cases:
\begin{enumerate}
\item Case $\ind < n$.
By \ri{\cref{lemma:interpretationofmanybullets}}{\cref{lemma:interpretationofmanybullets1}}\comma
$\indti{\tbullet[n] \Type} {\ind} = \EE$. We trivially  get  
\[
\funsubst (\Expression_1 \Expression_2) \in
\indti{\tbullet[n] \Type} {\ind}
\]
\item Case $\ind \geq n$. Suppose that $\funsubst \modelsi \TypeContext$. 
 It follows from \ri{\cref{lem:C}}{\cref{lemma:modelscontext}}
  that 
$\funsubst \modelsi \TypeContext_1$ and
$\funsubst \modelsi \TypeContext_2$.  
\begin{equation}
\label{equation:interpretationexpone}
\begin{array}{lll}
\funsubst(\Expression_1) & \in & \indti{\tbullet[n](\TypeS \arrow \Type)} \ind 
\mbox{ by \eqref{equation:inductiohypothesisoperator} } \\
              & =   & \indti{(\TypeS \arrow \Type)} {\ind-n} 
\mbox{ by \ri{\cref{lemma:interpretationofmanybullets}}{\cref{lemma:interpretationofmanybullets2}}} 
\end{array}
\end{equation}
\begin{equation}
\label{equation:interpretationexptwo}
\begin{array}{lll}
\funsubst(\Expression_2) & \in & \indti{\tbullet[n]\TypeS} \ind
\mbox{ by \eqref{equation:inductiohypothesisargument}} \\
   & =   & \indti{\TypeS} {\ind-n}
   \mbox{ by \ri{\cref{lemma:interpretationofmanybullets}}{\cref{lemma:interpretationofmanybullets2}}} 
   \end{array}
\end{equation}
By Definition of $\indti{(\TypeS \arrow \Type)}{\ind-n} $ and 
(\ref{equation:interpretationexpone})\comma 
there are 
two possibilities:
\begin{enumerate}
\item Case $\funsubst(\Expression_1) \in  \WNVAR $. Then\comma
\begin{equation}
\label{equation:applicationsoundness}
\funsubst(\Expression_1 \Expression_2) =
\funsubst(\Expression_1) \funsubst(\Expression_2)
\red^{*} \Context[\varX] \funsubst(\Expression_2)
\end{equation}
Hence\comma 
\[
\begin{array}{lll}
\funsubst(\Expression_1 \Expression_2) & \in \WNVAR & \mbox{by (\ref{equation:applicationsoundness})}  \\
&  \subseteq \indti{\tbullet[n] \Type}{\ind} &\mbox{ by \ri{\cref{lem:A}}{\cref{lemma:wnvar}}}.
\end{array}
\]
\item Case $ \funsubst(\Expression_1) \red^{*} \lambda x. \Expression' $
or $ \funsubst(\Expression_1) \red^{*} \Context[\Constant]$.
We also have that
\[
 \funsubst(\Expression_1)\Expression'' \in \indti{\Type}{\ind-n}  \ \ 
 \forall \Expression'' \in \indti{\TypeS}{\ind-n} 
 \]
 In particular %from  
 (\ref{equation:interpretationexptwo})\comma 
  implies 
\[
\funsubst (\Expression_1 \Expression_2) = \funsubst(\Expression_1) \funsubst(\Expression_2)
 \in \indti{\Type} {\ind -n}
\]
Since $\indti{\Type} {\ind -n}= \indti{\tbullet[n] \Type}{\ind}$ by  \ri{\cref{lemma:interpretationofmanybullets}}{\cref{lemma:interpretationofmanybullets2}}\comma we are done.
\end{enumerate}
\end{enumerate}

\mycase{Rule \refrule\introarrow}

\noindent The derivation   ends with the rule:   
    \[
    \inferrule{
      \wte{\TypeContext,\varX: \tbullet[n]\Type }{\Expression}{\tbullet[n]\TypeS}
    }{
      \wte{\TypeContext}{\Fun \varX \Expression}{\tbullet[n](\Type \arrow \TypeS)}
    }
    \]
 By induction hypothesis\comma for all $\ind \in \natset$ 
 \begin{equation}
\label{equation:inductionhypothesisbodyabstraction}
 \TypeContext,\varX: \tbullet[n]\Type \modelsi {\Expression} : {\tbullet[n]\TypeS}
 \end{equation}    
    We have two cases:   
    \begin{enumerate}    
    \item Case $\ind < n$. By \ri{\cref{lemma:interpretationofmanybullets}}{\cref{lemma:interpretationofmanybullets1}}\comma
$\indti {\tbullet[n](\Type \arrow \TypeS)} \ind = \EE$. We trivially 
 get 
     \[
        \funsubst (\Fun \varX \Expression) \in             \indti {\tbullet[n](\Type \arrow \TypeS)} \ind 
           \]    
    \item  Case $\ind \geq n$.  Suppose that $\funsubst \modelsi \TypeContext$. 
    By \ri{\cref{lemma:interpretationofmanybullets}}{\cref{lemma:interpretationofmanybullets2}}\comma
     it is enough  %we have 
    to prove 
    that 
    \[
    \funsubst (\Fun \varX \Expression) \in \indti {\Type \arrow \TypeS}{\ind - n}
    \]
    For this\comma suppose   
    $\ExpressionF \in  \indti{\Type}{\indj}$
    for $\indj \leq \ind -n$.
    % \indti{\tbullet[n]\Type}{\ind}$.  
    We consider the substitution function defined as 
    $\funsubst_0 =  \funsubst \cup \{ (\varX, \ExpressionF) \}$.
    We have that  
    \begin{equation}
    \label{equation:funsubstj+n}
    \funsubst_0\models_{j+n} \TypeContext, \varX :\tbullet[n]\Type 
    \end{equation}
    because
    \begin{enumerate}
    \item $\funsubst_0(\varX) = 
    \ExpressionF \in  \indti{\Type}{\indj} = \indti{\tbullet[n]
    \Type}{\indj+n}$
    by \ri{\cref{lemma:interpretationofmanybullets}}{\cref{lemma:interpretationofmanybullets2}}.
    
    \item $\funsubst_0 \models_{\indj+n} 
    \TypeContext$
    by \ri{\cref{lem:C}}{\cref{lemma:modelscontextlessindex}} 
    and the fact that $\funsubst_0 \modelsi \TypeContext$.
    \end{enumerate}
  It follows from \eqref{equation:inductionhypothesisbodyabstraction} and
  \eqref{equation:funsubstj+n} that
  \begin{equation}
  \label{equation:rhozero}
  \funsubst_0 (\Expression) 
        \in \indti{\tbullet[n]\TypeS}{\indj+n}
  \end{equation}
   Therefore\comma we  obtain   
   \[
     \begin{array}{lll}
     (\Fun \varX \Expression) \ExpressionF \red
       \funsubst(\Expression)\subst{\ExpressionF}{\varX}
       = \funsubst_0 (\Expression) 
        & \in \indti{\tbullet[n]\TypeS}{\indj+n} & \mbox{by 
        (\ref{equation:rhozero}}) \\
        & = \indti{\TypeS}{\indj} & \mbox{by  \ri{\cref{lemma:interpretationofmanybullets}}{\cref{lemma:interpretationofmanybullets2}}}
        \end{array}
      \]
      By \ri{\cref{lem:B}}{\cref{lemma:closedunderreductionandexpansion}}\comma we conclude
      \begin{equation}
        \begin{array}{lll}
     (\Fun \varX \Expression) \ExpressionF 
        & \in \indti{\TypeS}{\indj}.
        \end{array}
        \tag*{\qEd}
        \end{equation}
    \end{enumerate}
    \def\popQED{}
\end{proof}

\section{Proof of \cref{theorem:invariance}}
\label{appendix:SRprocesses}

We use $\thread \varX \Expression \subprocesseq_1 \Process$ as short for 
$\thread \varX \Expression\subprocesseq \Process$ 
%$\thread \varX \Expression$ occurs only once in $\Process$.  
 and there is only one thread named $\varX$ in $\Process$.  
 If $\thread \varX \Expression\subprocesseq_1 \Process$ 
we  denote by 
$\Process  \,\replace{\thread \varX \Expression}{\ProcessQ}$ the replacement of    the  {\em unique} 
 occurrence of the thread $\thread \varX \Expression$
by the process $\ProcessQ$ in the process  $\Process$. 
%
% Assume $\Expression$ occurs only once in $\Process$.
%We denote by $\Process \replace{\Expression}{\ExpressionF}$ the replacement of the  unique occurrence of  the  expression 
%$\Expression$ by the expression $\ExpressionF$
%in the process $\Process$.  
In particular, if $\thread \varX {\PContext[\send~\Channel^{\Polarity}~\Expression]}\subprocesseq_1 \Process$ we will abbreviate \begin{center}$\Process  \,\replace{\thread \varX {\PContext[\send~\Channel^{\Polarity}~\Expression]}}{\thread \varX {\PContext[\return~\Channel^{\Polarity}}]}$ as $\Process  \,\repsend $.\end{center}
Similarly, if $\thread \var
{\PContext[\receive~\Channel^{\Polarity}]}\subprocesseq_1 \Process$  we will abbreviate \begin{center}$\Process \,\replace{\thread \var
 {\PContext[\receive~\Channel^{\Polarity}]}}
 {\thread \var {\PContext[\return~\Pair{\Expression}{\Channel^{\Polarity}}]}}$ as $\Process  \,\repreceive $. \end{center} 
 Notice that in both cases $\PContext$ and $\Channel^{\Polarity}$ are uniquely determined by the body of the thread named $\var$, while the expression $\Expression$ occurs 
for $\send$ but not for $\receive$.  Writing $\Expression$ as argument of both $\repSend$ and $\repReceive$
 allows us to easily express the exchanged message. 
%  {\thread \varX {\PContext[\return~\Channel^{\Polarity}}]}
%We will use the following abbreviations througt the proofs.
%\[
% \begin{array}{lcl}
%\repsend & = & \replace{\thread \varX {\PContext[\send~\Channel^{\Polarity}~\Expression]}}
%  {\thread \varX {\PContext[\return~\Channel^{\Polarity}}]}
%\\
% \repreceive & = & \replace{\thread \var
% {\PContext[\receive~\Channel^{\Polarity}]}}
% {\thread \var {\PContext[\return~\Pair{\Expression}{\Channel^{\Polarity}}]}}
%\end{array}
% \] 
These replacements are useful to find the right re-arrangements
 of threads which are derivable after applying
  the rule \defrule{r-comm} to $\Process$.  Informally, the derivation of $\CA \models \Process$ must contain a sub-derivation of the shape
  \[
    \inferrule[\refrule{wp-par}]{
     \models \Process_1
      \\
     \models \Process_2
    }{
     \models\Process
       _1\parop\Process_2
    }
    ~~
   \Pol (\ProcessP_1)\setminus\set{\Channel^\Polarity}\indep\CB \Pol (\Process_2)\setminus\set{\Channel^{\co\Polarity}}
    \]
    with $\thread \varX \PContext[\send~{\Channel^{\Polarity}}~\Expression] \subprocesseq_1 \Process_1$
 and 
 $\thread \varY \PContext'[\receive~\Channel^{\co \Polarity}] \subprocesseq_1 \Process_2$. We build the desired process by 
 replacing   
   $\thread \varX 
   \PContext[\send~{\Channel^{\Polarity}}~\Expression]$ with 
  $\thread \varX \PContext[\return~{\Channel^{\Polarity}}] \parop \Process_2 \, \repreceivec$ 
     in $\Process_1$.
Consider the processes defined in  
\cref{example:r-comm}  and the reduction  $\Process \red \ProcessQ$
 using rule \defrule{r-comm}.  Let
\[
 \Process_1 = \thread \varX \send~\Channel^{+}~z \parop 
  \thread \varZ \return~1 \qquad
 \Process_2  =  \thread \varY \receive~\Channel^{-} 
 \]
 then $\Process = \Process_1 \parop \Process_2$.
 The  process $\ProcessQ'$ such that 
 $\ProcessQ' \equiv \ProcessQ$  and $\models{\ProcessQ'}$ 
 is obtained by replacing  in $ \Process_1$ the thread 
$ \thread \varX \send~\Channel^{+}~z $
 by the process  
  $\thread \varX \return~{\Channel^{+}} \parop 
  \thread \varY {\return {\Pair{z}{\Channel^{-}}}}$, i.e. 
 \[
 \ProcessQ' =
 (\thread \varX \return~{\Channel^{+}} \parop 
  \thread \varY {\return {\Pair{z}{\Channel^{-}}}}) \parop
 \thread \varZ \return~1
 \]

 \begin{lemma}
\label{lemma:replacementofreceive}
Let $\thread \var \PContext[\receive~\Channel^{\Polarity}]  \subprocesseq_1  \Process$  
%$\receive~\Channel^{\Polarity}$ occur once in $\Process$ in the thread named $\var$ 
and 
$\var^+\not\in\Pol(\Expression)$.
 If  
 $\CA \models \Process$
 and $\Pol(\Expression) \indep \Pol(\Expression)$ and
 $\Pol (\ProcessP) \indep \Pol(\Expression)$,
 then 
 $\models 
 \Process \, \repreceive
 %\replace{\receive~\Channel^{\Polarity}}{\return~\Pair{\Expression}{\Channel^{\Polarity}}}
 $.
 \end{lemma}
 
 \begin{proof} 
 By induction on the derivation of $\CA \models \Process$. 
 \end{proof}

\begin{lemma}
\label{lemma:findingtherightcouple} 
 Let 
$\thread \varX \PContext[\send~{\Channel^{\Polarity}}~\Expression] \subprocesseq_1 \Process$ 
 and 
 $\thread \varY \PContext'[\receive~\Channel^{\co \Polarity}] \subprocesseq_1 \ProcessQ$  and 
%$\Channel^\Polarity\not\in\Pol(\Expression)$. 
%  and  %Suppose also that 
% the channels 
 $\Channel^{\Polarity}$
% and $\Channel^{\co \Polarity}$  
occurs 
only once in $\Process$. 
 If $ \models \Process\parop \ProcessQ$,
 then
 there  is $\ProcessR$ such that  $ \models \ProcessR$ and
 %exists 
 $$\ProcessR \equiv \Process  \, \repsend \parop \ProcessQ \, \repreceivec$$ 
 \end{lemma}

\begin{proof}
It follows from 
$ \models \Process\parop \ProcessQ$ 
% and taking into account the hypotheses on $\Process$ and $\ProcessQ$ 
that 
 $ \models \Process$ and $ \models \ProcessQ$ and 
  \begin{equation}
\label{equation:b1}\Pol (\ProcessP)\setminus\set{\Channel^\Polarity}\indep\CB \Pol (\ProcessQ)\setminus\set{\Channel^{\co\Polarity}}\end{equation}
and  
\begin{equation}
\label{equation:b2} \Pol(\Expression)\indep \Pol(\Expression)
\end{equation}
We do induction on $ \models \Process$.

Suppose the last rule in the derivation is
\refrule{wp-thread}. Then 
\[
\inferrule[\refrule{wp-thread}]{
       }{
      \models \thread\var\PContext[\send~{\Channel^{\Polarity}}~\Expression]
    }
    ~~
    \begin{lines}[c]
     \var^+\not\in\Pol(\PContext[\send~{\Channel^{\Polarity}}~\Expression])
\\
          \Pol(\PContext[\send~{\Channel^{\Polarity}}~\Expression]) \indep   \Pol(\PContext[\send~{\Channel^{\Polarity}}~\Expression])
    \end{lines} 
\]
%It is easy to obtain the following derivation 
For  the
thread obtained applying  the replacement
$\repsend$  we derive:  %as follows.
\begin{equation}
\label{equation:b2a}
\inferrule[\refrule{wp-thread}]{
    }{
      \models 
      \thread\var\PContext[\return~{\Channel^{\Polarity}}] }
    ~~
    \begin{lines}[c]
     \var^+\not\in\Pol(\PContext[\return~{\Channel^{\Polarity}}])\\
          \Pol(\PContext[\return~{\Channel^{\Polarity}}]) \indep  
           \Pol(\PContext[\return~{\Channel^{\Polarity}}])
    \end{lines} 
\end{equation}
% Since $\Channel^{\Polarity}$ occurs only once
%in $\Process$, it cannot occur in $\Expression$ and
%we have that 
From (\ref{equation:b1}) and $\Channel^\Polarity\not\in\Pol(\Expression)$ we get 
$
\Pol (\ProcessQ ) \indep \Pol(\Expression)$. 
It follows from this, 
  (\ref{equation:b2})
and \cref{lemma:replacementofreceive} 
that 
\begin{equation}
\label{equation:b2b}
\models \ProcessQ  \, \repreceivec
\end{equation} 
The  condition $\Pol(\PContext[\send~{\Channel^{\Polarity}}~\Expression]) \indep   \Pol(\PContext[\send~{\Channel^{\Polarity}}~\Expression])$ 
 implies 
$\Pol(\PContext[\return~{\Channel^{\Polarity}}]) \indep \Pol (\Expression)$. This together with (\ref{equation:b1}) gives $\Pol(\PContext[\return~{\Channel^{\Polarity}}]) \setminus \{\Channel^{\Polarity}\}\indep\Pol (\ProcessQ \,\repreceivec)\setminus\set{\Channel^{\co\Polarity}}$.
 % $\Pol(\PContext[\return~{\Channel^{\Polarity}}]) \indep \Pol (\Expression)$.
Applying  \refrule{wp-par} to (\ref{equation:b2a})
and (\ref{equation:b2b}) 
we  derive: 
\[
\models \thread\var\PContext[\return~{\Channel^{\Polarity}}]
\parop \ProcessQ \, \repreceivec
\]

Suppose the last rule in the derivation is  
   \begin{equation}
\label{equation:b2c}
    \inferrule[\refrule{wp-par}]{
      \models \Process_1
      \\
      \models \Process_2
    }{
      \models\Process
       _1\parop\Process_2
    }
    ~~
   \Pol (\ProcessP_1)\setminus\set{\Thing^\Polarity}\indep \Pol (\Process_2)\setminus\set{\Thing^{\co\Polarity}}
    \end{equation}
    and  $\thread \varX \PContext[\send~{\Channel^{\Polarity}}~\Expression] \subprocesseq_1 \Process_1$.
 By induction hypothesis\comma
 $ \models \ProcessR_1
   $ for some \[\ProcessR_1 \equiv  \Process_1  \, \repsend  \parop \ProcessQ  \, \repreceivec\]
It follows from 
 (\ref{equation:b1})   and   $\Channel^\Polarity\not\in\Pol (\Process_2)$ 
 that 
$\Pol(\ProcessQ) \indep \Pol(\Process_2)$. 
 Since $\Pol (\ProcessR_1)=\Pol (\ProcessP_1)\cup\Pol (\ProcessQ)$
and using the side condition of (\ref{equation:b2c}), 
%$\Pol (\ProcessP_1)\setminus\set{\Thing^\Polarity}\indep \Pol (\Process_2)\setminus\set{\Thing^{\co\Polarity}}$, 
we get 
 $\Pol (\ProcessR_1) \setminus\set{\Thing^\Polarity} 
 \indep\Pol (\ProcessP_2)\setminus\set{\Thing^{\co\Polarity}}$.
  We  can apply \refrule{wp-par} and derive 
  \[
      \models \ProcessR_1 \parop \Process_2
   \] 
        Clearly, 
        \begin{equation}
      \ProcessR_1 \parop \Process_2 \equiv  (\Process_1 \parop \Process_2) \, \repsend  \parop \ProcessQ  \, \repreceivec
      \tag*{\qEd}
      \end{equation}
      \def\popQED{}
\end{proof}

\begin{lemma}
\label{lemma:invariancecomm}
Let  
%\bchm $\thread \varX \PContext[\send~{\Channel^{\Polarity}}~\Expression]\parop\thread \varY \PContext'[\receive~\Channel^{\co \Polarity}]\subprocesseq_1 \Process$ 
% \echm 
 $\thread \varX \PContext[\send~{\Channel^{\Polarity}}~\Expression] \subprocesseq_1 \Process$ 
 and 
 $\thread \varY \PContext'[\receive~\Channel^{\co \Polarity}] \subprocesseq_1 \Process$ 
 and 
 %$\Channel^\Polarity\not\in\Pol(\Expression)$. 
% and  
% the channels 
 $\Channel^{\Polarity}$
% $\Channel^{\co \Polarity}$  
 occurs only once in $\Process$. 
 If  $\CA \models \Process$,
 then there exists $ \ProcessQ $ such that
 $\CA \models \ProcessQ $
 and 
 $ \ProcessQ  \equiv \Process \, \repsend  \, \repreceivec$.
 
\end{lemma}

\begin{proof}
%This is proved by 
 By  induction on $\CA \models \Process$.
We only show the  most interesting case:
 \[
    \inferrule[\refrule{wp-par}]{
     \models \Process_1
      \\
     \models \Process_2
    }{
     \models\Process
       _1\parop\Process_2
    }
    ~~
   \Pol (\ProcessP_1)\setminus\set{\Channel^\Polarity}\indep\CB \Pol (\Process_2)\setminus\set{\Channel^{\co\Polarity}}
    \]
    By Lemma \ref{lemma:findingtherightcouple}, 
     there is $\ProcessQ $ such that
    $\models \ProcessQ $ and 
    \begin{equation}
      \ProcessQ \equiv \Process_1 \, \repsend \parop \Process_2 \, \repreceivec
 = (\Process_1  \parop \Process_2)  \, \repsend  \,\repreceivec
  \tag*{\qEd}
    \end{equation}
    \def\popQED{}
\end{proof}

We now give some lemmas in order to 
find  a  %the 
right re-arrangements
 of  the  threads   which are derivable after applying
  the rule \defrule{r-return} to $\new \var\Process$. 
  Informally, if $\var^+\in\Pol (\ProcessP)$, the derivation of $\CA \models \Process$ must contain a sub-derivation of the shape
  \[
    \inferrule[\refrule{wp-par}]{
     \models \Process_1
      \\
     \models \Process_2
    }{
     \models\Process
       _1\parop\Process_2
    }
    ~~
   \Pol (\ProcessP_1)\setminus\set{\var^-}\indep\CB \Pol (\Process_2)\setminus\set{\var^+}
    \]
    with $\thread \varX {\return~\Expression}\subprocesseq_1 \Process_1$. 
If $\var^{+} \not \in \Pol (\Process_1)$, then     
    the desired process is  obtained by replacing 
     $\thread \varX {\return~\Expression}$ with $\Process_2 \subst\Expression\var$   in $\Process_1$. 
     Otherwise we need to parenthesise differently $\Process
       _1\parop\Process_2$ in order to satisfy this condition. 
  Consider the 
 process $\Process$ 
of 
 \cref{example:r-return}
 which we write as  
 $\Process = \Process_1 \parop \Process_2$
 where
 \[
 \begin{array}{ll}
 \Process_1 = (\thread \varX \return~\Pair{\varZ_1}{\varZ_2}
 \parop \thread {\varZ_1}  \return~\varZ_2) \parop
 \thread {\varZ_2} \return~1 \\
 \Process_2 = 
 \thread \varY \send~\Channel^{+} \varX \parop
 \thread u \receive~\Channel^{-}
 \end{array}
 \]
 Let $\new\var\Process \red \ProcessQ$
 using rule \defrule{r-return}. 
 The process $\ProcessQ'$ such that 
 $\ProcessQ' \equiv \ProcessQ$  and $\models{\ProcessQ'}$ 
% when the reduction  $\Process_0 \red \ProcessQ$
% uses rule \defrule{r-return}
 is obtained by replacing  in $ \Process_1$ the thread $\thread \varX{\return~\Pair{\varZ_1}{\varZ_2}}$ %with the return
 by the process 
 $\ProcessQ_1 =\Process_2 \subst{\Pair{\varZ_1}{\varZ_2}}{\var} $, i.e. 
 \[
 \ProcessQ' = 
 \Process_1 \, 
  \replace{\thread \varX{\return~\Pair{\varZ_1}{\varZ_2}} }
 { \ProcessQ_1
 }
 \]
 
\begin{lemma} 
\label{lemma:replacementofthread}
Let $\thread \varX\return~\Expression\subprocesseq_1\Process$  
and 
$ \models \Process$.
 If $\CC \models \ProcessQ$
and 
%$\Pol(\Process)\setminus\set{\varX^{-}} \indep \Pol(\ProcessQ)$, 
\mbox{$\Pol(\Process)\setminus\set{\varX^{-}} \indep \Pol(\ProcessQ)
\setminus \Pol(\Expression)$,} 
%and $\Pol(\Expression) \subseteq  \Pol(\ProcessQ)$,
then 
$\models \Process \,\replace{\thread \varX \return~\Expression}{\ProcessQ}$.
\end{lemma}

\begin{proof}
By induction on  
% the derivation of 
$\models \Process$.
Suppose $\Process =  \thread\var\return~\Expression$
and the derivation of $\models \Process$ is:  
\[
 \inferrule[\refrule{wp-thread}]{
    }{
       \models \thread\var\return~\Expression
    }
    ~~
    \begin{lines}[c] 
    \var^+\not\in\Pol(\Expression)\\
    \Pol(\Expression)\indep\Pol(\Expression)
    \end{lines} 
\]
In this case $\models 
\Process \,\replace{\thread \varX\return~ \Expression}{\ProcessQ}$,
since 
$\Process \,\replace{\thread \varX\return~ \Expression}{\ProcessQ}
= \ProcessQ$.

Suppose  %now that 
$\Process = \Process_1 \parop \Process_2$ 
and the derivation of $\models \Process$ ends with the rule:  
 \begin{equation}
  \label{equation:lemmaB1:first}
    \inferrule[\refrule{wp-par}]{
       \models \ProcessP_1
      \\
       \models \ProcessP_2
    }{
       \models\ProcessP_1\parop\ProcessP_2
    }
    ~~
    \Pol (\ProcessP_1)\setminus\set{\Thing^\Polarity}\indep\CB \Pol (\Process_2)\setminus\set{\Thing^{\co\Polarity}}
    \end{equation}
 Let     
 %Suppose that 
 $\thread \varX\return~ \Expression %$ occurs once in  $
 \subprocesseq_1   \Process_1$ and $\Thing\not=\var$.
 By induction hypothesis
 $
 \models \ProcessP_1 \,\replace{\thread \varX\return~ \Expression}{\ProcessQ}
  $ since $\Pol (\ProcessP_1) \subset \Pol (\ProcessP)$ and $\Pol(\Process)\setminus\set{\varX^{-}} \indep \Pol(\ProcessQ)
 \setminus \Pol(\Expression) $ imply $\Pol(\Process_1)\setminus\set{\varX^{-}} \indep 
  \Pol(\ProcessQ) \setminus \Pol(\Expression)$.
  Now we apply \refrule{wp-par} using this new premise:
  \begin{equation}
  \label{equation:lemmaB1:second}
    \inferrule[\refrule{wp-par}]{
       \models \ProcessP_1 \,\replace{\thread \varX\return~ \Expression}{\ProcessQ}
      \\
      \models \ProcessP_2
    }{
     \models\ProcessP_1 \,\replace{\thread \varX\return~ \Expression}{\ProcessQ}\parop\ProcessP_2
    }
    ~~
     \Pol (\ProcessP_1 \,\replace{\thread \varX\return~ \Expression}{\ProcessQ})\setminus\set{\Thing^\Polarity}\indep\CB \Pol (\Process_2)\setminus\set{\Thing^{\co\Polarity}}
    \end{equation}
   We need to prove  that 
   the side condition of  (\ref{equation:lemmaB1:second}) %\cref{equation:lemmaB1:second}
%   $ \Pol (\ProcessP_1\replace{\thread \varX\return~ \Expression}{\ProcessQ})\setminus\set{\Thing^\Polarity}\indep\CB \Pol (\Process_2)\setminus\set{\Thing^{\co\Polarity}}$ 
    holds. Since 
\[\Pol (\ProcessP_1 \,\replace{\thread \varX\return~ \Expression}{\ProcessQ})\setminus\set{\Thing^\Polarity}
\subseteq 
  \Pol (\ProcessP_1)\setminus\set{\Thing^\Polarity} \cup 
\Pol (\ProcessQ) \setminus (\Pol(\Expression) \cup \set{\Thing^\Polarity} ) 
\]
it is enough to  show %prove 
that 
%$ \Pol (\ProcessP_1)\setminus\set{\Thing^\Polarity} 
%\indep\CB \Pol (\Process_2)\setminus\set{\Thing^{\co\Polarity}} $
%and 
$
\Pol (\ProcessQ) \setminus (\Pol(\Expression) \cup \set{\Thing^\Polarity} ) 
\indep\CB \Pol (\Process_2)\setminus\set{\Thing^{\co\Polarity}} 
$. 
%
%$
%\Pol (\ProcessP_1) \cup \Pol (\ProcessQ)  
%\setminus\set{\Thing^\Polarity}\indep\CB \Pol (\Process_2)\setminus\set{\Thing^{\co\Polarity}} 
%$. 
This is a consequence of 
%the side condition of \ref{equation:lemmaB1:first} 
  $ \Pol (\ProcessQ)\setminus \Pol(\Expression)  \indep\Pol (\Process_2)$, being  $\Pol(\Process)\setminus\set{\varX^{-}} \indep \Pol(\ProcessQ) \setminus \Pol(\Expression) $ 
and $\Pol (\Process_2) \subseteq \Pol (\Process)\setminus\set{\varX^{-}}$.

    \noindent
    The case $\Thing=\var$ is similar and simpler than  the
    previous one.
\end{proof}

\begin{lemma}\label{lem:uffa}
If $\models \Process\parop\ProcessQ$ and $\thread\var\Expression\subprocess_1 \Process\parop\ProcessQ$ and
$\var^+\in \Pol(\Process)$ and $\var^+\in \Pol(\ProcessQ)$, then there are $\Process'$, $\ProcessQ'$ such that $\Process'\parop\ProcessQ'\equiv\Process\parop\ProcessQ$ and $\models \Process'\parop\ProcessQ'$ and $\thread\var\Expression\subprocesseq_1 \Process'$ and $\var^+\not\in \Pol(\Process')$. 
\end{lemma}
\begin{proof} 
We assume  $\thread\var\Expression\subprocesseq_1 \Process$, the proof for $\thread\var\Expression\subprocesseq_1 \ProcessQ$ being symmetric. The derivation of $\models \Process\parop\ProcessQ$ must end by:
\begin{equation}
 \label{equation:a1}
 \inferrule[\refrule{wp-par}]{
       \models \Process
      \\
      \models \ProcessQ
    }{
       \models\Process\parop\ProcessQ
    }
    ~~
   {\Pol (\ProcessP)\setminus\set{\var^-}\indep\CB \Pol (\ProcessQ)\setminus\set{\var^+}}
    \end{equation}
The proof is by induction on the derivation of $ \models \Process$. From $\thread\var\Expression\subprocesseq_1 \Process$ and $\var^+\in \Pol(\Process)$ 
 and  $\varX^{+} \not \in \Pol(\Expression)$ 
we get $\Process\equiv\Process_1\parop\Process_2$. %We assume  
 Let  $\thread\var\Expression\subprocesseq_1 \Process_1$, then   the derivation of $\models \Process$ must end by:
\begin{equation}
 \label{equation:a2}
 \inferrule[\refrule{wp-par}]{
       \models \Process_1
      \\
      \models \Process_2
    }{
       \models\Process_1\parop\Process_2
    }
    ~~
    \Pol (\ProcessP_1)\setminus\set{\var^-}\indep\CB \Pol (\Process_2)\setminus\set{\var^+}
    \end{equation}

    If $\var^+\not\in \Pol(\Process_1)$ we can choose $\Process'=\Process_1$ and $\ProcessQ'=\Process_2\parop\ProcessQ$. In fact we can  derive:
 
    \begin{prooftree}
     \AxiomC{$\models \Process_1$}
     \AxiomC{$\models \Process_2$}
      \AxiomC{$\models \ProcessQ$}
      \LeftLabel{$\refrule{wp-par}$} 
      \RightLabel{{ 
      $\Pol (\ProcessP_2)\indep\CB \Pol (\ProcessQ)$}}
     \BinaryInfC{$\models \Process_2 \parop \ProcessQ$}
 \LeftLabel{$\refrule{wp-par}$} 
 \RightLabel{{ 
      $\Pol (\ProcessP_1)\setminus\set{\var^-}\indep\CB \Pol (\Process_2 \parop \ProcessQ)\setminus\set{\var^+}$}}
   \BinaryInfC{$\models \Process_1 \parop (\Process_2 \parop \ProcessQ)$}  
\end{prooftree}
%The sides conditions follow from 
%\ref{equation:a1}, \ref{equation:a2} and
%the facts that
%$\Pol (\ProcessP_1)$ and $\Pol (\ProcessP_2)$
%are subsets of $\Pol(\Process)$.
The first side condition follows from $\var^-\not\in\Pol (\ProcessP_2)$, $\Pol (\ProcessP_2)\subseteq\Pol(\Process)$, and the side condition of (\ref{equation:a1}). The second side condition follows from $\Pol (\ProcessP_1)\subseteq\Pol(\Process)$ and the side conditions of (\ref{equation:a1}), (\ref{equation:a2}).

  If $\var^+\in \Pol(\Process_1)$ by induction there are $\Process_1'$, $\Process_2'$ such that  $   
  \Process_1'\parop\Process_2'\equiv\Process_1\parop\Process_2$ 
   and $\models \Process_1'\parop\Process_2'$ and $\thread\var\Expression\subprocesseq_1 \Process_1'$ and $\var^+\not\in \Pol(\Process'_1)$. 
 We can choose $\Process'=\Process_1'$ and $\ProcessQ'=\Process_2'\parop\ProcessQ$. In fact we can derive:
  
 \begin{prooftree}
     \AxiomC{$\models \Process'_1$}
     \AxiomC{$\models \Process'_2$}
      \AxiomC{$\models \ProcessQ$}
      \LeftLabel{$\refrule{wp-par}$} 
      \RightLabel{{ 
      $\Pol (\ProcessP'_2)\indep\CB \Pol (\ProcessQ)$}}
     \BinaryInfC{$\models \Process'_2 \parop \ProcessQ$}
 \LeftLabel{$\refrule{wp-par}$} 
 \RightLabel{{ 
      $\Pol (\ProcessP'_1)\setminus\set{\var^-}\indep\CB \Pol (\Process'_2 \parop \ProcessQ)\setminus\set{\var^+}$}}
   \BinaryInfC{$\models \Process'_1 \parop (\Process'_2 \parop \ProcessQ)$}  
\end{prooftree}
The first side condition follows from  $\var^-\not\in\Pol (\ProcessP'_2)$, $\Pol (\ProcessP'_2)\subseteq\Pol(\Process)$, and the side condition of (\ref{equation:a1}). 
%\ref{equation:a1} and
%$\Pol (\Process'_2) \subseteq \Pol (\Process'_1 \parop \Process'_2)
%= \Pol (\Process)$. 
Observe that
$\models \Process'_1 \parop \Process'_2$  implies that
\begin{equation}
\label{equation:yetanothersidecondition}
\Pol (\ProcessP'_1)\setminus\set{\var^-}\indep\CB \Pol (\Process'_2)\setminus\set{\var^+}
\end{equation}
Then the second side condition follows from 
 $\Pol (\Process'_1) \subseteq  
\Pol (\Process)$, (\ref{equation:yetanothersidecondition}),   
and the side condition of (\ref{equation:a1}).   
        \end{proof}

\begin{lemma}\label{lem:ret}
Let $\CA\models\Process$.  
%and 
% all threads of $\Process$ have different names.
\begin{enumerate}
\item\label{lem:ret3} 
If  
$\var^{-} \not \in \Pol (\Process)$ and 
$\Pol(\Expression) \indep   \Pol(\Expression)$ and 
$\Pol (\Process) \indep \Pol(\Expression)$,
then
$\models\Process\subst\Expression\var$.

\item\label{lem:ret1} If $\thread\var{\return~\Expression} \subprocesseq_1  \Process$,  
then $\models\ProcessQ\subst\Expression\var$ for some $\ProcessQ$ such that $\ProcessQ\parop\thread\var{\return~\Expression}\equiv\Process$.
\end{enumerate}
\end{lemma}
\begin{proof}  Both items 
%\cref{lem:ret3} and \cref{lem:ret1} 
are
 proved by induction on the derivation of $\CA\models\Process$.

 (\cref{lem:ret3}). We show only the case of \refrule{wp-thread}. Suppose that
 \[
 \inferrule[\refrule{wp-thread}]{
    }{
      \models \thread\varY \ExpressionF
    }
    ~~
    \begin{lines}[c]
     \varY^+\not\in\Pol(\ExpressionF)\\
          \Pol(\ExpressionF) \indep   \Pol(\ExpressionF)
    \end{lines} 
 \]
We can do the following inference:
\[
 \inferrule[\refrule{wp-thread}]{
    }{
      \models \thread\varY \ExpressionF\subst\Expression\var
    }
    ~~
    \begin{lines}[c]
      \varY^+\not\in\Pol(\ExpressionF\subst\Expression\var)\\
          \Pol(\ExpressionF\subst\Expression\var) \indep   \Pol(\ExpressionF\subst\Expression\var)
    \end{lines} 
 \]
 The  side condition  %premise  
 $\varY^+\not\in\Pol(\ExpressionF\subst\Expression\var)$
holds because 
 $\varY^+\not\in\Pol(\ExpressionF)$ and  $\Pol (\Process)=
\Pol(\ExpressionF)\cup\set{\varY^-}$ and  $\Pol(\ExpressionF)\cup\set{\varY^-} \indep \Pol(\Expression)$. 
The side condition $\Pol(\ExpressionF\subst\Expression\var) \indep   \Pol(\ExpressionF\subst\Expression\var)$
holds because $\Pol (\Process) \indep   \Pol(\Expression)$
and
$\Pol(\Expression) \indep   \Pol(\Expression)$.

(\cref{lem:ret1}).
%If  $\Process=\thread\var{\return~\Expression}$ we can choose $\ProcessQ=\emptyprocess$ by rule $\defrule{wp-empty}$.
%
%
%If  $\thread\var{\return~\Expression}\subprocess\Process$ and 
%$\var^+\not\in \Pol (\ProcessP)$ then
% $\Process \equiv \thread\var{\return~\Expression} \parop \ProcessQ$
% and $\ProcessQ = \ProcessQ\subst\Expression\var$.
%    We get  $\models \ProcessQ$ by \cref{lem:Luca3}. 
 If  $\var^+\not\in \Pol (\ProcessP)$ we can choose $\ProcessQ=\Process \, \replace{\thread \varX {\return~\Expression}}{\emptyprocess}$ by rule $\defrule{wp-empty}$ and 
\cref{lemma:replacementofthread}.
Otherwise 
suppose $\Process=\Process_1\parop\Process_2$ and the last rule of the derivation is: 
 \begin{equation}
 \label{equation:wparruleitem2}
 \inferrule[\refrule{wp-par}]{
       \models \Process_1
      \\
      \models \Process_2
    }{
       \models\Process_1\parop\Process_2
    }
    ~~
    \Pol (\ProcessP_1)\setminus\set{\Thing^\Polarity}\indep\CB \Pol (\Process_2)\setminus\set{\Thing^{\co\Polarity}}
    \end{equation}

We can assume  $\thread\var{\return~\Expression} \subprocesseq_1 \Process_1$ since the case 
  $\thread\var{\return~\Expression} \subprocesseq_1 \Process_2$  
is symmetric.
We distinguish  three cases:

\begin{enumerate}

\item 
Case  $\var^+\not\in \Pol (\ProcessP_1)$. 
 %and $\var^+\in \Pol (\ProcessP_2)$.  
The key observation is that $\Thing^\Polarity=\var^-$ and $\Thing^{\co\Polarity}=\var^+$. 
 % It follows 
  From 
  the side condition of (\ref{equation:wparruleitem2}) and 
   $\Pol (\ProcessP_1)\setminus\set{\var^{-}} \supseteq \Pol(\Expression)$ and %also from 
   $\var^{+} \not \in \Pol(\Expression)$   we have  %that
   $\Pol (\ProcessP_2) \indep \Pol(\Expression)$.  Since $\models{\thread\var{\return~\Expression}}$ implies $\Pol (\Expression) \indep \Pol(\Expression)$, 
 \cref{lem:ret3} gives 
  $  \models\Process_2\subst\Expression\var$. 
 %Since 
  From this and  $\Pol (\ProcessP_1)\setminus\set{\var^{-}}\indep
 \Pol(\Process_2\subst\Expression\var) \setminus \Pol (\Expression)$, 
  %and $\Pol(\Expression)\subseteq\Pol(\Process_2\subst\Expression\var)$, 
it follows  by  %from 
\cref{lemma:replacementofthread}
that
\[
\models \Process_1 \,\replace{\thread \varX {\return~\Expression}}{\Process_2\subst\Expression\var}
\]
We  can choose  
$\ProcessQ=\Process_1 \, \replace{\thread \varX {\return~\Expression}}{\Process_2}$, since
it is not difficult to check
that
\[\thread \varX {\return~\Expression} \parop 
\Process_1 \, \replace{\thread \varX {\return~\Expression}}{\Process_2}
\equiv \Process_1 \parop \Process_2
\]
\item 
Case  $\var^+\in \Pol (\ProcessP_1)$ and 
$\var^+\in \Pol (\ProcessP_2)$. By \cref{lem:uffa} 
 there are $\Process'_1$ and $\Process'_2$ such that
 $\Process_1 \parop \Process_2 \equiv 
\Process'_1 \parop \Process'_2$
and
 $\thread\var{\return~\Expression} \subprocesseq_1 \Process'_1$ 
 and 
 $\var^+ \not \in \Pol (\ProcessP'_1)$.  
% and 
%$\var^+\in \Pol (\ProcessP'_2)$.
We can now proceed as in the previous case.
Note that this case and the previous one 
are sort of ``base cases'' for which the induction hypothesis
is not needed.

\item 
Case $\var^+\not\in \Pol (\ProcessP_2)$. 
By induction hypothesis
 $\models\ProcessQ_1\subst\Expression\var$ for some $\ProcessQ_1$ such that $\ProcessQ_1\parop\thread\var{\return~\Expression}\equiv\Process_1$. 
We can apply rule $\refrule{wp-par}$ to $\models\ProcessQ_1\subst\Expression\var$ and $\models \Process_2$ since  $\Pol(\ProcessQ_1\subst\Expression\var)=\Pol (\ProcessP_1)\setminus\set{\var^{+}, \var^-}$. So we conclude  $\models\ProcessQ_1\subst\Expression\var\parop\Process_2$.
\qedhere
\end{enumerate}
\end{proof}
  
Since the definition of $\models$ is not invariant under $\equiv$,
we cannot prove that the reduction preserves well-polarisation 
by induction on $\red$. Instead, we use the following lemma, which immediately follows from the definition of $\red$:

\begin{lemma}[Inversion of $\red$]
\label{lemma:inversionlemmaforreduction}
If $\Process \red \Process'$, then
$\Process \equiv  \new{\Thing_1 \ldots \Thing_n} \Process_0$
and 
$\Process' \equiv \new{\Thing_1 \ldots \Thing_n} \Process'_0$
and 
one of the following  cases  %possibilities 
hold:

\begin{enumerate}

\item 
%$\thread\varX{\PContext[\open~\Channel]}\subprocesseq \threads(\Process)$ and $\server\Channel\Expression \subprocesseq \servers(\Process)$ and 
%$\set{\ChannelC,\varY}\subseteq\boundnames(\Process')$ and $\thread\varX{\PContext[\Return{\ChannelC^+}]}
%      \parop
%      \thread\varY{\Expression~\ChannelC^-}\subprocesseq \threads(\Process')$ and $\server\Channel\Expression \subprocesseq \servers(\Process')$
      
$\Process_0 = 
 \server\Channel\Expression \parop \thread\varX{\PContext[\open~\Channel]}\parop \ProcessQ$
 and\\
 $\Process'_0 = 
 \server\Channel\Expression \parop
 \new{\ChannelC\varY}(
      \thread\varX{\PContext[\Return{\ChannelC^+}]}
      \parop
      \thread\varY{\Expression~\ChannelC^-})
      \parop  \ProcessQ$.
      
\item $\Process_0=  
 \thread\varX{\PContext[\send~\Channel^\Polarity~\Expression]}
      \parop
      \thread\varY{\PContext'[\receive~\Channel^{\co\Polarity}]}
  \parop \ProcessQ$
 and\\ 
 $\Process'_0 = 
 \thread\varX{\PContext[\Return{\Channel^\Polarity}]}
      \parop
      \thread\varY{\PContext'[\Return{\Pair\Expression{\Channel^{\co\Polarity}}}]} \parop \ProcessQ$.

\item $\Process_0 =
 \thread\varX{\PContext[\future~\Expression]} \parop \ProcessQ$ and
 $\Process'_0 = 
      \new\varY(\thread\varX{\PContext[\Return\varY]} \parop \thread\varY\Expression)\parop \ProcessQ$.

\item $\Process_0 =\new\var(
 \thread\var{\Return\Expression} \parop \ProcessQ)$
 and 
 $\Process'_0 = 
 \ProcessQ \subst{\Expression}{\var}$.

\item 
$\Process_0 =
 \thread \var \Expression \parop \ProcessQ$
 and 
 $\Process'_0 = 
 \thread\var \Expression' \parop \ProcessQ$ with $\Expression \red \Expression'$.

\end{enumerate}

\end{lemma}

\begin{proof}[Proof of \cref{theorem:invariance}]
 Well-polarisation of 
$\Process$ implies 
  that 
\[
\Process \equiv
\new{\Thing_1\ldots\Thing_n}
(\ProcessQ\parop\ProcessR)
\] 
where   $\set{\Thing_1,\ldots,\Thing_n}=\boundnames(\ProcessP)$, $\ProcessQ \equiv\threads (\ProcessP)$ and 
 $\CA \models \ProcessQ$ 
and $\ProcessR=\servers(\ProcessP)$. 
 Using  \cref{lemma:inversionlemmaforreduction}, we analyse cases
 according to the shapes of $\Process$,  $\ProcessQ$ and $\ProcessR$.
We only show the interesting cases.
 \begin{enumerate}
 
 \item Case $\thread \var {\PContext [\open~\Channel]} \subprocesseq \ProcessQ$
 and $\server~\Channel~\Expression \subprocesseq  \ProcessR$. Hence,
 \[
\Process' \equiv
\new{\Thing_1\ldots\Thing_n \ChannelC y}
(\ProcessQ \,\replace{\thread \var {\PContext [\open~\Channel]}}{\thread \var {\PContext [\return~{\ChannelC^{+}]}}} \parop \thread\varY{\Expression~\ChannelC^-}\parop\ProcessR)
\] 
 It is easy to show that
  \[
  \models \ProcessQ \, \replace{\thread \var {\PContext [\open~\Channel]}}{\thread \var {\PContext [\return~{\ChannelC^{+}]}}}
  \]
  Since $\Process$ is typeable,  $\Pol(\Expression) = \emptyset$
  and  
 \[
 \models \thread \varY \Expression \ChannelC^{-}
 \]
 Using \refrule{wp-par}, we 
 obtain that
 \[
 \models \ProcessQ  \,\replace{\thread \var {\PContext [\open~\Channel]}}{\thread \var {\PContext [\return~{\ChannelC^{+}]}}} \parop \thread \varY \Expression \ChannelC^{-}
  \]
  Hence, $\Process'$ is well-polarised.
  
\item Case $\ProcessQ \equiv  
 \thread\varX{\PContext[\send~\Channel^\Polarity~\Expression]}
      \parop
      \thread\varY{\PContext'[\receive~\Channel^{\co\Polarity}]}
  \parop \ProcessQ_0$.
  Then,
  \[
\Process' \equiv
\new{\Thing_1\ldots\Thing_n}
(\ProcessQ'\parop\ProcessR)
\] 
  where 
 $\ProcessQ' = 
 \thread\varX{\PContext[\Return{\Channel^\Polarity}]}
      \parop
      \thread\varY{\PContext'[\Return{\Pair\Expression{\Channel^{\co\Polarity}}}]} \parop \ProcessQ_0$.  
     Typeability of $\Process$ implies that $\Channel^\Polarity$       %and  $ \thread\varX{\PContext[\send~\Channel^\Polarity~\Expression]}
%     $ and $
%      \thread\varY{\PContext'[\receive~\Channel^{\co\Polarity}]}$
 %     occur only once in $\Process$.   
occurs only once  
and  
that the above threads are the unique ones 
%the those are the only threads
 named $\varX$ and $\varY$ in $\Process$.    
By Lemma \ref{lemma:invariancecomm}\comma
there exists $\ProcessQ''$ such that
$\ProcessQ'' \equiv \ProcessQ'$ and
$\CA \models \ProcessQ''$.
Then $\Process'$ is well-polarised.

\item Case $\Thing_n=\var$ and $\ProcessQ \equiv 
\thread \var {\Return\Expression} \parop \ProcessQ_0$.
 Then\comma
  \[
\Process' \equiv
\new{\Thing_1\ldots\Thing_{n-1}}
(\ProcessQ_0\subst{\Expression}{\var}  \parop\ProcessR)
\] 
%If $\var$ does not occur in $\ProcessQ_0$ 
% it follows from \cref{lem:Luca3}  that $\models \ProcessQ_0$.
% Hence $\Process' $ is well-polarised.
% If $\var$ occurs in $\ProcessQ_0$, 
% it follows from  \cref{lem:ret1} of  \cref{lem:ret}
% that 
% $ \models \ProcessQ'_0 \subst{\varX}{\Expression}$
% for some $\ProcessQ'_0 \equiv \ProcessQ_0$ 
% and hence in this case $\Process' \subst{\varX}{\Expression}$
% is well-polarised too.
 Typeability of $\Process$ implies %that
%$\thread \var {\Return\Expression}\subprocesseq_1 $ 
%occurs only once in $
%\ProcessQ$.  
 that the above  thread is the only one named $\varX$ in $\Process$. 
It follows from  \cref{lem:ret1} of  \cref{lem:ret}
 that 
 $ \models \ProcessQ'_0 \subst{\Expression}{\varX}$
 for some $\ProcessQ'_0 \equiv \ProcessQ_0$ 
 and hence $\Process'$
 is well-polarised.
\qedhere
\end{enumerate}
\end{proof}

\section{Proof of \cref{theorem:acyclicimpliesSR}}\label{appendix:SRP}

\begin{proof}[Proof of~\cref{theorem:acyclicimpliesSR}]
 The proof is by induction on the definition of
  $\red$. We only show the most interesting cases.   
  
\mycase{ Case 
  $ \new\varX(
        \thread\varX{\return~\Expression}
        \parop
        \Process ) \red \Process\subst\Expression\varX $}       
        
  \noindent      Let  $\Process_1$ and $\Process_2$ be such that
        $\Process \equiv \Process_1 \parop \Process_2$  and  $\Process_1$ contains
        all and only  the  %those 
        threads in  whose bodies the variable $\varX$ occur.        
    It follows from 
    $\wtp{\TypeContext}{\new\varX( \thread\varX{\return~\Expression}
      \parop
      \Process )}{\ProvideContext}$ and the 
      Inversion Lemmas for Processes and Expressions 
     (Items \ref{lem:inv7}, \ref{lem:inv9} and  \ref{lem:inv11}  of \cref{lem:invp}
    and  \ri{\cref{lem:inv}}{Items \ref{lem:inv1} and  \ref{lem:inv5}}) 
     that 
    \[
    \begin{array}{lll}
    \wte{\TypeContext_0}{\ExpressionE}{\Type}\ \ \ \ & 
     \wtp{\TypeContext_1, \varX:\Type}{\Process_1}{\ProvideContext_1} \ \ \ \  & 
      \wtp{\TypeContext_2}{\Process_2}{\ProvideContext_2}
    \end{array}
    \]
    where $\TypeContext = \TypeContext_0 + \TypeContext_1 + \TypeContext_2 $
    and $\ProvideContext =\ProvideContext_1 +\ProvideContext_2$.
         Since  $\new\varX(
        \thread\varX{\return~\Expression}
        \parop
        \Process )$ 
        is well-polarised,  
        if $\thread \varY \ExpressionF$ is in $\Process_1$
         (i.e. $\varX$ occurs in $\ExpressionF$),   then
        $\varY$ cannot 
        occur in  $\Expression$  
        by \ri{\cref{lem:Luca}}{\cref{lem:Luca4}}.
        Hence, 
          $\dom(\TypeContext_1) \cap \dom (\ProvideContext_1) = \emptyset$. 
          %We split the process $\Process$ into two because we can apply  \cref{lemma:substitutionp} only to $\Process_1$
     Then     we can apply  \cref{lemma:substitutionp} to $\Process_1$ and   obtain  \[\wtp{\TypeContext_0 + \TypeContext_1}{\Process_1\subst\Expression\varX}{\ProvideContext_1}\]
%        We cannot apply \cref{lemma:substitutionp} to $\Process_2$ because
%        it may contain a thread $\thread \varY \ExpressionF$ such that $\varY$ occurs in $\Expression$.
         By %applying 
        rule \defrule{par} we derive  %obtain       
    \[
    \wtp{\TypeContext_0 + \TypeContext_1 + \TypeContext_2}
    {\Process_1\subst\Expression\varX \parop \Process_2}{\ProvideContext}\]
    
 \mycaseA{   Case
       $ \thread\varX{\PContext_1[\send~\Channel^\Polarity~\Expression]}
        \parop
        \thread\varY{\PContext_2[\receive~\Channel^{\co\Polarity}]}
        \red
        \thread\varX{\PContext_1[\return~\Channel^\Polarity]}
        \parop
        \thread\varY{\PContext_2 [\return~\Pair\Expression{\Channel^{\co\Polarity}}]}
      $}

\noindent
 It follows from
    $\wtp{\TypeContext}{
      \thread\varX{\PContext_1[\send~\Channel^\Polarity~\Expression]}
      \parop
      \thread\varY{\PContext_2[\receive~\Channel^{\co\Polarity}]
      }}{\ProvideContext}$
    and the Inversion Lemmas for Processes %and Expressions
    (Items \ref{lem:inv7} and \ref{lem:inv9}
    of \cref{lem:invp}) 
%    and 
%     Items  \ref{lem:inv1},  \ref{lem:inv2} and \ref{lem:inv5}
%      of \cref{lem:inv})
     that
    $\TypeContext = \TypeContext_1 +\TypeContext_2 $ 
     and $\ProvideContext =
\varX: \tbullet[n_1] \Type_1, \varY: \tbullet[n_2]\Type_2$ and 
      \begin{equation}
        \label{eq:sendtwo}
        \wte{\TypeContext_1}{\PContext_1[\send~\Channel^\Polarity~\Expression]}{\tbullet^{n_1}(\tio\Type_1)}
      \end{equation}
      %\break
    \begin{equation}
      \label{eq:receivetwo}
      \wte{\TypeContext_2}{\PContext_2[\receive~\Channel^{\co\Polarity}]}{\tbullet^{n_2}(\tio\Type_2)}
    \end{equation}
Using the fact that $\TypeContext$ is balanced, it is not difficult to show that 
\[
\begin{array}{ll}
\Channel^{\Polarity} : \tbullet[m] (\tout{\Type}\SessionType)
\in  \TypeContext_1 \\
\Channel^{\co \Polarity} : \tbullet[m] (\tin{\Type} \co \SessionType)
\in  \TypeContext_2 
\end{array}
\]
%where 
 for some $m$ such that 
$m \leq n_1$ and $m \leq n_2$ by \cref{lemma:delayreceive}. 
    By applying  {\cref{lem:key}} to (\ref{eq:sendtwo}),
    we have 
    $\TypeContext_1 = \TypeContext_3 + \TypeContext_4,\Channel^{\Polarity} : \tbullet[m] (\tout{\Type}\SessionType)$ with
    \begin{equation}
      \label{eq:sendthree}
      \begin{array}{ll}
       \wte{\TypeContext_3, \varZ: \tbullet[n_1] \tio \SessionType }{\PContext_1[\varZ]}{\tbullet^{n_1}(\tio\Type_1)}
       \\
       \wte{\TypeContext_4,\Channel^{\Polarity} : \tbullet[m] (\tout{\Type}\SessionType)}
      {\send~\Channel^\Polarity~\Expression}{\tbullet[n_1] (\tio \SessionType)}
       \end{array}
    \end{equation}
 Items  \ref{lem:inv1},  \ref{lem:inv2} and \ref{lem:inv5} of \cref{lem:inv} give
 \begin{equation}
      \label{eq:te}
      \wte{\TypeContext_4}
      {\Expression}{\tbullet[n_1] \Type}
       \end{equation}
    Using rules \refrule\const, \refrule\axiom, \refrule\introbullet, \refrule\elimarrow\ being $m \leq n_1$ we
    derive
    \begin{equation}
    \label{eq:sendfour}
    \wte{\Channel^\Polarity: \tbullet[m] \SessionType}{\return~\Channel^\Polarity}
    {\tbullet[n_1](\tio\SessionType)}
    \end{equation} 
    By applying ~\cref{lemma:substitution}
    to  \eqref{eq:sendthree} and \eqref{eq:sendfour} we get
    \[\wte{\TypeContext_3,\Channel^\Polarity: \tbullet[m]\SessionType}{\PContext_1[\return~\Channel^\Polarity
      ]}{\tbullet^{n_1} (\tio\Type_1)}\]
    hence by \refrule{thread} we derive
    \begin{equation}
      \label{eq:sendreturnthree}
      \wtp{\TypeContext_3,\Channel^\Polarity:\tbullet[m]\SessionType}{\thread\varX{\PContext_1[\return~\Channel^\Polarity]}}{\varX:\tbullet[n_1]\Type_1}
    \end{equation} 
      By applying 
  {\cref{lem:key}} to \eqref{eq:receivetwo} %and also using Inversion Lemma,
    \begin{equation}
    \label{eq:receivealone}  
    \begin{array}{ll}
   \wte{\TypeContext_5, \varZ:\tbullet[n_2]\tio(\Type\times\co\SessionType) }
   {\PContext_2[\varZ]}{\tbullet[n_2](\tio \Type_2)} \\
    \wte{\Channel^{\co\Polarity}: 
    \tbullet[m] \tin{\Type}\co\SessionType}{\receive~\Channel^{\co\Polarity}}{\tbullet[n_2]\tio(\Type\times\co\SessionType)}
    \end{array}
    \end{equation}  
   for  $\TypeContext_2=\TypeContext_5,
  \Channel^{\co\Polarity}: \tbullet[m] \tin{\Type}\co\SessionType$. %and $m \leq n_2$. 
   From (\ref{eq:te}) and $m \leq n_2$  using rules \refrule\const, \refrule\axiom, \refrule\introbullet, \refrule\elimarrow\ we derive 
    \begin{equation}
      \label{eq:receivereturn}
      \wte{\TypeContext_4, \Channel^{\co\Polarity}: \tbullet[m] \co\SessionType}{\return~\Pair\Expression{\Channel^{\co\Polarity}}}{\tbullet[n_2] \tio(\Type\times\co\SessionType)}
    \end{equation}
    Applying~\cref{lemma:substitution}
    to
    (\ref{eq:receivealone})   and (\ref{eq:receivereturn}) it follows 
    that   
    \begin{equation}
      \label{eq:receivereturntwo}
      \wte{\TypeContext_4+ \TypeContext_5, 
      \Channel^{\co\Polarity}: \tbullet[m] \co\SessionType}{\PContext_2[\return~\Pair\Expression{\Channel^{\co\Polarity}}]}{\tbullet^{n_2} (\tio\Type_2)}
    \end{equation}
  From well-polarisation and \ri{\cref{lem:Luca}}{\cref{lem:Luca5}},
   $\varY$ cannot occur in $\Expression$.
        Then we  can 
    apply rule \refrule{thread}
    to (\ref{eq:receivereturntwo}) deriving
    \begin{equation}
      \label{eq:receivereturnthree}
      \wtp{\TypeContext_4+ \TypeContext_5, \Channel^{\co\Polarity}: \tbullet[m] \co\SessionType}{\thread\varY{\PContext_2[\return~\Pair\Expression{\Channel^{\co\Polarity}}]}}{\varY:\tbullet[n_2]\Type_2} 
    \end{equation}
    By applying rule \refrule{par} to \eqref{eq:sendreturnthree} and
    \eqref{eq:receivereturnthree} we conclude
    \[
    \wtp{\TypeContext_3+\TypeContext_4+\TypeContext_5,\Channel^\Polarity: \tbullet[m] \SessionType, \Channel^{\co\Polarity}: \tbullet[m] \co\SessionType}{ 
      \thread\varX{\PContext_1[\return~\Channel^\Polarity]}
      \parop
      \thread\varY{\PContext_2[\return~\Pair\Expression{\Channel^{\co\Polarity}}]}
    }{\ProvideContext}\]
     where
%    $\TypeContext 
%      \redctx  \TypeContext_3 + \TypeContext_4 + \TypeContext_5,\Channel^\Polarity: \tbullet[m] \SessionType, 
%      \Channel^{\co\Polarity}: \tbullet[m] \co\SessionType $.  
 $ \TypeContext_3 + \TypeContext_4 + \TypeContext_5,\Channel^\Polarity: \tbullet[m] \SessionType, 
      \Channel^{\co\Polarity}: \tbullet[m] \co\SessionType $ is balanced.
     \end{proof}

\end{document}